\documentclass[pdflatex,iicol,sn-mathphys-ay]{sn-jnl}% Math and Physical Sciences Author Year Reference Style
\usepackage{graphicx}%
\usepackage{multirow}%
\usepackage{amsmath,amssymb,amsfonts}%
\usepackage{amsthm}%
\usepackage{mathrsfs}%
\usepackage[title]{appendix}%
\usepackage{xcolor}%
\usepackage{textcomp}%
\usepackage{manyfoot}%
\usepackage{booktabs}%
\usepackage{algorithm}%
\usepackage{algorithmicx}%
\usepackage{algpseudocode}%
\usepackage{listings}%
\usepackage{tabularx}
\usepackage{makecell} 
\usepackage{epstopdf}
\theoremstyle{thmstyleone}%
\newtheorem{thm}{Theorem}%  meant for continuous numbers
\newtheorem{prop}{Proposition}%
\newtheorem{lemma}{Lemma}%

\theoremstyle{thmstyletwo}%
\newtheorem{example}{Example}%

\theoremstyle{thmstylethree}%
\newcommand{\indep}{\rotatebox[origin=c]{90}{$\models$}}
\newcommand{\setword}[2]{%
	\phantomsection
	#1\def\@currentlabel{\unexpanded{#1}}\label{#2}%
}
\begin{document}

\title[Sufficient  Dimension Reduction]{Sufficient  Dimension Reduction via  Inverse Conditional Mean or Variance Independence}


\author[1]{\fnm{Jicai} \sur{Liu}}%\email{liujicai1234@126.com}
\author[2]{\fnm{Yu} \sur{Zhang}}%\email{iiauthor@gmail.com}

%\author[3]{\fnm{Heng} \sur{Lian}}%\email{iiiauthor@gmail.com}

\author*[3]{\fnm{Jinhong} \sur{Li}}\email{jinhongli0106@gmail.com}

\affil[1]{\orgdiv{School of Statistics and Mathematics}, \orgname{Shanghai Lixin University of Accounting and Finance}, \orgaddress{\state{Shanghai}, \country{China}}}

\affil[2]{\orgdiv{ Department of Probability and Statistics}, \orgname{ Peking University}, \orgaddress{\state{Beijing}, \country{China}}}

%\affil[3]{\orgdiv{Department of Mathematics}, \orgname{City University of Hong Kong}, \orgaddress{\state{Hong Kong}, \country{China}}}

\affil*[3]{\orgdiv{School of Statistics and Data Science}, \orgname{Zhejiang Gongshang University}, \orgaddress{ \state{Zhejiang}, \country{China}}}


\abstract{
This paper presents a unified framework for sufficient dimension reduction (SDR) that  generalizes several existing SDR techniques and offers  new insights into the connection between inverse conditional moment independence and dimension reduction. The framework is built on two forms of inverse independence between  the response vector and predictors: inverse conditional mean independence (ICMI) and inverse conditional variance independence (ICVI). For each form, we develop  two general classes of matrices capable of recovering the central subspace,  based on projection and kernel techniques respectively. This yields four distinct estimators: projection- and kernel-based variants under both ICMI and ICVI frameworks.
Under standard regularity conditions, we establish the theoretical properties of these estimators and derive their convergence rates in high-dimensional settings. The proposed methods exhibit robustness to outliers in the response variable while maintaining computational competitiveness.   Simulation studies and real-data analyses demonstrate the practical effectiveness of the proposed methods.

}




\maketitle

\section{Introduction}

Sufficient dimension reduction (SDR)  has attracted considerable attention  as an effective  tool for high-dimensional data analysis.
%For example, see \citet{Likc91, li1992principal,CookOn, Cook1998regression,cook2002dimension, Xia2002adaptive}.
Its primary goal is to reduce the dimensionality of the predictors
without losing regression information. Let $Y$ be the response variable and $\mathbf{X} = (X_1, \ldots, X_p)^T$ be the predictor vector. This reduction is achieved by projecting $\mathbf{X}$ onto a lower-dimensional subspace, say $\boldsymbol{\Gamma}^T \mathbf{X}$, such that
\begin{eqnarray}\label{intr0}
	Y \indep {\bf X}  | \mathbf{\Gamma}^T{\bf X},
\end{eqnarray}
where   $``\indep"$ indicates conditional independence and $\mathbf{\Gamma} \in \mathbb{R}^{p \times d}$  with $\text{rank}(\mathbf{\Gamma})=d <p$. This implies that $Y$ depends on $\mathbf{X}$ only through $\boldsymbol{\Gamma}^T \mathbf{X}$. The column space of $\boldsymbol{\Gamma}$ is referred to as a dimension reduction subspace. Under mild conditions, \citet{YinSuccessive} showed that the intersection of all such subspaces remains a dimension reduction subspace and is called the central subspace, denoted by $\mathcal{S}_{Y \mid \mathbf{X}}$. The dimension $d = \dim(\mathcal{S}_{Y \mid \mathbf{X}})$ is known as the structural dimension. This work focuses on estimating     $\mathcal{S}_{\textbf{Y}| \textbf{X}}$ under the assumption that  $d$ is known; readers interested in estimating $d$ are referred to Chapters 9-10 of \citet{li2018sufficient}.

Many methods have been proposed for estimating $\mathcal{S}_{Y|\mathbf{X}}$, such as sliced inverse regression (SIR; \citet{Likc91}), sliced average variance estimation (SAVE; \citet{cook91weisberg}), and principal Hessian directions (pHd; \citet{li1992principal}), among
others.  Further developments include minimum average variance estimation (MAVE; \citet{Xia2002adaptive}), directional regression (DR; \citet{Bing07}), principal fitted components (PFC; \citet{CookFisher07}), and semiparametric approaches \citep{Ma2012semiparametric}. For a comprehensive review, see \citet{li2018sufficient}.

In this paper, we introduce a new class of  SDR methods based on   inverse  conditional mean  independence to recover   $\mathcal S_{Y|{\bf X}}$  in  \eqref{intr0}.    Our approach is motivated by  the principal fitted components (PFC) framework   \citep{CookFisher07,PFC08} and  focuses    on
quantifying conditional mean  (in)dependence between ${\bf X}$ and $Y$.  We now provide   a brief description of the PFC model.  Let $\mathbf{X}_{y} $ denote the $p$-dimensional random vector
distributed as $\mathbf{X} |(Y=y).$  Assume that $\mathbf{X} |(Y=y) \sim N\left(\boldsymbol{\mu}_{y}, \boldsymbol{\Delta}\right)$ with  $\boldsymbol{\Delta}>0$, i.e.,
\begin{equation}\label{inverseregression}
	\mathbf{X}_{y}  =\boldsymbol{\mu}_{y}+\boldsymbol{\Delta}^{1/2}\boldsymbol{\varepsilon}, ~~~\boldsymbol{\varepsilon} \sim N\left(\textbf{0}, \mathbf{I}_{p}\right).
\end{equation}
Setting   $\boldsymbol{\mu}_{y}=\boldsymbol{\mu}+\mathbf{B} \boldsymbol{v}_{Y}$ and $\boldsymbol{\Delta}=\sigma^{2} \mathbf{I}_{p}$, the model simplifies to the principal components model
\begin{equation}\label{PC}
	\mathbf{X}_{y} =\boldsymbol{\mu}+\mathbf{B} \boldsymbol{v}_{Y}+\sigma\boldsymbol{\varepsilon},
\end{equation}
where   $ \mathbf{B} \in \mathbb{R}^{p \times d}$ and $\boldsymbol{v}_{Y} \in \mathbb{R}^{d}$ is unknown   with  $\operatorname{E}\left(\boldsymbol{v}_{Y}\right)=0$.
Further, by assuming     $\boldsymbol{\mu}_{y}=\boldsymbol{\mu}+\mathbf{B}\boldsymbol{\beta} \mathbf{f}(Y),$  we obtain the PFC model
\begin{equation}\label{PFC}
	\mathbf{X}_{y} =\boldsymbol{\mu}+\mathbf{B}\boldsymbol{\beta} \mathbf{f}(Y)+\sigma\boldsymbol{\varepsilon},
\end{equation}
where   $\boldsymbol{\beta} \in \mathbb{R}^{d \times r}$ with $d \leq \min(r, p)$, and $\mathbf{f}(Y) \in \mathbb{R}^r$ is a known, user-specified vector-valued function of the response.  Under the assumption that $\boldsymbol{\varepsilon}$ follows a multivariate normal distribution, \citet{CookFisher07} showed that $\mathcal{S}_{Y|\mathbf{X}} = \mathcal{S}(\boldsymbol{\Delta}^{-1} \mathbf{B})$, where $\mathcal{S}(\mathbf{A})$ denotes the column space of a matrix $\mathbf{A}$.

%
% The PFC model can be expressed as the inverse regression model:
%\begin{equation}\label{PFC}
%\mathbf{X}=\boldsymbol{\mu}+\mathbf{B}\boldsymbol{\beta} \mathbf{f}(Y)
%+\sigma\boldsymbol{\varepsilon},
%\end{equation}
%where  $\boldsymbol{\beta} \in \mathbb{R}^{d \times r}$ with $d \leq \min (r, p),$    $\mathbf{f} \in \mathbb{R}^{r}$ is a known user-specified vector-valued function of the response and  the error
%$\boldsymbol{\varepsilon}$ is independent of $Y$ and normally distributed with mean  zero and covariance matrix $\boldsymbol{\Delta}.$    \citet{CookFisher07} has proved that
%$\mathcal{S}_{Y | \mathbf{X}}=\operatorname{span}\left(\Delta^{-1} \mathbf{B}\right)$ under the
%assumption  that $\boldsymbol{\varepsilon}$  follows from a multivariate normal distribution.
%Recently,  there is an increasing interest in  the inverse model-based SDR approaches, which  directly specify a model for the inverse regression  of $\mathbf{X}$ on  $Y$, like the PFC model in \eqref{PFC}.  See, \citet{AdragniSufficient} for  a review.

Let $\left(\mathbf{B}, \mathbf{B}_{\perp} \right) \in \mathbb{R}^{p \times p}$ be  an orthogonal matrix.
Linearly transforming both sides of model  \eqref{PFC} and taking conditional  expectation,
we have  that
%\begin{eqnarray*}
%&&\mathbf{B}^{T}\mathbf{X}_{y} =\mathbf{B}^{T}\boldsymbol{\mu}+ \beta\left\{\mathbf{f}_{y}-\mathrm{E}\left(\mathbf{f}_{Y}\right)\right\}
%+\sigma\mathbf{B}^{T}\boldsymbol{\varepsilon}
% ~~\mathrm{and} ~~ \mathbf{B}_{\perp}^{T} \mathbf{X}_{y}=\mathbf{B}_{\perp}^{T} \boldsymbol{\mu}+\sigma\mathbf{B}_{\perp}^{T}\boldsymbol{\varepsilon},
%\end{eqnarray*}
%which implies that
\begin{equation}\label{eq:cond-moments}
\begin{array}{rcl}
E\{\mathbf{B}^{T}\mathbf{X}|Y\}
  &=& \mathbf{B}^{T}\boldsymbol{\mu} + \beta \mathbf{f}(Y), \\[2mm]
E\{\mathbf{B}_{\perp}^{T}\mathbf{X}|Y\}
  &=& \mathbf{B}_{\perp}^{T}\boldsymbol{\mu}
   = E\{\mathbf{B}_{\perp}^{T}\mathbf{X}\}.
\end{array}
\end{equation}
 The above expressions indicate  that   the transformed vector $[\mathbf{B}^{T}\mathbf{X}, \mathbf{B}_{\perp}^{T} \mathbf{X}]$ splits into two parts:   $\mathbf{B}^{T}\mathbf{X}$
 is conditionally mean dependent on $Y$, while $\mathbf{B}_{\perp}^{T}\mathbf{X}$
is conditionally mean independent of  $Y$. Thus, estimating $\mathbf{B}$ is equivalent to
estimate its  orthogonal completion  $\mathbf{B}_{\perp},$  such that
\begin{eqnarray}\label{PFC2}
	E\{ \mathbf{B}_{\perp}^{T} \mathbf{X}|Y \}=E\{ \mathbf{B}_{\perp}^{T} \mathbf{X}\}.
\end{eqnarray}
That is,    the conditional mean of  $ \mathbf{B}_{\perp}^{T} \mathbf{X}$ given  $Y$ is independent of  $Y$.

Motivated by this notion of conditional mean independence in \eqref{PFC2}, we consider  the  general  Inverse  Conditional Mean  Independence
(ICMI) relationship,  defined as
\begin{eqnarray}\label{conmeanind1}
	E\{\textbf{X} | \textbf{Y}\}=E\{\textbf{X}\},  ~a.s.,
\end{eqnarray}
where $\textbf{X}\in \mathbb{R}^{p}$  and $\textbf{Y}\in \mathbb{R}^{q}.$   Note that the response $\textbf{Y}$  can be multivariate.     To quantify this relationship, we propose two types of ICMI matrices. Crucially, we show that the eigenvectors corresponding to nonzero eigenvalues of these matrices capture directions of conditional mean dependence, while their orthogonal complement-the null space-identifies the subspace of conditional mean independence, as illustrated in \eqref{independ-beta} and \eqref{depend-beta}. As a result, the ICMI matrices offer a unified and natural framework for estimating   the central subspace $\mathcal{S}_{Y \mid \mathbf{X}}$ in model \eqref{intr0}  under the linearity assumption, as established in Theorem \ref{them-projection}.

The ICMI method described above relies on the conditional mean \(E\{\mathbf{X} \mid \mathbf{Y}\}\), making it a first-order approach. However, first-order methods are known to be ineffective when inverse mean degeneracy occurs, see   \citet{li2018sufficient}. To address this issue, we introduce the concept of  Inverse Conditional Variance Independence (ICVI), defined as
\begin{eqnarray}\label{equvlent-IV}
	{E}\{[\mathbf{X}-{E}\{\mathbf{X}\}][\mathbf{X}-{E}\{\mathbf{X}\}]^T| \mathbf{Y}\} =\boldsymbol{\Sigma}, ~a.s.,
\end{eqnarray}
where   $\boldsymbol{\Sigma}={E}\{[\mathbf{X}-{E}\{\mathbf{X}\}][\mathbf{X}-{E}\{\mathbf{X}\}]^T\}.$
Analogous to the ICMI framework, we develop a class of ICVI methods to quantify deviations from condition \eqref{equvlent-IV}. These methods enable the recovery of the central subspace \(\mathcal{S}_{Y \mid \mathbf{X}}\)  under the linearity and constant variance assumptions, as established in Theorem \ref{them-kenel-HD}.

A frequently studied problem is measuring the conditional mean independence
\begin{eqnarray}\label{scale-vector-mean-independ}
	E\{Y| \mathbf{X}\}=E\{Y\},  ~a.s.
\end{eqnarray}
Many statistical inference problems ultimately reduce to~\eqref{scale-vector-mean-independ}, such as regression model checking, tests of conditional moment restrictions, and feature screening (see, e.g., \citet{BIERENS1982105, Stute1997Nonparametric, Stinchcombe1998Consistent, escanciano2006consistent, shao2014martingale}). There are at least two key differences between existing studies and our work.
First, the aforementioned studies mainly address forward conditional mean independence, whereas we focus specifically on inverse conditional mean independence $E\{\textbf{X}|Y\}=E\{\textbf{X}\}$.  In subsequent sections, we further clarify the connection between inverse conditional mean independence and the SDR problem.  Second,   $Y$ is scalar in \eqref{scale-vector-mean-independ},  whereas the ICMI problem in~\eqref{conmeanind1} considers conditional mean dependence of a multivariate $Y$ with respect to a random vector. Thus, measuring ICMI is different from the forward regression problem in~\eqref{scale-vector-mean-independ}.

In the context of time series analysis, \citet{Lee2016Martingale} introduced the martingale difference divergence matrix based on~\eqref{conmeanind1} for dimension reduction in stationary multivariate time series. The dimension reduction models in \citet{Lee2016Martingale} can be equivalently formulated as factor models with martingale difference errors, thus essentially addressing unsupervised problems. From the perspective of inverse regression in~\eqref{inverseregression}, \citet{Lee2016Martingale} considers the unsupervised model in~\eqref{PC}. Therefore, our work can be viewed as a generalization of \citet{Lee2016Martingale} from unsupervised to supervised learning.

In this paper, we introduce   projection and kernel-based ICMI and ICVI matrices to quantify the relationships in~\eqref{conmeanind1} and~\eqref{equvlent-IV}. We first extend the CUME method \citep{Zhu2010Dimension} using projection techniques to handle multivariate responses, leading to the projection-based ICMI matrix.  We then employ Bochner's theorem to derive a class of translation-invariant kernel-based ICMI matrices. Under a non-sparse high-dimensional setting where $p = o(n)$, we establish the consistency of the resulting estimators-attaining the optimal growth rate of $p$ without sparsity assumptions. These theoretical results improve upon existing convergence rates for diverging dimensions, such as those in \citet{lin2018} and \citet{wang2019yu}.

The rest of the paper is organized as follows. Section~\ref{section2} introduces the ICMI matrix to characterize the relationship in~\eqref{conmeanind1}. Section~\ref{section3} presents  the projection-based ICMI matrix and establishes theoretical properties of its estimator. Section~\ref{section4}  develops the kernel-based ICMI matrix.
The framework is extended to the ICVI approach in Section \ref{Extension}.
Numerical studies, including simulations and real-data analyses, are reported in Sections~\ref{simulation} and~\ref{Realdata}, respectively. The paper concludes with a discussion in Section~\ref{Conclusion}, and all technical proofs are provided in the Appendix.

\section{ICMI matrix} \label{section2}

 This section constructs an inverse conditional mean independence  matrix to quantify the relationship in \eqref{conmeanind1}. We demonstrate that the eigenvectors of this matrix   encode the central subspace $\mathcal{S}_{Y \mid \mathbf{X}}$, thereby providing a principled framework for dimension reduction via spectral analysis.

We begin by assuming that
 $\textbf{X}$ is a univariate random variable. In this case, the ICMI relationship in \eqref{conmeanind1} admits the   equivalent forms
 \begin{eqnarray}\label{equvlent-IM}
	&&\mathbb{E}\{X | \mathbf{Y}\}=\mathbb{E}\{X\}, a.s.\\
	 & \Longleftrightarrow&  E\{X-E\{X\}| \mathbf{Y}\}=0 \nonumber \\
	&\Longleftrightarrow & E\{[X-E\{X\}] w(\mathbf{Y}, \boldsymbol{\mathbf{\xi}})  \}=0,\forall w(\cdot, \boldsymbol{\xi})\in \mathcal{H},\nonumber
\end{eqnarray}
where   $``\Leftrightarrow"$ denotes equivalence   and  $\mathcal{H}$  is   a suitably chosen function space. Here,  $w(\cdot, \boldsymbol{\mathbf{\xi}})$, parameterized by $\boldsymbol{\mathbf{\xi}}$,  is referred to as a weight function.   As shown in \citet{Stinchcombe1998Consistent},   the equivalence in \eqref{equvlent-IM} holds if and only if   $\mathrm{Span}\{ \mathcal{H}\}$  is   dense  in the weak topology of the conjugate space of $L^{p}(\Omega, \sigma(\mathbf{Y}), P),$ for some $p \in[1, \infty]$.   Various classes of    $\mathcal{H}$ have been proposed in the statistics and econometrics literature, including
$\mathcal{H}=\{ w(\mathbf{Y}, \boldsymbol{\mathbf{\xi}})\!:\! w(\mathbf{Y}, \boldsymbol{\mathbf{\xi}})=I( \mathbf{Y}\leq \boldsymbol{\mathbf{\xi}}),\boldsymbol{\mathbf{\xi}} \in \mathbb{R}^{q} \}$ \citep{Stute1997Nonparametric},   $\mathcal{H}=\{ w(\mathbf{Y}, \boldsymbol{\mathbf{\xi}})\!:\! w(\mathbf{Y}, \boldsymbol{\mathbf{\xi}})=\exp(i \boldsymbol{\mathbf{\xi}}^T  \mathbf{Y}),\boldsymbol{\mathbf{\xi}} \in \mathbb{R}^{q}, i=\sqrt{-1} \}$ \citep{BIERENS1982105}, $\mathcal{H}=\{ w(\mathbf{Y}, \boldsymbol{\mathbf{\xi}})\!:\! w(\mathbf{Y}, \boldsymbol{\mathbf{\xi}})=\exp(\boldsymbol{\mathbf{\xi}}^T  \mathbf{Y}),\boldsymbol{\mathbf{\xi}} \in \mathbb{R}^{q} \}$ \citep{bierens1990consistent, Stinchcombe1998Consistent},
$\mathcal{H}=\{ w(\mathbf{Y}, \boldsymbol{\mathbf{\xi}}): w(\mathbf{Y}, \boldsymbol{\mathbf{\xi}})=1 /\{1+\exp (c-\boldsymbol{\mathbf{\xi}}^T  \mathbf{Y})\}, c \neq 0,\boldsymbol{\xi} \in \mathbb{R}^{q} \}$\citep{Bierens97Ploberger} and so on.

Based on the equivalence in \eqref{equvlent-IM}, we define the ICMI measure using the weight function $w(\cdot, \cdot)$ as
\begin{eqnarray}\label{meas-weight}
	\mathrm{ICMI}_{w}(X|\mathbf{Y}) &=&  E\{[E\{(X_1- E\{X_1\}) w(\mathbf{Y}_1, \mathbf{Y}_2 )\}]^2 \}\nonumber\\
	&=&\mathrm{E}\{ ( {X}_1-E\{ {X}_1\}) ({X}_2-E\{ {X}_2\})\nonumber\\
	&&\times   w(\mathbf{Y}_1, \mathbf{Y}_3 ) \overline{w(\mathbf{Y}_2, \mathbf{Y}_3)}\},
\end{eqnarray}
where  $(X_i,\mathbf{Y}_i)$ for  $i=1,2,3$     are   independent and identically distributed (i.i.d.) copies  of $(X,\mathbf{Y})$, and
$\overline{w(\mathbf{Y}_2, \mathbf{Y}_3)}$  denotes the complex conjugate of $ {w(\mathbf{Y}_2, \mathbf{Y}_3)}.$ From  \eqref{equvlent-IM} and \eqref{meas-weight}, it follows that   $\mathrm{ICMI}_{w}(X|\mathbf{Y})=0$    if and only if  $E\{{X}|\mathbf{Y}\}=E\{X\}$ a.s.
Hence, $\mathrm{ICMI}_{w}(X|\mathbf{Y})$ serves as a valid measure of ICMI.

We now focus on the indicator weight class  $\mathcal{H}=\{ w(\mathbf{Y},\boldsymbol{\mathbf{\xi}})\!:\! w(\mathbf{Y}, \boldsymbol{\mathbf{\xi}})=I( \mathbf{Y}\leq \boldsymbol{\mathbf{\xi}}),\boldsymbol{\mathbf{\xi}} \in \mathbb{R}^{q} \}$.   In this case, $\mathrm{ICMI}_{w}(X|\mathbf{Y})^2$ becomes
\begin{eqnarray*}\label{IM-I}
	\mathrm{ICMI}_{\mathrm{ID}}(X|\mathbf{Y})
	&=& \mathrm{E} \{ ( {X}_1\!-\!E\{ {X}_1\}) ({X}_2\!-\! E\{ {X}_2\})\\
	&&\times I(\mathbf{Y}_1\leq \mathbf{Y}_3 ) I(\mathbf{Y}_2\leq \mathbf{Y}_3 ) \}.
\end{eqnarray*}
Here, the subscript  ``ID"  indicates the use of the indicator weight function.
When  $\textbf{Y}$ is univariate, \citet{jasa2011tm10563} introduced the SIRS measure for feature screening, defined as
\begin{equation*}
	{  { \Lambda}}_\mathrm{SIRS} =\mathbb{E} \{\Omega^{2}( {Y}) \},
\end{equation*}
where $\Omega(y)=\mathbb{E}\{[X-E\{X\} ]E[I( {Y}< {y} )| X]\}=\operatorname{cov}(X, I( {Y}< {y})).$      It is easily verified that ${  { \Lambda}}_\mathrm{SIRS}= \mathrm{ICMI}_{\mathrm{ID}}(X| {Y}),$ showing that ${  { \Lambda}}_\mathrm{SIRS}$ is a special case of $\mathrm{ICMI}_{w}(X|\mathbf{Y})$  under the indicator weight.

We now extend the measure
 $\mathrm{ICMI}_{\mathrm{ID}}(X| \mathbf{Y})$   to the multivariate setting where   $\textbf{X}\in \mathbb{R}^{p} $. A natural generalization is given by the matrix
 \begin{eqnarray*}
 \mathbf{M}&=&\mathrm{E} \{( \textbf{X}_1-E\{ \textbf{X}_1\}) (\textbf{X}_2-E\{ \textbf{X}_2\})^{T} \\
	&&\times I(\mathbf{Y}_1\leq \mathbf{Y}_3 ) I(\mathbf{Y}_2\leq \mathbf{Y}_3 )  \}.
\end{eqnarray*}
The following results clarify the link between  $\mathbf{M}$ and the conditional mean independence in   \eqref{PFC2}:
\begin{description}
\item [(i)] If there exists a nonzero vector $\boldsymbol{\beta} \in \mathbb{R}^{p}$ such that $E\{\boldsymbol{\beta}^{T}\textbf{X} \mid \textbf{Y}\} = E\{\boldsymbol{\beta}^{T}\textbf{X}\}$, then by the properties of $\mathrm{ICMI}_{\mathrm{ID}}(X|\mathbf{Y})$, we have
\begin{eqnarray}\label{independ-beta}
\mathrm{ICMI}_{\mathrm{ID}}(\boldsymbol{\beta}^{T}\textbf{X}|\mathbf{Y})
=\boldsymbol{\beta}^{T}\mathbf{M}\boldsymbol{\beta}=0.
\end{eqnarray}
This implies that $\boldsymbol{\beta}$ is an eigenvector of $\mathbf{M}$ corresponding to a zero eigenvalue.

\item [(ii)] Conversely, if there exists a nonzero $\boldsymbol{\beta} \in \mathbb{R}^{p}$ such that $E\{\boldsymbol{\beta}^{T}\textbf{X} \mid \textbf{Y}\} \neq E\{\boldsymbol{\beta}^{T}\textbf{X}\}$, then
     \begin{eqnarray}\label{depend-beta}
     \mathrm{ICMI}_{\mathrm{ID}}(\boldsymbol{\beta}^{T}\textbf{X}|\mathbf{Y})
=\boldsymbol{\beta}^{T}\mathbf{M}\boldsymbol{\beta}> 0.
\end{eqnarray}
Such $\boldsymbol{\beta}$ corresponds to an eigenvector of $\mathbf{M}$ associated with a nonzero eigenvalue.
\end{description}

In summary, the matrix $\mathbf{M}$ systematically encodes the ICMI relationship between linear combinations of $\mathbf{X}$ and $\mathbf{Y}$,   providing the basis for identifying conditionally mean independent directions and thus estimating $\mathbf{B}_{\perp}$ (equivalently, $\mathbf{B}$),  as established in \eqref{PFC2}. We therefore define $\mathbf{M}$ as the Inverse Conditional Mean Independence Matrix (ICMIM)  under the indicator weight function
\begin{eqnarray*}\label{ICMIM-indictor}
	\mathrm{ICMI.M}_\mathrm{ID}(\textbf{X}|\mathbf{Y})\!\!
	&=&\!\!\mathrm{E} \{( \textbf{X}_1\!\!-\!\! E\{ \textbf{X}_1\}) (\textbf{X}_2\!\!-\!\!E\{ \textbf{X}_2\})^{T}\\
	&&\times I(\mathbf{Y}_1\leq \mathbf{Y}_3 ) I(\mathbf{Y}_2\leq \mathbf{Y}_3 ) \}.
\end{eqnarray*}
This matrix serves as a foundational tool for quantifying the ICMI structure in \eqref{conmeanind1} and recovering the central subspace $\mathcal{S}_{Y \mid \mathbf{X}}$.

When $\textbf{Y}$ is a univariate  random  variable, \citet{Zhu2010Dimension}   proposed the CUME matrix for estimating   $\mathcal S_{Y|{\bf X}},$ given by
\begin{eqnarray*}\label{CUME}
	{ \boldsymbol{ \Lambda}}_\mathrm{CUME}
	&=&\mathrm{E} \{\mathbf{m}( {Y}) \mathbf{m}^{T}( {Y})\varpi( {Y})\}\nonumber\\
	&=&\mathrm{E} \{ ( \mathbf{X}_1-E\{ \mathbf{X}_1\}) (\mathbf{X}_2-E\{ \mathbf{X}_2\})^T\nonumber\\
	&&\times I(Y_1\leq Y_3 ) I(Y_2\leq Y_3 )
	\varpi( Y_3) \},
\end{eqnarray*}
where   $ \mathbf{m}( {y})=\mathrm{E}\{ (\mathbf{X}-E\{\textbf{X}\}) I(Y \leq  {y})\} $  for $ {y} \in \mathbb{R} $  and  $\varpi(\cdot)$ is a nonnegative weight function.  Clearly,  ${ \boldsymbol{ \Lambda}}_\mathrm{CUME}=\mathrm{ICMI.M}_\mathrm{ID}(\textbf{X}|\mathbf{Y})$  if $\varpi(\cdot)=1$.

A key distinction is that $\mathrm{ICMI}_{\mathrm{ID}}(X|\mathbf{Y})$ and $\mathrm{ICMI.M}_{\mathrm{ID}}(\textbf{X}|\mathbf{Y})$ extend to multivariate $\mathbf{Y}$, unlike the ${  { \Lambda}}_\mathrm{SIRS}$  and ${ \boldsymbol{ \Lambda}}_\mathrm{CUME}$ methods. Additionally, although all four rank-based statistics benefit from robustness to heavy-tailed distributions and require no moment assumptions, they share a fundamental weakness: their reliance on fragile component-wise comparisons. A minor perturbation in any dimension can drastically alter the result, leading to their well-known instability in multivariate settings-a challenge we directly address in the subsequent sections.

Note that the SDR method based on the PFC model in \eqref{PFC} requires the error term $\boldsymbol{\varepsilon}$ to follow a Gaussian distribution. In contrast, our ICMI method using $\mathrm{ICMI.M}_{\mathrm{ID}}(\textbf{X} \mid \mathbf{Y})$ does not rely on this assumption.
This can be seen from   \eqref{eq:cond-moments}, which only requires $ {E}\{\boldsymbol{\varepsilon} \mid \mathbf{Y})\} = 0$ to hold. Although our  method is motivated by the PFC framework, it is fundamentally built upon the ICMR relationship given in \eqref{conmeanind1}, which is equivalent to $ {E}\{\boldsymbol{\varepsilon} \mid \mathbf{Y}\} = 0$ with $\boldsymbol{\varepsilon} = \mathbf{X} - {E}\{\mathbf{X}\}$. A deeper discussion of the connection
  between   ICMI-based and PFC-based methods is provided in Section \ref{IRMR}.

\section{ICMI matrix via projection }\label{section3}

\subsection{Methodology}

As aforementioned, the rank-based measures $\mathrm{ICMI}_{\mathrm{ID}}(X|\mathbf{Y})$ and $\mathrm{ICMI.M}_{\mathrm{ID}}(\textbf{X}|\mathbf{Y})$ exhibit instability in multivariate settings. To overcome this, we adopt a projection technique, drawing upon its well-documented strength in the literature \citep{Zhu1997On, escanciano2006consistent, Baringhaus2004On, Zhu2017Projection}.

We begin by assuming that $\textbf{X}$ is a univariate random variable. For any $q > 1$, let $\mathbb{S}^{q-1} = \{\boldsymbol{\beta} \in \mathbb{R}^q : \|\boldsymbol{\beta}\| = 1\}$ denote the unit hypersphere in $\mathbb{R}^q$. According to Theorem 1 in \citet{BIERENS1982105}, the following equivalence holds
\begin{eqnarray*}
	&&E\{X \mid \mathbf{Y}\} = E\{X\}, \quad \text{a.s.} \\
	&\Longleftrightarrow& E\{[X - E\{X\}]   I(\boldsymbol{\beta}^\top\mathbf{Y} \leq u)\} = 0,
\end{eqnarray*}
for all $u \in \mathbb{R}$ and $\boldsymbol{\beta} \in \mathbb{S}^{q-1}$.
By averaging over all $u \in \mathbb{R}$ and all one-dimensional projection directions $\boldsymbol{\beta} \in \mathbb{S}^{p-1}$, we conclude that $E\{X \mid \mathbf{Y}\} = E\{X\},$ a.s. if and only if
\begin{eqnarray*}
	\int_{\mathbb{S}^{p-1}} \int_{-\infty}^\infty \{ E\{[X - E\{X\}]  I(\boldsymbol{\beta}^\top\mathbf{Y} \leq u)\} \}^2 \\
	\times \, dF_{\boldsymbol{\beta}^\top\mathbf{Y}}(u) \, d\boldsymbol{\beta} = 0.
\end{eqnarray*}
Then, we  define the following measure
\begin{eqnarray*}
	&&\mathrm{ICMI}_\mathrm{PR}({X}|\mathbf{Y}) \\
	&=&\int_{\mathbb{S}^{p-1}} \int_{-\infty}^\infty \{ E\{[X-E\{X\}]I(\boldsymbol{\beta}^T\mathbf{Y}\leq u)\} \}^2\\
	&&\times d F_{\boldsymbol{\beta}^T\mathbf{Y}}(u) d\boldsymbol{\beta}.
\end{eqnarray*}
Here, the subscript ``PR" in $\mathrm{ICMI}_{\mathrm{PR}}(X \mid \mathbf{Y})$ indicates that it relies on a projection technique.

By \citet{escanciano2006consistent} or  Lemma 1 in  \citet{Zhu2017Projection}, some algebra yields that
\begin{eqnarray*}\label{IMP}
	\mathrm{ICMI}_\mathrm{PR}({X}|\mathbf{Y}) \!\!&=&\!\!-\mathrm{E}\{ ( {X}_1\!\!-\!\! E\{ {X}_1\}) ({X}_2\!\!-\!\! E\{ {X}_2\})\\
	&&\times\mathrm{ang}(\mathbf{Y}_1- \mathbf{Y}_3,\mathbf{Y}_2- \mathbf{Y}_3)  \},
\end{eqnarray*}
where $\mathrm{ang}(\textbf{U}_1,\textbf{U}_2)$ denotes the angle between vectors
$\textbf{U}_1$ and $\textbf{U}_2$, defined as
$$\mathrm{ang}(\textbf{U}_1,\textbf{U}_2)=\arccos \{ {\textbf{U}_1^T \textbf{U}_2}/(\|\textbf{U}_1\|_2\| \textbf{U}_2\|_2) \}.$$
From this definition, it follows that $\mathrm{ICMI}_\mathrm{PR}({X}|\mathbf{Y}) \geq0,$   with equality     if and only if    $E\{{X}|\mathbf{Y}\}=E\{X\},$ a.s.

We   extend the definition of $\mathrm{ICMI}_\mathrm{PR}(X \mid \mathbf{Y})$ to the multivariate case where $\mathbf{X}$ is a random vector. Following the reasoning for $\mathrm{ICMI.M}_{\mathrm{ID}}(\mathbf{X} \mid \mathbf{Y})$, we define the projection-based ICMI matrix as
\begin{eqnarray*}\label{MCUME}
	\mathrm{ICMI.M}_\mathrm{PR}(\textbf{X}|\mathbf{Y})\!\!
	&=&\!\!-\mathrm{E}\{  ( \mathbf{X}_1- E\{\! \mathbf{X}_1\!\}) (\mathbf{X}_2\!\!-\!\! E\{\! \mathbf{X}_2\!\})^T\\ &&\times\mathrm{ang}(\mathbf{Y}_1- \mathbf{Y}_3,\mathbf{Y}_2- \mathbf{Y}_3)\!\}.
\end{eqnarray*}
In contrast to $\mathrm{ICMI.M}_\mathrm{ID}(\mathbf{X} \mid \mathbf{Y})$,   which depends on the ranks of  $\mathbf{Y}$, the  measure $\mathrm{ICMI.M}_\mathrm{PR}(\mathbf{X} \mid \mathbf{Y})$  is constructed using angles between differences of  $\mathbf{Y}$ vectors.  This   preserves the advantages of $\mathrm{ICMI.M}_\mathrm{ID}(\mathbf{X} \mid \mathbf{Y})$ while mitigating potential instability caused by component-wise ranking.

\begin{prop} \label{PR_CUME}
If ${Y}$ is continuous and univariate, and the weight function  is chosen as   $\varpi(\cdot) = 1$ in the definition of $\boldsymbol{\Lambda}_{\mathrm{CUME}}$, then  we have
$\mathrm{ICMI.M}_\mathrm{PR}(\textbf{X}|{Y})=2\pi{ \boldsymbol{ \Lambda}}_\mathrm{CUME}.$
 \end{prop}

Proposition \ref{PR_CUME} shows that when $\mathbf{Y}$ is continuous and univariate, $\mathrm{ICMI.M}_{\mathrm{PR}}(\mathbf{X} \mid Y)$ is proportional to $\boldsymbol{\Lambda}_{\mathrm{CUME}}$, differing only by a constant factor of $2\pi$. Therefore, $\mathrm{ICMI.M}_{\mathrm{PR}}(\mathbf{X} \mid Y)$ can be viewed as a natural extension of $\boldsymbol{\Lambda}_{\mathrm{CUME}}$ to the case where $\mathbf{X}$ is a multivariate random vector.

A natural point of comparison for our proposed measure $\mathrm{ICMI.M}_{\mathrm{PR}}(\mathbf{X} \mid \mathbf{Y})$ is the martingale difference divergence matrix (MDDM) introduced by \citet{Lee2016Martingale},  given by
\begin{eqnarray*}
	&&\mathrm {MDDM}(\mathbf{X}| \mathbf{Y})\nonumber\\
	&=&\!\!  \int_{\mathbb{R}^q}\!\! | E\{\mathbf{X}e^{i<\textbf{s}, \mathbf{Y}>}\}\!-\!E\{\mathbf{X}\}E\{e^{i<\textbf{s}, \mathbf{Y}>}\}|^2\omega(\textbf{s})d\mathbf{s}\nonumber \\
	&=&  -\mathrm{E} \{( \mathbf{X}_1\!-\!E\{ \mathbf{X}_1\}) (\mathbf{X}_2\!-\!E\{ \mathbf{X}_2\})^T\|\textbf{Y}_1\!-\!\textbf{Y}_2 \| \},
\end{eqnarray*}
with $\omega(\mathbf{s}) = 1/(\|\mathbf{s}\|^{1+q}c_q)$, $c_q = \pi^{(q+1)/2}/\Gamma((q+1)/2)$, and $\Gamma(\cdot)$ being the gamma function.

The connection and distinction between $\mathrm{MDDM}(\mathbf{X} \mid \mathbf{Y})$ and $\mathrm{ICMI.M}_{\mathrm{PR}}(\mathbf{X} \mid \mathbf{Y})$ can be clarified using Lemma 6.1 of \citet{2018arXiv180300715K}. In particular, the angular distance between any $\mathbf{y}_1, \mathbf{y}_2 \in \mathbb{R}^{q}$ is defined as
\[
\rho_{\text{Angle}}^{\mathbf{W}}(\mathbf{y}_1, \mathbf{y}_2) = \mathbb{E}_{\mathbf{W}}[\mathrm{ang}(\mathbf{y}_1 - \mathbf{W}, \mathbf{y}_2 - \mathbf{W})].
\]
Taking the expectation with respect to the Lebesgue measure $\boldsymbol{\nu}$ yields
\begin{eqnarray*}
\rho_{\text{Angle}}^{\boldsymbol{\nu}}(\mathbf{y}_1, \mathbf{y}_2) &=& \mathbb{E}_{\mathbf{W} \sim \boldsymbol{\nu}}[\mathrm{ang}(\mathbf{y}_1 - \mathbf{W}, \mathbf{y}_2 - \mathbf{W})] \\
&=& \gamma_{d} \|\mathbf{y}_1 - \mathbf{y}_2\|,
\end{eqnarray*}
where $\gamma_{d}$ is a positive constant. This leads to the following unified representations
\begin{eqnarray*}
		\!\!\mathrm{ICMI.M}_\mathrm{PR}(\textbf{X}|\mathbf{Y})
		\!\!&=&\!\!-\mathrm{E}\{ \rho_{\text {Angle}}^{\textbf{Y}}(\mathbf{Y}_1, \mathbf{Y}_2)\} \nonumber\\
		&&\times ( \mathbf{X}_1\!\!-\!\! E\{ \!\mathbf{X}_1\!\})(\mathbf{X}_2\!\!-\!\! E\{\! \mathbf{X}_2\!\})^T \},\label{MDDM-PR}\nonumber\\
		\!\!\mathrm {MDDM}(\mathbf{X}| \mathbf{Y})\!\!&=&\!\!-\mathrm{E}\{\rho_{\text {Angle}}^{\boldsymbol{\nu}}(\mathbf{Y}_1, \mathbf{Y}_2) \\
		&&\times(\mathbf{X}_1\!\!-\!\! E\{\! \mathbf{X}_1\!\}) (\mathbf{X}_2\!\!-\!\! E\{ \!\mathbf{X}_2\!\})^T\}.\nonumber
\end{eqnarray*}
The  key  difference between the two methods is  their distance metrics: $\mathrm{ICMI.M}_{\mathrm{PR}}$ adopts $\rho_{\text{Angle}}^{\mathbf{Y}}$   based on the     distribution of $\mathbf{Y}$, whereas $\mathrm{MDDM}$ relies on $\rho_{\text{Angle}}^{\boldsymbol{\nu}}$  derived from the Lebesgue measure. Moreover,  $\mathrm{MDDM}$ requires the moment condition $\mathrm{E}\{\|\mathbf{Y}_1 - \mathbf{Y}_2\|^2\} < \infty$ to be well-defined, whereas $\mathrm{ICMI.M}_{\mathrm{PR}}$  imposes no such requirement on $\mathbf{Y}$. As a result, $\mathrm{ICMI.M}_{\mathrm{PR}}$ exhibits   robustness when $\mathbf{Y}$ follows a heavy-tailed distribution.

As outlined in Section \ref{section2}, ICMI matrices can be used to identify the central subspace $\mathcal{S}_{Y \mid \mathbf{X}}$. The following theorem establishes the validity of our ICMI approach for sufficient dimension reduction.

\begin{thm}\label{them-projection}
	Assume that the linearity condition holds, i.e.,  $E\{\mathbf{X} | \mathbf{\Gamma}^T \mathbf{X}\}=\mathbf{P}_{\mathbf{\Gamma}}^T(\mathbf{\Sigma}) \mathbf{X}$,   where  $\mathbf{P}_{\mathbf{\Gamma}}(\mathbf{\Sigma})=\mathbf{\Gamma}(\mathbf{\Gamma}^T \mathbf{\Sigma} \Gamma)^{-1} \mathbf{\Gamma}^T \boldsymbol{\Sigma} $.   Then, we have  $\mathcal{S}({ \boldsymbol{ \Lambda}}_\mathrm{PR}) \subseteq \mathbf{\Sigma} \mathcal{S}_{\textbf{Y}| \textbf{X}}$, where ${ \boldsymbol{ \Lambda}}_\mathrm{PR}=  \mathrm{ICMI.M}_\mathrm{PR}(\textbf{X}|\mathbf{Y})$.
\end{thm}

The linearity condition   in Theorem \ref{them-projection} is standard in the SDR literature (see, e.g., \citet{Likc91}), and  generally holds  when $\mathbf{X}$ follows an elliptically symmetric distribution. Under this condition, the theorem implies    that the eigenvectors of $\mathbf{\Sigma}^{-1} \boldsymbol{\Lambda}_\mathrm{PR}$ lie in the central subspace $\mathcal{S}_{\mathbf{Y} \mid \mathbf{X}}$, thereby providing a basis for its   estimation.

We now consider the estimation of   $\mathcal{S}_{\mathbf{Y} \mid \mathbf{X}}$. Let $\{\left(\mathbf{X}_{i}, \mathbf{Y}_{i}\right), i=1, \cdots, n\}$  be  i.i.d. copies of $(\mathbf{X}, \mathbf{Y}),$  where $\mathbf{X}_{i}=({X}_{i1},\cdots, {X}_{ip})^T$ and $\mathbf{Y}_{i}=({Y}_{i1},\cdots, {Y}_{iq})^T.$ The sample mean and covariance matrix of
 $\mathbf{X}$ are given by
  $\overline{\mathbf{X}}_n=\frac{1}{n} \sum_{i=1}^{n} \mathbf{X}_{i}$ and $\widehat{\mathbf{\Sigma}}_n=\frac{1}{n} \sum_{i=1}^{n}\left(\mathbf{X}_{i}-\overline{\mathbf{X}}_n\right)\left(\mathbf{X}_{i}-\overline{\mathbf{X}}_n\right)^T.$ An estimator for $\mathrm{ICMI.M}_\mathrm{PR}(\textbf{X}|\mathbf{Y})$ is   provided by
\begin{eqnarray*}
	\widehat{\mathrm{ICMI.M}}_{\mathrm{PR}}(\mathbf{X} | \mathbf{Y})
	&=& - \frac{1}{n^3}\!\! \sum_{i,j,k=1}^n\!\!
	(\mathbf{X}_i\!\! -\!\!\overline{\mathbf{X}}_{n} )(\mathbf{X}_j\!\! -\!\!\overline{\mathbf{X}}_{n} )^T\\
	&&\times \mathrm{ang}(\mathbf{Y}_i -\mathbf{Y}_k, \mathbf{Y}_j-\mathbf{Y}_k)\\
	&=&\widehat{\boldsymbol{ \Lambda}}_\mathrm{PR}.
\end{eqnarray*}

Since $\widehat{\boldsymbol{ \Lambda}}_\mathrm{PR}$ only involves simple moment estimates, it can be easily computed from the sample. By Theorem \ref{them-projection}, and assuming the structural dimension
$d$  is known, the central subspace can be estimated by the span of the first
$d$ leading eigenvectors of ${\widehat{\mathbf{\Sigma}}_n^{-1} \widehat{ \boldsymbol{ \Lambda}}_\mathrm{PR}}.$

\subsection{Asymptotic properties}

This section establishes the theoretical properties of the projection-based ICMI method in high-dimensional settings. Before presenting the main results, we introduce some necessary notation. For a matrix \(\mathbf{A} \in \mathbb{R}^{p \times p}\), let \(\lambda_{\min}(\mathbf{A})\) and \(\lambda_{\max}(\mathbf{A})\) denote its smallest and largest singular values, respectively. The spectral norm of \(\mathbf{A}\) is defined as \(\|\mathbf{A}\| = \sqrt{\lambda_{\max}(\mathbf{A}^\mathrm{T} \mathbf{A})}\). For a  random variable \(Z\), its sub-Gaussian norm
$\|\cdot\|_{\psi_{2}}$ is given by
\begin{equation*}
\|Z\|_{\psi_{2}} =  \sup _{k \geq 1} k^{-1 / 2} ( {E}\{|Z|^{k}\} )^{1 / k},
\end{equation*}
and its sub-exponential norm  $\|\cdot\|_{\psi_{1}}$ is defined as
\begin{equation*}
\|Z\|_{\psi_{1}}  =  \sup _{k \geq 1} k^{-1} ( {E}\{|Z|^{k}\} )^{1 / k}.
\end{equation*}

To establish the asymptotic properties, we make the following assumptions.

\begin{description}
\item[Condition 1.] The random vector \(\mathbf{X} \in \mathbb{R}^p\) is sub-Gaussian. That is, there exists a constant \(C > 0\) such that for any unit vector \(\boldsymbol{v} \in \mathbb{R}^p\) with \(\|\boldsymbol{v}\| = 1\),
\[
\| \boldsymbol{v}^\mathrm{T} (\mathbf{X} - \mathbb{E}[\mathbf{X}]) \|_{\psi_2}^2 \leq C \cdot \mathbb{E}\left[ (\boldsymbol{v}^\mathrm{T} (\mathbf{X} - \mathbb{E}[\mathbf{X}]))^2 \right].
\]

\item[Condition 2.] There exist positive constants \(C_1\) and \(C_2\) such that
\[
C_1 \leq \lambda_{\min}(\mathbf{\Sigma}) \leq \lambda_{\max}(\mathbf{\Sigma}) \leq C_2.
\]
\end{description}

The sub-Gaussian assumption in Condition 1   is commonly used for the consistency of the sample covariance
matrix $ \widehat{\mathbf{\Sigma}}_n$, see   Theorem 4.7.1    in  \citet{vershynin2018high} and  Assumption 1 in  \citet{bunea2015sample}.  Condition 2 requires the eigenvalues of $\mathbf{X}$ to be bounded away from zero and infinity, a standard assumption
 in the literature of high-dimensional covariance matrix estimation.
See, for example, \citet{Bickel2008Covariance},   \citet{lin2018} and  \citet{wang2019yu}.

\begin{thm}\label{them-kenerl-consistency}
Under Condition 1,  for any \(\alpha \in (0,1)\),
we have
\begin{eqnarray*}\label{T1n-S}
	\operatorname{pr}\Big( \|\widehat{ \boldsymbol{ \Lambda}}_\mathrm{PR} -{ \boldsymbol{ \Lambda}}_\mathrm{PR}\|    \leq C \|\mathbf{\Sigma}\|  \sqrt{ \frac{ p +\log(1/\alpha)}{n} } \Big)\geq 1- \alpha,
\end{eqnarray*}
where $C$ is a  positive constant.
\end{thm}

If $\|\mathbf{\Sigma}\|$ is bounded, Theorem~\ref{them-kenerl-consistency} implies the following convergence rate
\begin{equation*}\label{consis-rate-PR}
\|\widehat{\boldsymbol{\Lambda}}_{\mathrm{PR}}-\boldsymbol{\Lambda}_{\mathrm{PR}} \|
=O_{p}\left( \sqrt{ \frac{p }{n}}\right).
\end{equation*}
Hence, the estimator is consistent provided that $p = o(n)$.
Our proof relies primarily on the Chernoff method \citep{vershynin2018high} and U-statistics theory to bound the tail probability $\mathrm{pr}\big( |\widehat{\boldsymbol{\Lambda}}\mathrm{PR} - \boldsymbol{\Lambda}\mathrm{PR}| > t \big)$; see Lemma~\ref{lemma1-thprojection} in the Appendix.

It is worth noting the convergence rate $ O(\sqrt{ {p }/{n}})$  of $\widehat{\boldsymbol{\Lambda}}_\mathrm{PR}$ derived in Theorem~\ref{them-kenerl-consistency}. Two closely related works are those of \citet{lin2018} and \citet{wang2019yu}. In \citet{lin2018}, the convergence rate of  SIR  is given by
  \begin{equation*}\label{consis-rate-Lin}
\|\widehat{\boldsymbol{\Lambda}}_{\mathrm{SIR}}-\boldsymbol{\Lambda}_{\mathrm{SIR}} \|=O_{p}\left(\frac{1}{H^{\vartheta}}+\frac{H^{2} p}{n}+\sqrt{\frac{H^{2} p}{n}}\right),
\end{equation*}
where $H$ denotes the number of slices and $\vartheta$ is a nonnegative constant (see Theorem~1 in \citet{lin2018}). By choosing $H = O\left( (p / n)^{-1/(2(\vartheta+1))} \right)$, the optimal rate is achieved
\begin{equation*}\label{consis-rate-Lin-optimal}
\|\widehat{\boldsymbol{\Lambda}}_{\mathrm{SIR}}-\boldsymbol{\Lambda}_{\mathrm{SIR}} \|=O_{p}\left((p / n)^{\frac{\vartheta}{2(\vartheta+1)}}\right).
\end{equation*}
On the other hand, \citet{wang2019yu} established the following rate for ${\boldsymbol{ \Lambda}}_\mathrm{CUME}$
\begin{equation*}\label{consis-rate-wang2019yu-optimal}
\|\widehat{ \boldsymbol{ \Lambda}}_\mathrm{CUME} - { \boldsymbol{ \Lambda}}_\mathrm{CUME}\| =O_{p}\left(\sqrt{\frac{\max (p, \log n) }{n}}\right).
\end{equation*}
Compared with these results, our derived rate provides an improvement in the high-dimensional regime where the dimension $p$ diverges.

\begin{thm}\label{them-projection-consistency}
Under Conditions 1-2,
we have
\begin{eqnarray*}
	 \| \widehat{\mathbf{\Sigma}}_n^{-1}\widehat{ \boldsymbol{ \Lambda}}_\mathrm{PR}-  {\mathbf{\Sigma}}^{-1} { \boldsymbol{ \Lambda}}_\mathrm{PR} \|=O_{p}\left( \sqrt{ \frac{p }{n}}\right).
\end{eqnarray*}
\end{thm}

Theorem \ref{them-projection-consistency} shows that $\widehat{\mathbf{\Sigma}}_n^{-1}\widehat{\boldsymbol{\Lambda}}_\mathrm{PR}$ achieves a convergence rate of $O_p(\sqrt{p/n})$.  By contrast,  \citet{Zhu2010Dimension} obtained for the CUME method the rate
$$ \| \widehat{\mathbf{\Sigma}}_n^{-1}\widehat{ \boldsymbol{ \Lambda}}_\mathrm{CUME} -\mathbf{\Sigma}^{-1} \boldsymbol{ \Lambda}_\mathrm{CUME} \|=O(p n^{-1 / 2} \log n).$$
Comparison of these rates confirms that our convergence rate is sharper.

\section{ICMI matrix via kernel}\label{section4}
\subsection{Methodology}

This section presents a new class of ICMI methods based on translation-invariant kernels, developed through the equivalence in \eqref{equvlent-IM}. Our approach relies on the following celebrated Bochner's Theorem \citep{Bochner1933}, stated here in the form given by
 \citet[Theorem 6.6]{wendland_2004}.

\begin{lemma}\label{Bochner}
A continuous function  $K : \mathbb{R}^{m} \rightarrow \mathbb{C}$   is positive  semi-definite
if and only if it is the Fourier transform of a finite nonnegative Borel measure $\mu$ on $\mathbb{R}^{m} $, that is,
\begin{equation*}
	K(\mathbf{z})=\int_{\mathbb{R}^m} e^{-i \mathbf{z}^{\mathrm{T}} \mathbf{u}} d \mu(\mathbf{u}),
\end{equation*}
for any $\mathbf{z} \in \mathbb{R}^m.$
\end{lemma}

We begin by assuming that $ \mathbf{X}$ is a univariate random variable. Focusing on the Fourier weight function $w(\mathbf{Y}, \mathbf{y})=\exp(i \mathbf{y}^T\mathbf{Y})$,    \eqref{meas-weight} becomes
\begin{eqnarray*}
\mathrm{ICMI}_\mathrm{F}(X|\mathbf{Y}) &=& \mathrm{E} \{  ( {X}_1-E\{ {X}_1\}) ({X}_2-E\{ {X}_2\})\nonumber\\
&&\times  \exp (i \textbf{Y}_1^T\textbf{Y}_3 )\exp(-i \textbf{Y}_2^T\textbf{Y}_3 ) \},
\end{eqnarray*}
which can be rewritten as
\begin{eqnarray*}
\mathrm{ICMI}_\mathrm{F}(X|\mathbf{Y})
&=& \int_{\mathbb{R}^q}  |\mathrm{E}\{{X}\exp(i \textbf{u}^T\textbf{Y} )\} \\
&&-\mathrm{E}\{ X \}\mathrm{E}\{\exp(i \textbf{u}^T\textbf{Y})\}  |^2   d F_\textbf{Y}(\textbf{u}),
\end{eqnarray*}
where $ F_\textbf{Y}(\textbf{u})$ is the CDF of $\textbf{Y}.$ The subscript ``F" indicates the use of the Fourier weight function.

A natural comparison can be drawn between $\mathrm{MDDM}(X|\mathbf{Y})$ and $\mathrm{ICMI}_\mathrm{F}(X|\mathbf{Y})$. While both are derived from similar Fourier-type expressions, the former is defined with respect to the measure  $\omega(\textbf{s})d\textbf{s} = {1}/\{{\|\textbf{s}\|^{1+q}}c_q\}d\textbf{s}$,   whereas the latter uses the probability measure $\mathrm{d} F_{\mathbf{Y}}(\mathbf{u})$. This motivates a generalization via an arbitrary    nonnegative   measure $\mu$, leading to  a new class of ICMI measures. We thus  propose the following measure
 \begin{eqnarray*}
\mathrm{ICMI}_\mathrm{\mu}(X|\mathbf{Y})&=& \int_{\mathbb{R}^q}  |\mathrm{E}\{{X}\exp(i \textbf{u}^T\mathbf{Y} )\} -\mathrm{E}\{ X \}\nonumber\\
&&\times \mathrm{E}\{\exp(i \textbf{u}^T \mathbf{Y})\}  |^2  d \mu(\textbf{u}).
\end{eqnarray*}

Although
$\mathrm{ICMI}_\mathrm{\mu}(X|\mathbf{Y})$ provides a general measure of conditional mean independence,  its direct computation can be computationally intractable. However, when   $\mu(\cdot)$ is a finite nonnegative Borel measure   on $\mathbb{R}^{q}$,     Bochner's Theorem ensures that    $\mathrm{ICMI}_\mathrm{\mu}(X|\mathbf{Y})$   admits a closed-form representation, as shown in the following Theorem~\ref{Bochner-induced}.  In what follows, we restrict our attention to such measures $\mu(\cdot)$.

It should be emphasized that $\mathrm{ICMI}_\mathrm{\mu}(X|\mathbf{Y})$,   defined via a finite nonnegative Borel measure $\mu(\cdot)$ differs fundamentally from $\mathrm{MDDM}(X|\mathbf{Y})$. In the latter, the weight function   $\omega(\textbf{s})ds = {1}/\{{\|\textbf{s}\|^{1+q}}c_q\}ds,$   is not    integrable on $\mathbb{R}^{q}$. Therefore, $\mathrm{MDDM}(X|\mathbf{Y})$  cannot be viewed as a special case of $\mathrm{ICMI}_\mathrm{\mu}(X|\mathbf{Y})$ under any finite nonnegative Borel measure $\mu(\cdot)$.

\begin{thm}\label{Bochner-induced}
If $\mu(\cdot)$ is a finite nonnegative Borel measure on  $\mathbb{R}^{q} $ and $K(\mathbf{y})$ is a positive semidefinite
kernel induced by $\mu(\cdot)$, then we have
\begin{description}
	\item[(i)] $\mathrm{ICMI}_\mathrm{\mu}(X|\mathbf{Y})$ can be represented as
	\begin{eqnarray}
		\mathrm{ICMI}_\mathrm{\mu}(X|\mathbf{Y})
		&=&\mathrm{E}\big\{( {X}_1-E\{ {X}_1\}) ({X}_2-E\{ {X}_2\})\nonumber\\
		&&\times K(\mathbf{Y}_1-\mathbf{Y}_2)  \big\};\nonumber
	\end{eqnarray}
	
	\item[(ii)] $\mathrm{ICMI}_\mathrm{\mu}(X|\mathbf{Y})\geq0$ and  the equality holds  if and only if $\mathrm{E}\{X | \mathbf{Y}\}=\mathrm{E}\{X\},$ a.s.
\end{description}
\end{thm}

Theorem \ref{Bochner-induced} shows that \(\mathrm{ICMI}_{\mu}(X \mid \mathbf{Y})\) admits a representation via a translation-invariant kernel function \(K: \mathbb{R}^{q} \times \mathbb{R}^{q} \rightarrow \mathbb{C}\), which depends only on the difference between its arguments, that is, \(K(\mathbf{y}_1, \mathbf{y}_2) = K(\mathbf{y}_1 - \mathbf{y}_2)\) for any \(\mathbf{y}_1, \mathbf{y}_2 \in \mathbb{R}^{q}\). We accordingly adopt the notation \(\mathrm{ICMI}_{K}(X \mid \mathbf{Y})\) in place of \(\mathrm{ICMI}_{\mu}(X \mid \mathbf{Y})\). A variety of translation-invariant kernel functions are available in the kernel methods literature. Common examples include the Gaussian kernel,   the Laplacian kernel,   and
the  rational quadratic  kernel.  A comprehensive overview of such kernels can be found in \citet{Shawetaylor2005Kernel,scholkopf2002learning}.

We  now extend  $\mathrm{ICMI}_\mathrm{K}(X|\mathbf{Y})$ to   the multivariate setting  where    $\textbf{X}$  is multivariate. Define  the following ICMI matrix based on translation-invariant kernels
 \begin{eqnarray*}%\label{MCUME-K}
\mathrm{ICMI.M}_\mathrm{K}(\mathbf{X}|\mathbf{Y})
\!\!&=&\!\! \mathrm{E}\big\{ ( \mathbf{X}_1\!-\!E\{ \mathbf{X}_1\}) (\mathbf{X}_2\!\!-\!\! E\{ \mathbf{X}_2\})^T\\
&&\times  K(\mathbf{Y}_1\!\!-\! \!\mathbf{Y}_2)  \big\}= { \boldsymbol{ \Lambda}}_\mathrm{K}.
\end{eqnarray*}
A natural estimator of   $ { \boldsymbol{ \Lambda}}_\mathrm{K}$
  is   given by
\begin{eqnarray*}
\widehat{\mathrm{ICMI.M}}_\mathrm{K}(\mathbf{X} | \mathbf{Y})
\!\!&=&\!\! \frac{1}{n^2} \sum_{i,j=1}^n
(\mathbf{X}_i -\overline{\mathbf{X}}_{n} ) (\mathbf{X}_j -\overline{\mathbf{X}}_{n} )^T\\
&&\times \mathrm{K}(\mathbf{Y}_i-\mathbf{Y}_j)=\widehat{ \boldsymbol{ \Lambda}}_\mathrm{K}.
\end{eqnarray*}
As in Theorem \ref{them-projection}, we obtain the following result.

\begin{thm}\label{them-kenel}
Assume that the linearity condition holds  $E\{\mathbf{X} | \mathbf{\Gamma}^T \mathbf{X}\}=\mathbf{P}_{\mathbf{\Gamma}}^T(\mathbf{\Sigma}) \mathbf{X}$,   where  $\mathbf{P}_{\mathbf{\Gamma}}(\mathbf{\Sigma})=\mathbf{\Gamma}(\mathbf{\Gamma}^T \mathbf{\Sigma} \Gamma)^{-1} \mathbf{\Gamma}^T \boldsymbol{\Sigma} $.   Then, we have
 $\mathcal{S}({ \boldsymbol{ \Lambda}}_\mathrm{K}) \subseteq \mathbf{\Sigma} \mathcal{S}_{\textbf{Y}| \textbf{X}}$.
\end{thm}

Theorem \ref{them-kenel} ensures that   $\mathrm{ICMI.M}_\mathrm{K}(\mathbf{X}|\mathbf{Y})$   lies in     $\mathcal S_{\textbf{Y}| \textbf{X}}$ under the linearity condition. We now present the convergence rate of  $\widehat{\mathrm{ICMI.M}}_\mathrm{K}(\mathbf{X} | \mathbf{Y})$ in the following theorem.

\begin{thm}\label{them- Bochner-consistency}
Under Condition 1,  for any \(\alpha \in (0,1)\),
we have
\begin{eqnarray*}\label{T1n-S-kernel}
	\mathrm{pr}\Big(  \|\widehat{ \boldsymbol{ \Lambda}}_\mathrm{K} - { \boldsymbol{ \Lambda}}_\mathrm{K}\|    \leq C \|\mathbf{\Sigma}\|  \sqrt{ \frac{p+\log(1/\alpha)}{n} } \Big)\geq 1-2\alpha,
\end{eqnarray*}
where  $C>0$ is a constant.  Under the additional Condition 2, we have
\begin{eqnarray*}
	 \|\widehat{\mathbf{\Sigma}}_n^{-1}\widehat{ \boldsymbol{ \Lambda}}_\mathrm{K}-  {\mathbf{\Sigma}}^{-1} { \boldsymbol{ \Lambda}}_\mathrm{K} \|=O_{p}\left( \sqrt{ \frac{p }{n}}\right).
\end{eqnarray*}
\end{thm}

Theorem \ref{them- Bochner-consistency} shows  that $\widehat{\mathbf{\Sigma}}_n^{-1}\widehat{ \boldsymbol{ \Lambda}}_\mathrm{K}$  is consistent at the rate of $\sqrt{ {p }/{n}},$ provided   $p/n\rightarrow 0$.  These results align with those in Theorems~\ref{them-kenerl-consistency} and~\ref{them-projection-consistency}.

\subsection{Connections to the  PFC model }\label{IRMR}
This section investigates the connections among the PFC model in \eqref{PFC}, the Factor Model with Martingale Difference Error (FM-MDE) \citep{Lee2016Martingale}, and our kernel-based ICMI methods.

By the properties of the kernel function, there exist a Hilbert space $\mathcal{H}$ and a feature map $\boldsymbol{\phi} : \mathbb{R}^{q} \rightarrow \mathcal{H}$ such that
$$K(\mathbf{y}_1-\mathbf{y}_2)=\langle\boldsymbol{\phi}(\mathbf{y}_1), \boldsymbol{\phi}(\mathbf{y}_2)\rangle_\mathcal{H},$$
where $\langle \cdot, \cdot \rangle_\mathcal{H}$ denotes the inner product in $\mathcal{H}$. It can be shown that
\begin{eqnarray}\label{equ-kernel}
E\{\mathbf{X}|\mathbf{Y}\}\!=\! E \{\mathbf{X}\} \!\Longleftrightarrow\!{E}\{\mathbf{X}| \boldsymbol{\phi}(\mathbf{Y})\}\!=\!\mathbb{E}\{\mathbf{X}\}.~~~
\end{eqnarray}
A proof is provided in the Appendix.
Motivated by \eqref{equ-kernel} and the discussion in Introduction, we consider the following inverse conditional mean regression model
\begin{equation}\label{PFC-kenerls-mean}
\mathbf{X} =\boldsymbol{\mu}+\mathbf{B}\boldsymbol{\beta} \boldsymbol{\phi}(\mathbf{Y})+\boldsymbol{\varepsilon},  ~ \text{ with }~ E\{\boldsymbol{\varepsilon} |\mathbf{Y}\}=0,
\end{equation}
where $\boldsymbol{\mu} \in \mathbb{R}^{p}$, $\mathbf{B} \in \mathbb{R}^{p \times d}$, and $\boldsymbol{\beta} \in \mathbb{R}^{d \times \mathrm{dim}(\mathcal{H})}$. Here, $\mathrm{dim}(\mathcal{H})$ may be infinite in theory.

In the context of time series analysis, \citet{Lee2016Martingale} introduced the FM-MDE model for stationary multivariate time series $\mathbf{X}_t \in \mathbb{R}^p$, which takes the form
\begin{equation*}
\textbf{X}_{t}=\mathbf{A} \textbf{Z}_{t}+\boldsymbol{\varepsilon}_{t},
\end{equation*}
where $\mathbf{A}$ is a $p \times d$ factor loading matrix, $\textbf{Z}_{t}$ is a $d$-dimensional latent factor, and $\boldsymbol{\varepsilon}_{t}$ is a martingale difference sequence with respect to $\mathcal{F}_{t-1}=\sigma(\textbf{X}_{t-1},\textbf{X}_{t-2}, \dots),$ i.e., $E\{\boldsymbol{\varepsilon}_{t}|\mathcal{F}_{t-1}\}=0.$
If we consider the case where $\mathbf{Y}$ is unobserved in model \eqref{PFC-kenerls-mean}, the specification for $\mathbf{X}_t$ reduces to the FM-MDE described above.

In contrast to the PFC model in \eqref{PFC}, model \eqref{PFC-kenerls-mean} only requires the error term $\boldsymbol{\varepsilon}$ to satisfy $\mathrm{E}\{\boldsymbol{\varepsilon} \mid \mathbf{Y}\} = 0$, thereby relaxing the normality assumption in \eqref{PFC}. Furthermore,
the basis function $\mathbf{f}_{y}$ in model \eqref{PFC} is specified by the user in practice. \citet{CookFisher07} and \citet{PFC08} recommend choices such as
\begin{eqnarray}\label{fy-cook}
&&\!\!\!\mathbf{f}_{y}\!\!=\!\! (y, y^{2}, \ldots, y^{r} )^{T}  ~~\mathrm{or}   \nonumber \\
&&\!\!\!\mathbf{f}_{y}\!\!=\!\!(\cos (2 \pi y), \sin (2 \pi y), \ldots, \cos(2 \pi r y), \sin (2 \pi r y) )^{T},    \nonumber \\
\end{eqnarray}
for some positive integer $r.$
These  functions can be viewed as finite-dimensional approximations to feature maps associated with specific kernels in model \eqref{PFC-kenerls-mean}. For instance, the Gaussian kernel
\begin{eqnarray*}
	K( {y}_1- {y}_2)&=&\exp(-( {y}_1- {y}_2)^2/\gamma)\\
	&=&\langle\boldsymbol{\phi}_{G}( {y}_1), \boldsymbol{\phi}_{G}( {y}_2)\rangle_\mathcal{H}, ~\gamma>0,
\end{eqnarray*}
 corresponds to the feature map
\begin{eqnarray*}
\boldsymbol{\phi}_{G}(y)
\!=\! e^{-\frac{y^{2}}{2 \sigma^{2}}}( 1, \sqrt{\frac{1}{1 !}} \frac{y}{\sigma},   \cdots, \sqrt{\frac{1}{k !}} \frac{y^{k}}{\sigma^{k}}, \cdots)^T.
\end{eqnarray*}
Another example is the  kernel from \citet[Example 3.14]{Shawetaylor2005Kernel},
$$
K( {y}_1- {y}_2)=\sum_{k=0}^{\infty} a_{k} \cos (k(  {y}_1- {y}_2) ),
$$
which corresponds to an orthogonal feature map
\begin{eqnarray*}
	\boldsymbol{\phi}_{O}(y)\!\!=\!\! \big(1, \sin (y), \cos(y),  \ldots,\sin(k y), \cos(k y), \ldots\big)^T.
\end{eqnarray*}
 Thus, the basis functions $\mathbf{f}_{y}$ in \eqref{fy-cook} can be interpreted as finite truncations of the infinite-dimensional feature maps $\boldsymbol{\phi}_{G}(y)$ or $\boldsymbol{\phi}_{O}(y)$, establishing a natural link between the PFC framework and the proposed kernel-based formulation.

The preceding discussion clarifies that selecting the basis function $\mathbf{f}_{y}$ in model \eqref{PFC} corresponds to choosing a specific feature map in model \eqref{PFC-kenerls-mean}. The kernel-based methodology eliminates this explicit selection requirement while maintaining theoretical coherence. Furthermore, model \eqref{PFC-kenerls-mean} naturally extends to multivariate   $\mathbf{Y}$, overcoming a significant limitation of the  PFC framework.

\section{ICVI  matrix }\label{Extension}

As established in the classical SDR literature, first-order methods such as SIR fail to  recover $\mathcal S_{\textbf{Y}| \textbf{X}}$  when   the regression
surface is symmetric about zero. For further details, refer to Chapter 5 of
 \citet{li2018sufficient}. Our proposed ICMI methods, being based on the first conditional moment $\mathrm{E}\{\mathbf{X} \mid Y\}$, fall into this category and consequently inherit the same limitation. To address this   shortcoming, this section develops a new framework   based on the ICVI relationship in \eqref{equvlent-IV}. The  new framework draws inspiration from classical second-order methods, such as SAVE, pHd and DR,   thereby   overcoming the inherent weaknesses of the ICMI framework.

 To elucidate the rationale underlying our proposed methods, we revisit   the inverse regression framework in \eqref{PFC}. Under this model, the following relationships hold
\begin{eqnarray*}
\!\!&&\operatorname{Cov}(\mathbf{B}^{T}\mathbf{X}|Y)\!\!=\!\!  \boldsymbol{\beta} \mathbf{f}(Y) \mathbf{f}(Y)^T \boldsymbol{\beta}^T,  \\
\!\!&& \operatorname{Cov}(\mathbf{B}_{\perp}^{T}  \mathbf{X}|Y)\!\!=\!\!
\sigma^2\mathbf{B}_{\perp}^{T}\operatorname{Cov}(\boldsymbol{\varepsilon})\mathbf{B}_{\perp}
\!\!=\!\!\mathbf{B}_{\perp}^{T}\operatorname{Cov}(\mathbf{X})\mathbf{B}_{\perp}^{T}.
\end{eqnarray*}
  This implies that $\mathbf{B}^{T}\mathbf{X}$  is conditionally variance-dependent on
 $Y$, while  $\mathbf{B}_{\perp}^{T} \mathbf{X}$  is conditionally variance-invariant with respect to
 $Y$.
Instead of directly estimating $\mathbf{B}$,  we   focus on estimating its orthogonal complement
  $\mathbf{B}_{\perp},$   which satisfies
\begin{eqnarray*}
\operatorname{Cov}(\mathbf{B}_{\perp}^{T}  \mathbf{X}|Y)=\operatorname{Cov}( \mathbf{B}_{\perp}^{T} \mathbf{X}).
\end{eqnarray*}
That is,    the conditional variance of  $ \mathbf{B}_{\perp}^{T} \mathbf{X}$ given  $Y$ does not vary with $Y$.
Building on this relationship, we   consider the   ICVI framework through  \eqref{equvlent-IV}.
Let $\mathbf{Z}=\mathbf{X}-{E}\{\mathbf{X}\}$ and $\mathbf{Z}_i=\mathbf{X}_i-{E}\{\mathbf{X}\}.$
Following reasoning analogous to \eqref{equvlent-IM}, we obtain the equivalence
\begin{eqnarray}\label{equvlent-equ-IV}
 {E}\{\mathbf{Z}\mathbf{Z}^T| \mathbf{Y}\}=\boldsymbol{\Sigma}
 \Longleftrightarrow  E\{[\mathbf{Z}\mathbf{Z}^T-\boldsymbol{\Sigma}]w(\mathbf{Y}, \boldsymbol{\mathbf{\xi}})  \}=0,  \nonumber
\end{eqnarray}
for any $w(\cdot, \boldsymbol{\xi})\in \mathcal{H}$.

In parallel with the ICMI framework,    we consider two specific function classes for $\mathcal{H}$: the class of indicator functions $I(\mathbf{Y} \leq \boldsymbol{\xi})$ and the class of complex exponentials $\exp(i \boldsymbol{\xi}^\mathrm{T} \mathbf{Y})$.
  These choices lead to two   ICVI approaches for recovering  $\mathcal{S}_{\textbf{Y}| \textbf{X}}$.
 Specifically,  the projection-based approach yields the ICVI matrix
 \begin{eqnarray*}
\mathrm{ICVI.M}_\mathrm{PR}(\textbf{X}|\mathbf{Y})
\!\!&=&\!\! -\mathrm{E}\{ [\mathbf{Z}_1\mathbf{Z}_1^T-\boldsymbol{\Sigma}][\mathbf{Z}_2\mathbf{Z}_2^T-
\boldsymbol{\Sigma}]\\
&&\times  \mathrm{ang}(\mathbf{Y}_1- \mathbf{Y}_3,\mathbf{Y}_2- \mathbf{Y}_3)  \}\\
&=& {\boldsymbol{ \Lambda}}_\mathrm{V.PR},
\end{eqnarray*}
while the kernel-based approach gives
 \begin{eqnarray*}%\label{IMP}
\mathrm{ICVI.M}_\mathrm{K}(\textbf{X}|\mathbf{Y})
&=&\mathrm{E}\{ [\mathbf{Z}_1\mathbf{Z}_1^T-\boldsymbol{\Sigma}][\mathbf{Z}_2\mathbf{Z}_2^T-
\boldsymbol{\Sigma}]   \\
&&\times K(\mathbf{Y}_1-\mathbf{Y}_2) \}={\boldsymbol{ \Lambda}}_\mathrm{V.K}.
\end{eqnarray*}
 Both ${\boldsymbol{ \Lambda}}_\mathrm{V.PR}$ and ${\boldsymbol{ \Lambda}}_\mathrm{V.K}$  belong to second-order moment-based methods. Their theoretical properties are established in the following result.

\begin{thm}\label{them-kenel-HD}
Assume that the linearity condition   $E\{\mathbf{X} | \mathbf{\Gamma}^T \mathbf{X}\}=\mathbf{P}_{\mathbf{\Gamma}}^T(\mathbf{\Sigma}) \mathbf{X}$,
 and the constant variance condition
$\operatorname{Var}(\mathbf{X}| \mathbf{\Gamma}^T \mathbf{X})
=\mathbf{\Sigma}\{\textbf{I}-\mathbf{P}_{\mathbf{\Gamma}}\}$,  where  $\mathbf{P}_{\mathbf{\Gamma}}(\mathbf{\Sigma})=\mathbf{\Gamma}(\mathbf{\Gamma}^T \mathbf{\Sigma} \Gamma)^{-1} \mathbf{\Gamma}^T \boldsymbol{\Sigma}.$
Then, we have  $\mathcal{S}({\boldsymbol{ \Lambda}}_\mathrm{V.PR} )\subseteq \mathbf{\Sigma} \mathcal{S}_{\textbf{Y}| \textbf{X}}$ and $\mathcal{S}({\boldsymbol{ \Lambda}}_\mathrm{V.K}) \subseteq \mathbf{\Sigma} \mathcal{S}_{\textbf{Y}| \textbf{X}}.$
\end{thm}

This theorem confirms that both
  ${\boldsymbol{ \Lambda}}_\mathrm{V.PR}$ and ${\boldsymbol{ \Lambda}}_\mathrm{V.K}$ can be used to recover the central subspace. The constant variance condition adopted here is standard in second-order SDR methods, paralleling the routine use of the linearity condition in first-order inverse regression approaches \citep{li2018sufficient}.

Using moment-based estimation, the estimators for $\boldsymbol{\Lambda}_{\mathrm{V.PR}}$ and $\boldsymbol{\Lambda}_{\mathrm{V.K}}$ are   defined as
\begin{eqnarray*}
\widehat{ {\boldsymbol{ \Lambda}}}_\mathrm{V.PR} =- \frac{1}{n^3} \sum_{i,j,k=1}^n
(\mathbf{Z}_i\mathbf{Z}_i^T-\boldsymbol{\Sigma}_{n}) (\mathbf{Z}_j\mathbf{Z}_j^T-\boldsymbol{\Sigma}_{n})\\
\times\mathrm{ang}(\mathbf{Y}_i -\mathbf{Y}_k, \mathbf{Y}_j-\mathbf{Y}_k),
\end{eqnarray*}
and
\begin{equation*}
\widehat{ {\boldsymbol{ \Lambda}}}_\mathrm{V.K} =\frac{1}{n^2} \sum_{i,j=1}^n
(\mathbf{Z}_i\mathbf{Z}_i^T-\boldsymbol{\Sigma}_{n}) (\mathbf{Z}_j\mathbf{Z}_j^T-\boldsymbol{\Sigma}_{n})\mathrm{K}(\mathbf{Y}_i-\mathbf{Y}_j).
\end{equation*}
Following arguments similar to those for Theorems \ref{them-kenerl-consistency}-\ref{them-projection-consistency}, the asymptotic properties of $\widehat{\boldsymbol{\Lambda}}_{\mathrm{V.PR}}$ and $\widehat{\boldsymbol{\Lambda}}_{\mathrm{V.K}}$ can be established; we omit the details here.

\section{Simulation studies}\label{simulation}

This section evaluates the finite-sample performance of the proposed ICMI and ICVI methods against existing SDR  techniques  and MDDM. The comparison is conducted for both univariate and multivariate response scenarios. The competing SDR methods include:
\begin{description}
	\item[{For univariate responses:}]
  SIR \citep{Likc91}, CUME \citep{Zhu2010Dimension}, and SAVE \citep{cook91weisberg}.

\item[{For multivariate responses:}]
 SIR based on discretization-expectation estimation (SIR-DEE; \citet{Zhu2010Sufficient}), as well as SIR and SAVE using projective resampling (PR-SIR, PR-SAVE; \citet{li2008projective}).
 \end{description}
 In all subsequent simulations, SIR, SAVE, PR-SIR, and PR-SAVE are implemented with 5 slices.

We study the influence of the kernel function by employing the Gaussian kernel $K( \mathbf{Y}_1- \mathbf{Y}_2)=\exp(-\|\mathbf{Y}_1- \mathbf{Y}_2\|^2/\gamma)$, Laplace kernel $K(\mathbf{Y}_1-\mathbf{Y}_2)=\exp \left(- {\|\mathbf{Y}_1-\mathbf{Y}_2\|}/{\gamma}\right)$  and the rational quadratic kernel  $K(\mathbf{Y}_1-\mathbf{Y}_2)=
1-\frac{\|\mathbf{Y}_1-\mathbf{Y}_2\|^{2}}{\|\mathbf{Y}_1-\mathbf{Y}_2\|^{2}+\gamma}$  (with $\gamma > 0$) within our framework. This leads to the following methods in our simulations:
\begin{description}
	\item[{For ICMI:}] the projection-based method (IM$_{\mathrm{PR}}$) and the kernel-based methods (IM$_{\mathrm{Gauss}}$, IM$_{\mathrm{Lap}}$, IM$_{\mathrm{RQ}}$)

	\item[{For ICVI:}] the projection-based method (IV$_{\mathrm{PR}}$) and the kernel-based methods (IV$_{\mathrm{Gauss}}$, IV$_{\mathrm{Lap}}$, IV$_{\mathrm{RQ}}$).

\end{description}

Let $\mathbf{\Gamma}_0$ is a  $ p\times {d_0}$ matrixwhose columns form a basis for the central subspace $\mathcal{S}{\mathbf{Y}|\mathbf{X}}$, and let $\widehat{\mathbf{\Gamma}}$ be a corresponding estimator. The accuracy of $\widehat{\mathbf{\Gamma}}$ is evaluated using the distance metric
\begin{eqnarray*} \mathrm{dist}(\widehat{\mathbf{\Gamma}},\mathbf{\Gamma}_0)
=\|\mathrm{P}_{\widehat{\mathbf{\Gamma}}}-\mathrm{P}_{\mathbf{\Gamma}_0}\|_F,
\end{eqnarray*}
where $\|\cdot\|_\mathrm{F}$ is the Frobenius norm, and $\mathrm{P}_{\mathbf{\Gamma}}=\mathbf{\Gamma}(\mathbf{\Gamma}^T \mathbf{\Gamma})^{-1}\mathbf{\Gamma}^T$  is the projection matrix for any matrix $\mathbf{\Gamma}$.
 A smaller value of $\mathrm{dist}(\widehat{\mathbf{\Gamma}},\mathbf{\Gamma}_0)$ indicates a more accurate estimate. For each configuration, we report the mean and standard deviation of this distance over 500 replications.

\begin{example}[Univariate response]\label{exampl_1}
We consider two models for a univariate response $Y,$ given by
	\begin{description}
		\item[Case (i):]  $Y= \boldsymbol{\beta}_1^T\textbf{X}  + \varepsilon;$
		\item[Case (ii):]  $Y= \sin(\boldsymbol{\beta}_2^T\textbf{X})  + \exp(\boldsymbol{\beta}_3^T\textbf{X})\varepsilon,$
	\end{description}
	where $\boldsymbol{\beta}_1 =(1, 1, 1,1,0,\cdots, 0)^T/2,$   $\boldsymbol{\beta}_2 = (1,0,\cdots,0)^T$, and $\boldsymbol{\beta}_3 = (0,1,0,\cdots,0)^T.$   Following    \citet{Zhu2010Dimension},  we let the dimension $p$ diverge with the sample size $n$ as $p = \lfloor \sqrt{n} \rfloor - 5$. The predictor vector $\mathbf{X}$ is generated from a $p$-variate standard normal distribution. The error term $\varepsilon$ is generated from three different distributions:
	(1) $\varepsilon\thicksim  N(0, 1)$; (2) $\varepsilon\thicksim Cauchy(0, 1);$
	  (3)  a normal mixture: $\varepsilon\thicksim0.7N(0, 1) + 0.3N(0, 10)$.
\end{example}

To appreciate the impact of the parameter $\gamma$ on the kernel-based estimators, we consider the following    functions:
$\phi_{{\mathrm{Guass},\gamma}}(t)=1-\exp(-t^{2} / \gamma ) \text{ and } \phi_{{\mathrm{Lap},\gamma}}(t)=1-\exp  (-|t| / \gamma ).
$
It is easy to show that
\begin{eqnarray*}
	\mathbf{\Gamma}_0&=& \mathop{\arg\max}\limits_{\boldsymbol{\mathbf{\Gamma}^T \mathbf{\Gamma}=\textbf{I}_d } }  \{\mathbf{\Gamma}^T\mathrm{E}\{ ( \mathbf{X}_1-E\{ \mathbf{X}_1\}) \\
	&&\times(\mathbf{X}_2-E\{ \mathbf{X}_2\})^T K_{{\mathrm{Guass}}}(\mathbf{Y}_1-\mathbf{Y}_2) \}\mathbf{\Gamma} \}\\
	&=& \mathop{\arg\min}\limits_{\boldsymbol{\mathbf{\Gamma}^T \mathbf{\Gamma}=\textbf{I}_d } } \{\mathbf{\Gamma}^T\mathrm{E}\{ ( \mathbf{X}_1-E\{ \mathbf{X}_1\})\\
	&&\times (\mathbf{X}_2-E\{ \mathbf{X}_2\})^T \phi_{{\mathrm{Guass},\gamma}}(\|\mathbf{Y}_1- \mathbf{Y}_2\|) \}\mathbf{\Gamma} \}.
\end{eqnarray*}
where $K_{\mathrm{Gauss}}(t) = \exp(-t^2 / \gamma)$. The same results hold when $\phi_{\mathrm{Gauss},\gamma}(t)$ and $K_{\mathrm{Gauss}}(t)$ are replaced by $\phi_{\mathrm{Lap},\gamma}(t)$ and $K_{\mathrm{Lap}}(t) = \exp(-t / \gamma)$, respectively.
Note that for large $\gamma$,
 $$\phi_{{\mathrm{Guass},\gamma}}(t) \approx t^{2} / \gamma ~~ \text{ and }~~  \phi_{{\mathrm{Lap},\gamma}}(t) \approx |t|  / \gamma.  $$
 Thus, for large $\gamma$, $\mathrm{IM}_{\mathrm{Gauss}}$ and $\mathrm{IM}_{\mathrm{Lap}}$ behave similarly to OLS \citep{li1989regression} (for univariate $Y$) and MDDM, respectively. This implies that a larger value of $\gamma$ is generally preferable for light-tailed data, whereas a smaller $\gamma$ offers greater robustness for heavy-tailed distributions or in the presence of outliers.

Table \ref{tab_1-1} presents the sample means and standard deviations of $\mathrm{dist}(\widehat{\mathbf{\Gamma}},\mathbf{\Gamma}0)$ for seven first-order SDR methods: SIR, CUME,  MDDM, IM$_{\mathrm{PR}}$, IM$_{\mathrm{Guass}}$, IM$_{\mathrm{Lap}}$ and IM$_{\mathrm{RQ}}$, for  $(n,p) = (225,10)$, $(400,15)$, and $(625,20)$.

 As evidenced in Table \ref{tab_1-1},  the kernel-based ICMI methods (IM$_{\mathrm{Gauss}}$, IM$_{\mathrm{Lap}}$, IM$_{\mathrm{RQ}}$) outperform the other approaches across all scenarios under the given $\gamma$ values. The results   also reveals    that MDDM performs competitively, slightly surpassing SIR and IM$_{\mathrm{PR}}$ under normal errors. However, its performance deteriorates markedly under heavy-tailed error distributions, exhibiting higher bias and variance due to its reliance on the Euclidean distance $\|\mathbf{Y}_1 - \mathbf{Y}_2\|$, which is sensitive to outliers.  Furthermore, the equivalence between IM$_{\mathrm{PR}}$ and CUME for univariate responses, as established in Proposition \ref{PR_CUME}, is confirmed by the Case (i) results; thus, CUME results for Case (ii) are omitted.

 Fig.\ref{fig:kernel-1} reveals the impact of $\gamma$ on the MSE for Example \ref{exampl_1}(i) with $n=400$ and $p=15$.   As expected, the MSEs achieve their minimum at smaller values of $\gamma$ in the presence of heteroscedastic or non-normal errors. In contrast, for normal and homoscedastic errors, the minimum MSE occurs at a larger $\gamma$.

\begin{table*}[hptb]    
	\caption{The mean (standard deviation) of $\mathrm{dist}(\widehat{\mathbf{\Gamma}},\mathbf{\Gamma}_0)$
		for Example  \ref{exampl_1}.\label{tab_1-1}}
		\begin{tabular*}{\textwidth}{@{\extracolsep\fill}cclccc} \toprule
			\multicolumn{1}{c}{Settings}&\multicolumn{1}{c}{Errors} &\multicolumn{1}{c}{Methods}
			&\multicolumn{1}{c}{$\begin{array}{c}n=225  \vspace{-0.8ex}\\ (p=10)\end{array} $ }
			&\multicolumn{1}{c}{ $\begin{array}{c}n=400 \vspace{-0.8ex}\\ (p=15)\end{array} $}
			&\multicolumn{1}{c}{$\begin{array}{c}n=625 \vspace{-0.8ex}\\ (p=20)\end{array} $ }\\
			\midrule
			&$N(0, 1)$  &SIR                  &0.312 (0.072)  &0.291 (0.059) &0.271 (0.044)        \\
			&           &CUME                 &0.304 (0.067)  &0.284 (0.057) &0.265 (0.042)         \\                           	&           &MDDM                 &0.291 (0.065)  &0.273 (0.055) &0.255 (0.041)         \\[0.15cm]
			&           &IM$_{\mathrm{PR}}$   &0.304 (0.067)  &0.284 (0.057) &0.265 (0.042)  \\
			&$\gamma=40$&IM$_{\mathrm{Guass}}$&\textbf{0.281 (0.065)}  &\textbf{0.264 (0.053)}&\textbf{0.247 (0.040)} \\
			&           &IM$_{\mathrm{Lap}}$  &0.292 (0.066)  &0.274 (0.055) &0.256 (0.041)  \\
			&           &IM$_{\mathrm{RQ}}$   &0.284 (0.065)  &0.266 (0.053) &0.249 (0.040)  \\
			\cmidrule{2-6}
			\textbf{Case (i)}
			&Cauchy     &SIR                  &0.550 (0.134) &0.512 (0.096)     &0.472 (0.078)        \\
			&           &CUME                 &0.531 (0.133) &0.499 (0.094)     &0.462 (0.077)         \\
			&           &MDDM                 &0.748 (0.312) &0.712 (0.297)     &0.666 (0.296)         \\[0.15cm]
			&           &IM$_{\mathrm{PR}}$   &0.531 (0.133) &0.499 (0.094)     &0.462 (0.077)  \\
			&$\gamma=5 $&IM$_{\mathrm{Guass}}$&0.501 (0.120) &0.471 (0.088)     &0.435 (0.073) \\
			&           &IM$_{\mathrm{Lap}}$  &0.494 (0.121) &\textbf{0.464 (0.087)}     &0.429 (0.073)  \\
			&           &IM$_{\mathrm{RQ}}$   &\textbf{0.494 (0.119)} &0.465 (0.087)    &\textbf{0.429(0.072)}   \\
			\cmidrule{2-6}
			
			&Mixnormal  &SIR                  &0.550 (0.136)  &0.516 (0.099)   &0.484 (0.077)        \\
			&           &CUME                 &0.537 (0.137)  &0.509 (0.100)   &0.476 (0.078)         \\
			&           &MDDM                 &0.627 (0.174)  &0.588 (0.127)   &0.550 (0.097)         \\[0.15cm]
			&           &IM$_{\mathrm{PR}}$   &0.537 (0.137)  &0.509 (0.100)   &0.476 (0.078)  \\
			&$\gamma=5$ &IM$_{\mathrm{Guass}}$&0.460 (0.113)  &0.438 (0.084)   &0.405 (0.066) \\
			&           &IM$_{\mathrm{Lap}}$  &0.472 (0.119)  &0.447 (0.089)   &0.414 (0.066)  \\
			&           &IM$_{\mathrm{RQ}}$   &\textbf{0.460 (0.113)}  &\textbf{0.437 (0.085)}  &\textbf{0.405 (0.065) } \\    \hline			
			&$N(0, 1)$  &SIR                   &0.692 (0.132) &0.661 (0.102)   &0.620 (0.080)  \\
			&           &MDDM                  &0.632 (0.125) &0.613 (0.098)   &0.575 (0.078)   \\[0.15cm]
			&           &IM$_{\mathrm{PR}}$    &0.646 (0.119) &0.626 (0.096)   &0.588 (0.076)  \\
			&$\gamma=2$ &IM$_{\mathrm{Guass}}$ &0.622 (0.115) &0.598 (0.094)   &0.563 (0.072)   \\
			&           &IM$_{\mathrm{Lap}}$   &\textbf{0.602 (0.116)} &\textbf{0.578 (0.093)}  &\textbf{0.543 (0.070)}   \\
			&           &IM$_{\mathrm{RQ}}$    &0.607 (0.115) &0.584 (0.093)   &0.549 (0.071)   \\
			\cmidrule{2-6}
			
			\textbf{Case (ii)}
			&Cauchy     &SIR                   &0.832 (0.160)  &0.786 (0.113)   &0.737 (0.096)        \\
			&           &MDDM                  &1.220 (0.252)  &1.222 (0.240)   &1.204 (0.254)         \\[0.15cm]
			&           &IM$_{\mathrm{PR}}$    &0.845 (0.160)  &0.817 (0.114)   &0.773 (0.101)  \\
			&$\gamma=2 $&IM$_{\mathrm{Guass}}$ &0.771 (0.141)  &0.741 (0.109)   &0.705 (0.091) \\
			&           &IM$_{\mathrm{Lap}}$   &\textbf{0.759 (0.141)}  &\textbf{0.730 (0.105)}  &\textbf{0.690 (0.090)}   \\
			&           &IM$_{\mathrm{RQ}}$    &0.760 (0.140)  &0.731 (0.106)   &0.693 (0.091) \\
			\cmidrule{2-6}
			
			&Mixnormal  &SIR                   &0.845 (0.166)      &0.818 (0.124)  &0.755 (0.095)     \\
			&           &MDDM                  &1.074 (0.224)      &1.058 (0.191)  &0.996 (0.169)         \\[0.15cm]
			&           &IM$_{\mathrm{PR}}$    &0.866 (0.165)      &0.845 (0.128)  &0.798 (0.103)  \\
			&$\gamma=2$ &IM$_{\mathrm{Guass}}$ &0.782 (0.144)      &0.768 (0.112)  &0.721 (0.088) \\
			&           &IM$_{\mathrm{Lap}}$   &0.771 (0.145)      &0.756 (0.111)  &0.711 (0.089)   \\
			&           &IM$_{\mathrm{RQ}}$    &\textbf{0.771 (0.144)}      &\textbf{0.756 (0.111)}&\textbf{0.711 (0.088)} \\
			\botrule
\multicolumn{5}{l}{Note: The smallest mean values are shown in bold. }
		\end{tabular*}
\end{table*}

\begin{figure*}
	\centering
	\includegraphics[width=5.5cm,height=5cm]{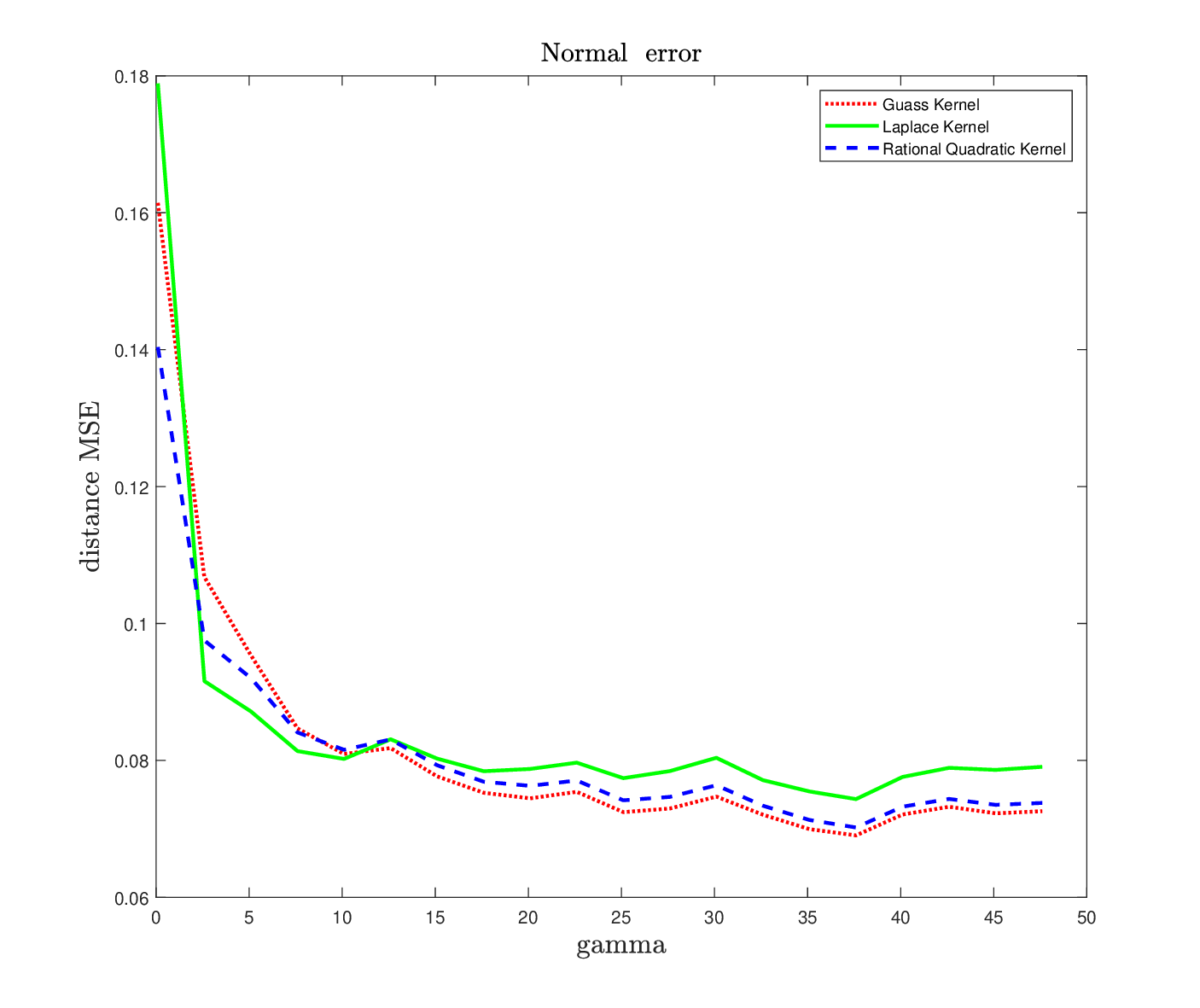}\hspace{-0.4cm}
	\includegraphics[width=5.5cm,height=5cm]{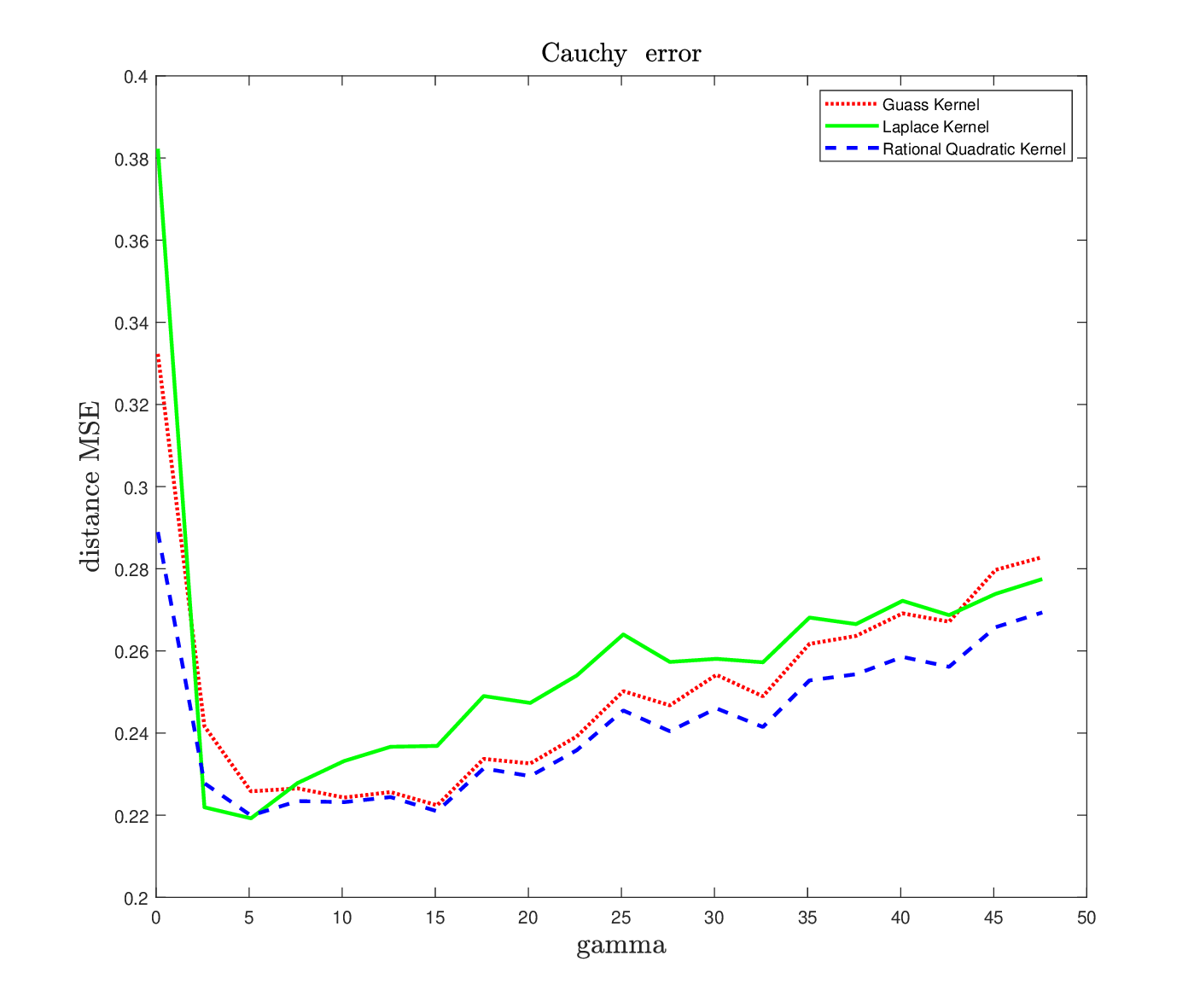}\hspace{-0.4cm}
	\includegraphics[width=5.05cm,height=5cm]{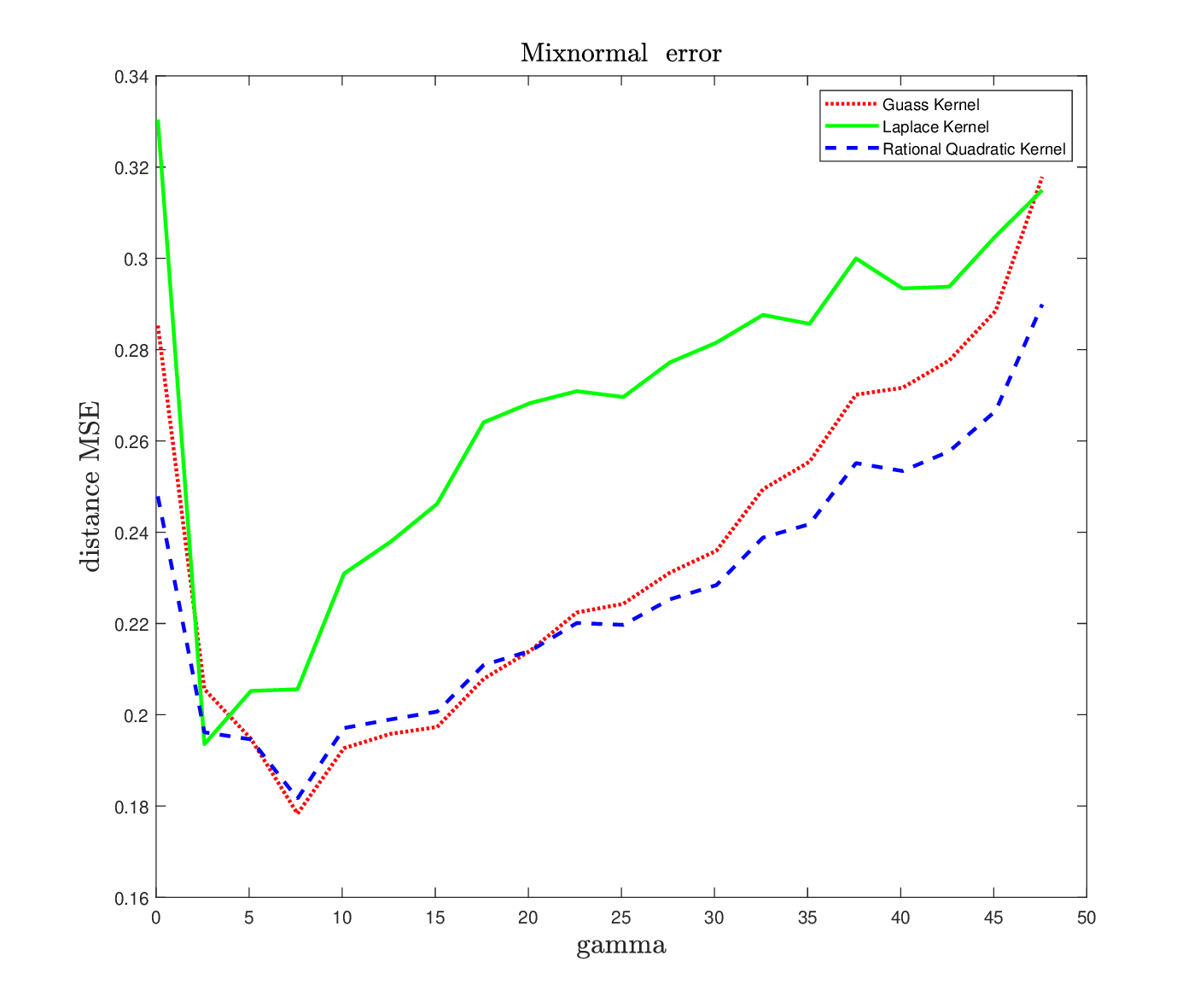}
	\caption{MSE of the three kernel-based methods for Example  \ref{exampl_1} Case (i).
	\label{fig:kernel-1}}
\end{figure*}
%\begin{figure*}
%	\centering
%	\includegraphics[width=5.5cm,height=5cm]{Figure/normal_ken_ii.eps}\hspace{-0.4cm}
%	\includegraphics[width=5.5cm,height=5cm]{Figure/Cauchy_ken_ii.eps}\hspace{-0.4cm}
%	\includegraphics[width=5.5cm,height=5cm]{Figure/mixture_ken_ii.eps}
%	\caption{MSE of three kernel-based methods for Example  \ref{exampl_1} Case (ii).
%	\label{fig:kernel-2}}
%\end{figure*}

\begin{example}[Multivariate response]\label{exampl_2}
This example evaluates the finite-sample performance of the proposed methods for a multivariate response. We consider the model from Example 1 in \citet{li2008projective}, given by
 $$
	Y_{1}=\boldsymbol{\beta}_{1}^{T} \mathbf{X}+\varepsilon_{1},
	\quad Y_{2}=\boldsymbol{\beta}_{2}^{T} \mathbf{X}+\varepsilon_{2},
	\quad Y_{3}=\varepsilon_{3}, \quad Y_{4}=\varepsilon_{4},$$
	where $p=6, q=4$,
	$\boldsymbol{\beta}_{1}=(1,0,0,0,0,0)^{T}$, $\boldsymbol{\beta}_{2}=(0,2,1,0,0,0)^{T}$ and  $\mathbf{X} \sim \mathrm{N}(\mathbf{0}, \mathbf{I}_{6}).$
	The error vector  $\boldsymbol{\varepsilon}=(\varepsilon_{1},\cdots, \varepsilon_{4})^T$ is generated from three different distributions:
 \begin{description}
		\item[Case (i):] $\boldsymbol{\varepsilon} \sim \mathrm{N}_{4}(\mathbf{0}, \mathbf{\Delta})$, where $\mathbf{\Delta} = \operatorname{diag}(\mathbf{\Delta}_{1}, \mathbf{I}_{2})$ and
        \[
        \mathbf{\Delta}_{1} = \begin{pmatrix}
            1 & -0.5 \\
            -0.5 & 1
        \end{pmatrix};
        \]
        \item[Case (ii):] $\boldsymbol{\varepsilon} \sim t_1(\mathbf{0}, \mathbf{\Delta})$ (multivariate $t$-distribution);
        \item[Case (iii):] a normal mixture: $\varepsilon_j \stackrel{\text{i.i.d.}}{\sim} 0.7 N(0, 1) + 0.3 N(0, 10)$ for $j=1,\ldots,4$.
    \end{description}
\end{example}

The simulation results for Example \ref{exampl_2} with   $n = 200, 400,$ and $800$ are summarized in Table \ref{tab_2-1}, where the PR-SIR method  is implemented with 1000   resamples.
 The results indicate that under normal errors, our four methods perform comparably to PR-SIR, though both are slightly less accurate than MDDM. For non-Gaussian errors, however, our methods perform significantly better than MDDM and also outperform PR-SIR in most scenarios.  These results demonstrate that the proposed methods provide reasonably accurate estimates for multivariate response vectors.

\begin{table*}[hptb]
	\caption{The mean (standard deviation) of  $\mathrm{dist}(\widehat{\mathbf{\Gamma}},\mathbf{\Gamma}_0)$
		for Example  \ref{exampl_2}.\label{tab_2-1}}
		\begin{tabular*}{\textwidth}{@{\extracolsep\fill}clccc} \toprule
			\multicolumn{1}{c}{Errors} &\multicolumn{1}{c}{Methods}
			&\multicolumn{1}{c}{$n=200$ }
			&\multicolumn{1}{c}{$n=400$}
			&\multicolumn{1}{c}{$n=800$}\\
			\midrule
			$\mathrm{N}_{4}(\mathbf{0}, \mathbf{\Delta})$
			& SIR-DEE             &0.357 (0.097) &0.238 (0.055) &0.171 (0.045)        \\
			& PR-SIR              &0.235 (0.083) &0.153 (0.040) &0.110 (0.040)         \\                       & MDDM                &\textbf{0.230 (0.080)} &\textbf{0.149 (0.042)} &\textbf{0.110 (0.039)} \\[0.15cm]
			&IM$_{\mathrm{PR}}$   &0.234 (0.082) &0.153 (0.043) &0.113 (0.040)  \\
			$\gamma=40$&IM$_{\mathrm{Guass}}$&0.236 (0.081) &0.154 (0.042) &0.112 (0.039) \\
			&IM$_{\mathrm{Lap}}$  &0.232 (0.080) &0.151 (0.042) &0.111 (0.039)  \\
			&IM$_{\mathrm{RQ}}$   &0.238 (0.081) &0.155 (0.042) &0.113 (0.039)  \\    \hline

			${t}_1(\mathbf{0}, \mathbf{\Delta})$
			&SIR-DEE              &0.927 (0.325)    &0.740 (0.274)      &0.500 (0.202)        \\
			&PR-SIR               &0.430 (0.153)    &0.286 (0.097)      &0.197 (0.057)         \\
			&MDDM                 &0.795 (0.357)    &0.545 (0.335)      &0.416 (0.353)         \\[0.15cm]
			&IM$_{\mathrm{PR}}$   &0.411 (0.137)    &0.285 (0.093)      &0.200 (0.060)  \\
			$\gamma=5 $&IM$_{\mathrm{Guass}}$&0.436 (0.117)    &0.309 (0.091)      &0.210 (0.060) \\
			&IM$_{\mathrm{Lap}}$  &\textbf{0.365 (0.129)}    &\textbf{0.258 (0.089)}     &\textbf{0.176 (0.048)}  \\
			&IM$_{\mathrm{RQ}}$   &0.378 (0.111)    &0.271 (0.087)      &0.183 (0.052) \\    \hline
			
			Mixnormal
			&SIR-DEE              &0.794  (0.254)        &0.534 (0.164)         &0.360 (0.098)       \\
			&PR-SIR               &0.747  (0.263)        &0.466 (0.153)         &0.317 (0.091)        \\
			&MDDM                 &0.871  (0.275)        &0.592 (0.196)         &0.408 (0.123)        \\[0.15cm]
			&IM$_{\mathrm{PR}}$   &0.724  (0.229)        &0.497 (0.154)         &0.353 (0.103)  \\
			$\gamma=5$ &IM$_{\mathrm{Guass}}$&0.741  (0.226)        &0.488 (0.138)         &0.329 (0.085) \\
			&IM$_{\mathrm{Lap}}$  &0.701  (0.266)        &0.436 (0.138)         &0.295(0.084)        \\
			&IM$_{\mathrm{RQ}}$   &\textbf{0.674  (0.232)}        &\textbf{ 0.431 (0.125)}        &\textbf{0.289 (0.079)}  \\    \botrule			
		\end{tabular*}
\end{table*}

\begin{example}[Multivariate response with heteroscedastic errors]\label{exampl_3-3}
This example considers multivariate response models with heteroscedastic errors, adapted from Examples 4.2--4.3 in \citet{li2008projective}, given by	
	\begin{description}
		\item[Case (i):] $ Y_{1}=2\exp \left(\varepsilon_{1}\right),   Y_{i}=\varepsilon_{i},   i=2,3,4;$
		\item[Case (ii):]  $Y_{1}= \varepsilon_{1},   Y_{i}=\varepsilon_{i},   i=2,3,4,$
	\end{description}
	where  $(\varepsilon_{1}, \ldots, \varepsilon_{4})^{T}$  follows from $\mathrm{N}_4(\mathbf{0}, \mathbf{\Delta}) $ with $\mathbf{\Delta}=$
	$\operatorname{diag}\left(\boldsymbol{\Delta}_{1}, \mathbf{I}_{2}\right)$ and
	$$
	\boldsymbol{\Delta}_{1}=\left(\begin{array}{cc}
		1 & \sin \left(\boldsymbol{\beta}^{T} \mathbf{X}\right) \\
		\sin \left(\boldsymbol{\beta}^{T} \mathbf{X}\right) & 1
	\end{array}\right).
	$$
	Here, $\mathbf{X} \sim \mathrm{N}(\mathbf{0}, \mathbf{I}_{6})$ and  $\boldsymbol{\beta}=(0.8, 0.6,0,0,0,0)^{T}.$
\end{example}

For both cases in Example \ref{exampl_3-3},   $E\{\mathbf{Y} \mid \mathbf{X}\}$ is non-random,   $\mathrm{var}(\mathbf{Y} \mid \mathbf{X}) = \mathrm{var}(\mathbf{Y} \mid \boldsymbol{\beta}^{T}\mathbf{X})$, and   $\mathcal{S}_{\mathbf{Y} \mid \mathbf{X}} = \mathcal{S}(\boldsymbol{\beta})$. For detailed derivations, see \citet{li2008projective}. The results in Table \ref{tab_3-1} indicate that the kernel-based methods IM$_{\mathrm{Gauss}}$, IM$_{\mathrm{Lap}}$, and IM$_{\mathrm{RQ}}$  achieve superior performance in most cases, followed by SIR-DEE and PR-SIR. In contrast, both MDDM and IM$_{\mathrm{PR}}$ exhibit inferior performance. The inferior accuracy  of MDDM can be attributed to its sensitivity to outliers, while the reason for the poor performance of IM$_{\mathrm{PR}}$ remains unclear and warrants further study.

\begin{table*}[hptb]
	\caption{The mean (standard deviation) of $\mathrm{dist}(\widehat{\mathbf{\Gamma}},\mathbf{\Gamma}_0)$
		for Example  \ref{exampl_3-3}.\label{tab_3-1}}
		\begin{tabular*}{\textwidth}{@{\extracolsep\fill}clccc} \toprule
			\multicolumn{1}{c}{Settings} &\multicolumn{1}{c}{Methods}
			&\multicolumn{1}{c}{$n=200$ }
			&\multicolumn{1}{c}{$n=400$}
			&\multicolumn{1}{c}{$n=800$}\\
			\midrule
			\textbf{Case (i)}
			&SIR-DEE              &0.594 (0.197) &0.385 (0.124) &0.260 (0.083)      \\
			&PR-SIR               &0.697 (0.294) &0.406 (0.168) &0.244 (0.086)    \\                           	
			&MDDM                 &1.126 (0.271) &0.884 (0.364) &0.506 (0.292)    \\[0.15cm]
			&IM$_{\mathrm{PR}}$   &1.062 (0.291) &0.794 (0.331) &0.453 (0.229) \\
			$\gamma=1$           &IM$_{\mathrm{Guass}}$&0.591 (0.214) &0.379 (0.130) &0.249 (0.083)\\
			&IM$_{\mathrm{Lap}}$  &\textbf{0.562 (0.217)} &\textbf{0.355 (0.120)} &\textbf{0.231 (0.076)}  \\
			&IM$_{\mathrm{RQ}}$   &0.568 (0.220) &0.357 (0.121) &0.232 (0.077) \\    \hline
			
			\textbf{Case (ii)}
			&SIR-DEE                  &0.582 (0.207)    &0.384 (0.129)     &0.271 (0.084)       \\
			&PR-SIR                   &0.636 (0.289)    &0.351 (0.136)     &0.213 (0.071)        \\                           	
			&MDDM                     &0.991 (0.321)    &0.625 (0.300)     &0.309 (0.124)      \\[0.15cm]
			&IM$_{\mathrm{PR}}$       &1.111 (0.278)    &0.823 (0.334)     &0.432 (0.219)  \\
			$\gamma=1$           &IM$_{\mathrm{Guass}}$    &0.596 (0.250)    &0.341 (0.121)     &0.218 (0.072)\\
			&IM$_{\mathrm{Lap}}$      &\textbf{0.567 (0.253)}    &\textbf{0.320 (0.118)}     &\textbf{0.203 (0.068)} \\
			&IM$_{\mathrm{RQ}}$       &0.573 (0.255)    &0.323 (0.119)     &0.204 (0.068)  \\
            \botrule
		\end{tabular*}
\end{table*}

\begin{example}[Univariate response  with zero-symmetric surface]\label{exampl_4-1}
	This example considers a univariate response model  where  the regression surface is symmetric about zero.   This model is  from    \citet{Bing07}, given by
	\begin{eqnarray*}
		Y=0.4(\boldsymbol{\beta}_{1}^T \mathbf{X})^{2}+|\boldsymbol{\beta}_{2}^{\mathrm{T}} \mathbf{X}|^{1/2}+0.4 \epsilon,
	\end{eqnarray*}
	where    $\boldsymbol{\beta}_{1}$ and $\boldsymbol{\beta}_{2}$ are $p$-dimensional vectors. The first six components of $\boldsymbol{\beta}_{1}$ and $\boldsymbol{\beta}_{2}$ are $(1, 1, 1, 0, 0, 0)^T$ and $(1, 0, 0, 0, 1, 3)^T$, respectively, and the remaining components are zero.  As in Example \ref{exampl_1}, we set $p = \lfloor \sqrt{n} \rfloor - 5$, and let $\mathbf{X}$ follow the multivariate standard normal distribution.
\end{example}

In Table  \ref{tab_4-1}, we compare nine   methods: the first-order methods--SIR, MDDM, IM$_{\mathrm{PR}}$, IM$_{\mathrm{Guass}}$--and the second-order methods--SAVE, IV$_{\mathrm{PR}}$, IV$_{\mathrm{Guass}}$, IV$_{\mathrm{Lap}}$, IV$_{\mathrm{RQ}}$. The comparisons are conducted under two error distributions: $\varepsilon \sim N(0, 1)$ and $\varepsilon \sim \text{Cauchy}(0, 1)$.
   Since   IM$_{\mathrm{Lap}}$ and IM$_{\mathrm{QR}}$ perform similarly to IM$_{\mathrm{Guass}}$,   they are omitted.

As shown in Table \ref{tab_4-1},
 all first-order methods fail, as expected, to recover $\mathcal{S}_{Y|\mathbf{X}}$ due to the symmetry of the regression surface about the origin.  In contrast, all second-order methods successfully estimate the central subspace. Notably, while the simulation setting is favorable for SAVE, IV$_{\mathrm{PR}}$ and the kernel-based ICVI methods consistently outperform it, with IV$_{\mathrm{PR}}$ being the top performer.

\begin{table*}[hptb]
	\caption{The mean (standard deviation) of                                      $\mathrm{dist}(\widehat{\mathbf{\Gamma}},\mathbf{\Gamma}_0)$
		for Example  \ref{exampl_4-1}.\label{tab_4-1}}
		\begin{tabular*}{\textwidth}{@{\extracolsep\fill}clccc} \toprule
			\multicolumn{1}{c}{Errors} &\multicolumn{1}{c}{Methods}
			&\multicolumn{1}{c}{$\begin{array}{c}n=225  \vspace{-0.8ex}\\ (p=10)\end{array} $ }
			&\multicolumn{1}{c}{ $\begin{array}{c}n=400 \vspace{-0.8ex}\\ (p=15)\end{array} $}
			&\multicolumn{1}{c}{$\begin{array}{c}n=625 \vspace{-0.8ex}\\ (p=20)\end{array} $ }\\
			\midrule
			$N(0, 1)$  &SIR                     &1.777  (0.137)     &1.859 (0.102)     &1.892 (0.074)        \\
			&MDDM                    &1.656  (0.168)     &1.781 (0.138)     &1.831 (0.106)        \\
			&IM$_{\mathrm{PR}}$      &1.759  (0.141)     &1.856 (0.099)     &1.890 (0.076)         \\
			$\gamma=2$ &IM$_{\mathrm{Guass}}$   &1.708  (0.154)     &1.830 (0.108)     &1.868 (0.086)     \\  \vspace{-1.5mm}\\
			&SAVE                    &0.631  (0.185)     &0.545 (0.103)     &0.490 (0.075)        \\
			&IV$_{\mathrm{PR}}$      &\textbf{0.460  (0.127)}     &\textbf{0.434 (0.078)} &\textbf{0.401 (0.058)}  \\
			$\gamma=2$ &IV$_{\mathrm{Guass}}$   &0.520  (0.129)     &0.485 (0.085)     &0.443 (0.064) \\
			&IV$_{\mathrm{Lap}}$     &0.526  (0.132)     &0.490 (0.089)     &0.444 (0.066)  \\
			&IV$_{\mathrm{RQ}}$      &0.519  (0.130)     &0.485 (0.087)     &0.441 (0.065)  \\                              \hline

			Cauchy     &SIR                  &1.794 (0.130)          &1.856  (0.097)             &1.898 (0.072)     \\
			&MDDM                 &1.744 (0.139)          &1.831  (0.104)             &1.871 (0.082)        \\
			&IM$_{\mathrm{PR}}$   &1.782 (0.137)          &1.855  (0.099)             &1.892 (0.071)         \\
			$\gamma=2$ &IM$_{\mathrm{Guass}}$&1.749 (0.142)          &1.835  (0.105)             &1.888 (0.072)    \\ \vspace{-1.5mm} \\
			&SAVE                 &0.959 (0.269)          &0.867  (0.232)             &0.746 (0.160)       \\
			&IV$_{\mathrm{PR}}$   &\textbf{0.702 (0.233)}          &\textbf{0.655(0.174)}         &\textbf{0.585 (0.109)}  \\
			$\gamma=2$ &IV$_{\mathrm{Guass}}$&0.755 (0.230)          &0.690 (0.166)             &0.620 (0.108)\\
			&IV$_{\mathrm{Lap}}$  &0.765 (0.235)          &0.699 (0.181)             &0.619 (0.113) \\
			&IV$_{\mathrm{RQ}}$   &0.754 (0.234)          &0.688 (0.172)             &0.615 (0.111) \\
             \botrule
		\end{tabular*}
\end{table*}

\begin{example}[Multivariate response with zero-symmetric surface]\label{exampl_4-2}
	We consider a multivariate response model with zero-symmetric regression surfaces.
  The model is the same as  Example  \ref{exampl_2} except that   $Y_1$ and $Y_2$ are replaced by
	\begin{eqnarray*}
		Y_{1}=(\boldsymbol{\beta}_{1}^{T} \mathbf{X})^2+\varepsilon_{1},    ~~ Y_{2}=|\boldsymbol{\beta}_{2}^{T} \mathbf{X}|+\varepsilon_{2}.
	\end{eqnarray*}
\end{example}

For Example \ref{exampl_4-2}, we again compare the nine methods across sample sizes $n=200, 400$, and 800. As presented in Table \ref{tab_4-2}, the performance of IV$_{\mathrm{PR}}$ is     better than that of IV$_{\mathrm{Gauss}}$, IV$_{\mathrm{Lap}}$, and IV$_{\mathrm{RQ}}$, and all four methods achieve comparable performance to PR-SAVE.

\begin{table*}[hptb]
	\caption{The mean (standard deviation) of $\mathrm{dist}(\widehat{\mathbf{\Gamma}},\mathbf{\Gamma}_0)$
		for Example  \ref{exampl_4-2}.\label{tab_4-2}}
		\begin{tabular*}{\textwidth}{@{\extracolsep\fill}clccc}
			 \toprule
			\multicolumn{1}{c}{Errors} &\multicolumn{1}{c}{Methods}
			&\multicolumn{1}{c}{$n=200$ }
			&\multicolumn{1}{c}{$n=400$}
			&\multicolumn{1}{c}{$n=800$}\\
			\midrule
			$N(0, 1)$  &SIR-DEE                &1.433 (0.248)      &1.271 (0.270)     &1.085 (0.303)      \\
			&MDDM                   &1.512 (0.216)      &1.495 (0.235)     &1.489 (0.230)        \\
			&IM$_{\mathrm{PR}}$     &1.589 (0.197)      &1.583 (0.212)     &1.568 (0.209)         \\
			$\gamma=30$&IM$_{\mathrm{Guass}}$  &1.513 (0.212)      &1.499 (0.234)     &1.492 (0.225)     \\  \vspace{-1.5mm} \\
			&PR-SAVE                &0.335 (0.089)      &0.220 (0.061)     &0.152 (0.040)        \\
			&IV$_{\mathrm{PR}}$     &\textbf{0.296 (0.079)}      &\textbf{0.199 (0.056)}  &\textbf{0.139 (0.038)}  \\
			$\gamma=30$&IV$_{\mathrm{Guass}}$  &0.319 (0.085)      &0.215 (0.060)     &0.149 (0.040) \\
			&IV$_{\mathrm{Lap}}$    &0.325 (0.087)      &0.217 (0.062)     &0.151 (0.040)  \\
			&IV$_{\mathrm{RQ}}$     &0.319 (0.085)      &0.214 (0.060)     &0.149 (0.040)  \\                              \hline

			Cauchy     &SIR-DEE               &1.584 (0.194)          &1.548 (0.214)         &1.539 (0.205)      \\
			&MDDM                  &1.591 (0.186)          &1.600 (0.202)         &1.615 (0.184)       \\
			&IM$_{\mathrm{PR}}$    &1.590 (0.207)          &1.595 (0.209)         &1.611 (0.195)        \\
			$\gamma=5$ &IM$_{\mathrm{Guass}}$ &1.584 (0.219)          &1.560 (0.215)         &1.601 (0.194)   \\  \vspace{-1.5mm}\\
			&PR-SAVE               &0.512 (0.161)          &0.304 (0.082)         &0.204 (0.051)    \\
			&IV$_{\mathrm{PR}}$    &\textbf{0.427 (0.119)}          &\textbf{0.282 (0.076)}       &\textbf{0.192 (0.048)}  \\
			$\gamma=5$ &IV$_{\mathrm{Guass}}$ &0.557 (0.181)          &0.321 (0.084)         &0.210 (0.056) \\
			&IV$_{\mathrm{Lap}}$   &0.485 (0.154)          &0.285 (0.077)         &0.191 (0.049)  \\
			&IV$_{\mathrm{RQ}}$    &0.494 (0.148)          &0.291 (0.077)         &0.194 (0.051)  \\			
			\botrule
			\end{tabular*}
\end{table*}

\begin{table*}[hptb]    
	\caption{Average runtime (in seconds) of various SDR methods.
	 \label{tab_1-2_runtime}}
		\begin{tabular*}{\textwidth}{@{\extracolsep\fill}cclccc}
            \toprule
			\multicolumn{1}{c}{Settings}&\multicolumn{1}{c}{Errors} &\multicolumn{1}{c}{Methods}
			&\multicolumn{1}{c}{$\begin{array}{c}n=225  \vspace{-0.8ex}\\ (p=10)\end{array} $ }
			&\multicolumn{1}{c}{ $\begin{array}{c}n=400 \vspace{-0.8ex}\\ (p=15)\end{array} $}
			&\multicolumn{1}{c}{$\begin{array}{c}n=625 \vspace{-0.8ex}\\ (p=20)\end{array} $ }\\
			\midrule
\multirow{1}{*}{\makecell[l]{Example \ref{exampl_1}\\  Case (i)}}

            &$N(0, 1)$  &SIR                  &0.0002  &0.0003 &0.0004  \\
			&           &CUME                 &0.0001  &0.0002 &0.0004  \\                           	
            &           &MDDM                 &0.0004  &0.0015 &0.0035  \\[0.15cm]

			&           &IM$_{\mathrm{PR}}$   &0.0345  &0.3155 &0.9562  \\
			&$\gamma=40$&IM$_{\mathrm{Guass}}$&0.0005  &0.0015&0.0037   \\
			&           &IM$_{\mathrm{Lap}}$  &0.0004  &0.0015 &0.0036  \\
			&           &IM$_{\mathrm{RQ}}$   &0.0004  &0.0014 &0.0034  \\  \hline  \\[-0.3cm] 
\multicolumn{1}{c}{ }&\multicolumn{1}{c}{ } &\multicolumn{1}{c}{ }
	&\multicolumn{1}{c}{$\begin{array}{c}n=200  \vspace{-0.8ex}\\ (p=6)\end{array} $ }
&\multicolumn{1}{c}{ $\begin{array}{c}n=400 \vspace{-0.8ex}\\ (p=6)\end{array} $}
&\multicolumn{1}{c}{$\begin{array}{c}n=800 \vspace{-0.8ex}\\ (p=6)\end{array} $ }\\
  \cmidrule{4-6}
\multirow{1}{*}{Example \ref{exampl_2}}

        &   $\mathrm{N}_{4}(\mathbf{0}, \mathbf{\Delta})$	& SIR-DEE             &0.2274 &2.9356 &20.1340        \\
		&	& PR-SIR              &0.0260 &0.0530 &0.0796         \\
        &    & MDDM                &0.0004 &0.0014 &0.0049 \\[0.15cm]
		&	&IM$_{\mathrm{PR}}$   &0.0469 &0.5675 &3.0127  \\
		&	$\gamma=40$&IM$_{\mathrm{Guass}}$&0.0004 &0.0015 &0.0053 \\
		&	&IM$_{\mathrm{Lap}}$  &0.0004 &0.0014 &0.0050  \\
		&	&IM$_{\mathrm{RQ}}$   &0.0003 &0.0013 &0.0048  \\

  \hline
		\end{tabular*}
\end{table*}

We further evaluate the computational efficiency of our proposed methods against existing SDR competitors. To conserve space, we report results only for Example~\ref{exampl_1}, Case (i) (univariate response) and Example~\ref{exampl_2} with $\mathrm{N}_{4}(\mathbf{0}, \mathbf{\Delta})$ errors (multivariate response). The runtime comparisons are summarized in Table~\ref{tab_1-2_runtime}. The results indicate that for univariate responses, our methods 
(excluding IM$_{\mathrm{PR}}$) are somewhat slower than SIR and CUME, yet still remain highly efficient. The computational advantage of our approach becomes more evident in the multivariate response setting, where our kernel-based methods (IM$_{\mathrm{Gauss}}$, IM$_{\mathrm{Lap}}$, IM$_{\mathrm{RQ}}$) significantly outperform SIR-DEE and PR-SIR. The angle-based IM$_{\mathrm{PR}}$ method is comparatively slower, which can be attributed to its $O(n^3)$ complexity, as opposed to the $O(n^2)$ complexity of our other proposals.

\section{Real data analysis}\label{Realdata}

In this section, we illustrate the proposed method through analysis of two datasets:
the Minneapolis school data from \citet{Cook1998regression} and the plasma retinol and
beta-carotene data   from \citet{NIERENBERG1989DETERMINANTS}.

\subsection{Minneapolis   school data}

The Minneapolis school dataset comprises performance data from 63 elementary schools, as studied by \citet{Cook1998regression}. It includes four response variables measuring the percentages of fourth- and sixth-grade students scoring above and below the national average on standardized reading tests, denoted as ``4ABOVE",  ``4BELOW",  ``6ABOVE", and  ``6BELOW".

We aim to model the relationships between these four responses and the following eight predictors:
(1)   percentage of children receiving aid, (2)   percentage of children who do not live with both parents, (3)   percentage of  persons who suffer from poverty,
(4)   percentage of adults in the school area who completed high school, (5)
  percentage of minority students, (6)   percentage of mobility, (7)   percentage of  pupils who attend school regularly, and (8)   pupil-teacher ratio.

Following \citet{YinMoment}, we apply square root transformations to the seven percentage predictors to better meet linearity and homoscedasticity assumptions. The structural dimension is set to
$d=1$. Let $\hat{\mathcal S}_{\bf Y|\bf X}$ denote the estimated subspace using the full dataset. To compare our methods with alternatives, we generate
$B$ bootstrap samples, estimate subspaces $\hat{\mathcal S}^b_{\bf Y|\bf X}$ for each, nd compute the distances  $\mathrm{dist}(\hat{\mathcal S}_{\bf Y|\bf X}, \hat{\mathcal S}^b_{\bf Y|\bf X}),$  for $b = 1,\cdots, B$.

Table \ref{real-data-1-table}  summarizes  the  mean, standard deviation (SD) and MSE of  $\mathrm{dist}(\hat{\mathcal S}_{\bf Y|\bf X}, \hat{\mathcal S}^b_{\bf Y|\bf X}),$
over $B=2000$ simulations (with $\gamma=1000$ for the kernel-based methods).
The results indicate that second-order methods perform   better than first-order methods, with
IV$_{\mathrm{Lap}}$ achieving the best overall performance. Fig. \ref{fig:exampl_4} presents scatterplot matrices of the four responses against the first estimated direction obtained via the
 IM$_{\mathrm{Lap}}$ and  IV$_{\mathrm{Lap}}$ methods, respectively.  The figure reveals a quadratic trend between the responses and the first direction, a finding consistent with the results of \citet{YinMoment}.

\begin{table*}[hptb]
	\caption{Summary    of
		$\mathrm{dist}(\hat{\mathcal S}_{\bf Y|\bf X}, \hat{\mathcal S}^b_{\bf Y|\bf X})$    for the Minneapolis School Data.\label{real-data-1-table} }
		\begin{tabular*}{\textwidth}{@{\extracolsep\fill}lcccclccc} \toprule
			\multicolumn{1}{c}{Methods}
			&\multicolumn{1}{c}{ Mean }
			&\multicolumn{1}{c}{SD}
			&\multicolumn{1}{c}{MSE}  && \multicolumn{1}{c}{Methods}
			&\multicolumn{1}{c}{ Mean }
			&\multicolumn{1}{c}{SD}
			&\multicolumn{1}{c}{MSE}            \\
			\midrule
			PR-SIR               &1.050 &0.313 &1.200  &&  PR-SAVE                  &0.860     &0.380     &0.884          \\
			MDDM                 &0.992 &0.315 &1.083  &&  IV$_{\mathrm{PR}}$       &0.446     &0.374     &0.339                                                        \\
			IM$_{\mathrm{PR}}$   &1.062 &0.297 &1.216  &&  IV$_{\mathrm{Guass}}$    &0.473     &0.404     &0.387            \\
			IM$_{\mathrm{Guass}}$&0.892 &0.306 &0.889  &&  IV$_{\mathrm{Lap}}$      &\textbf{0.434}     &\textbf{0.378}     &\textbf{0.331}       \\
			IM$_{\mathrm{Lap}}$  &0.888 &0.344 &0.907  &&  IV$_{\mathrm{RQ}}$       &0.466     &0.400     &0.377      \\
			IM$_{\mathrm{QR}}$   &0.885 &0.347 &0.904  &&              \\
			\botrule
		\end{tabular*}
\end{table*}

\begin{figure*}
	\centering
	\includegraphics[width=6cm,height=6cm]{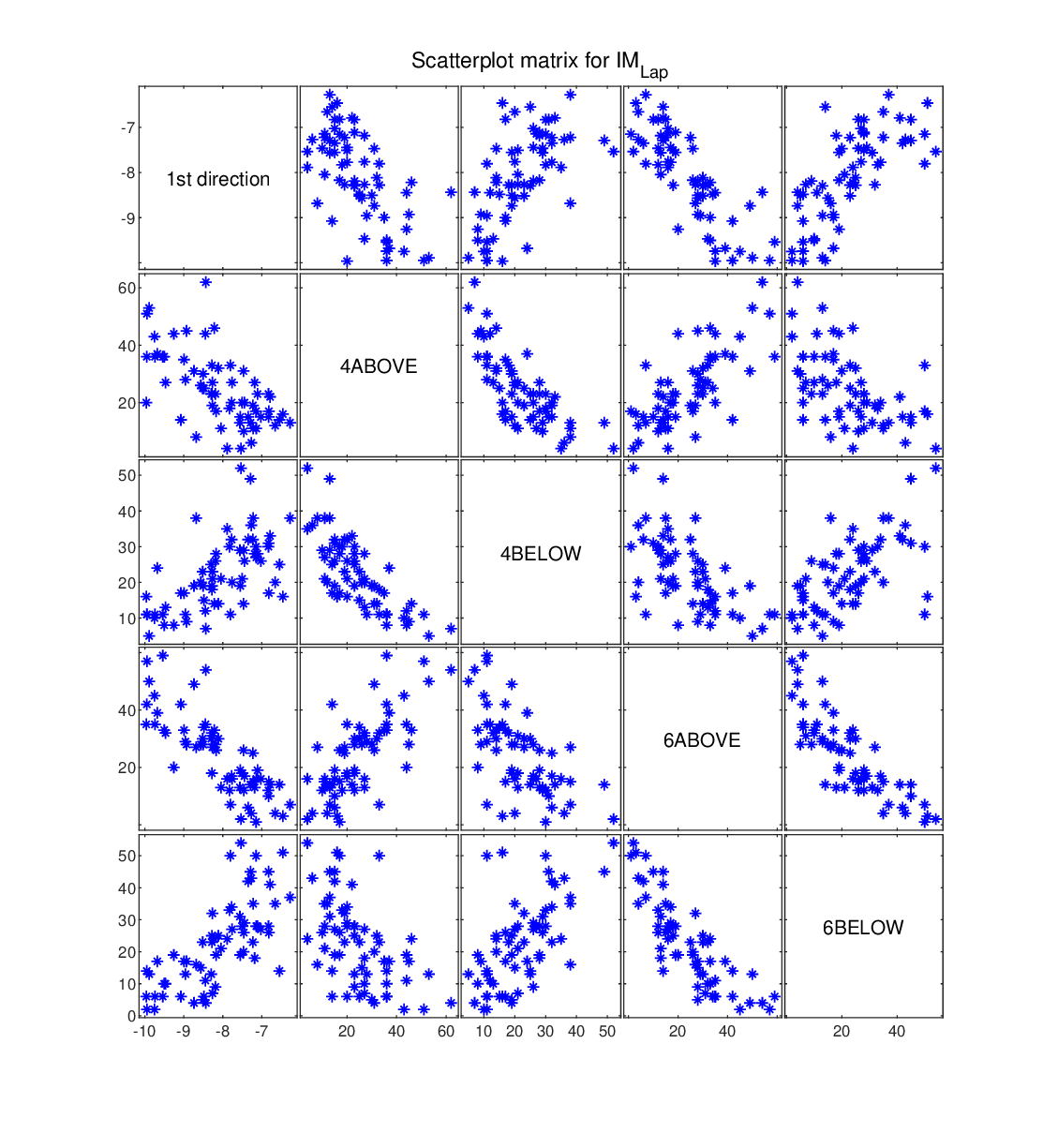}
	\includegraphics[width=6cm,height=6cm]{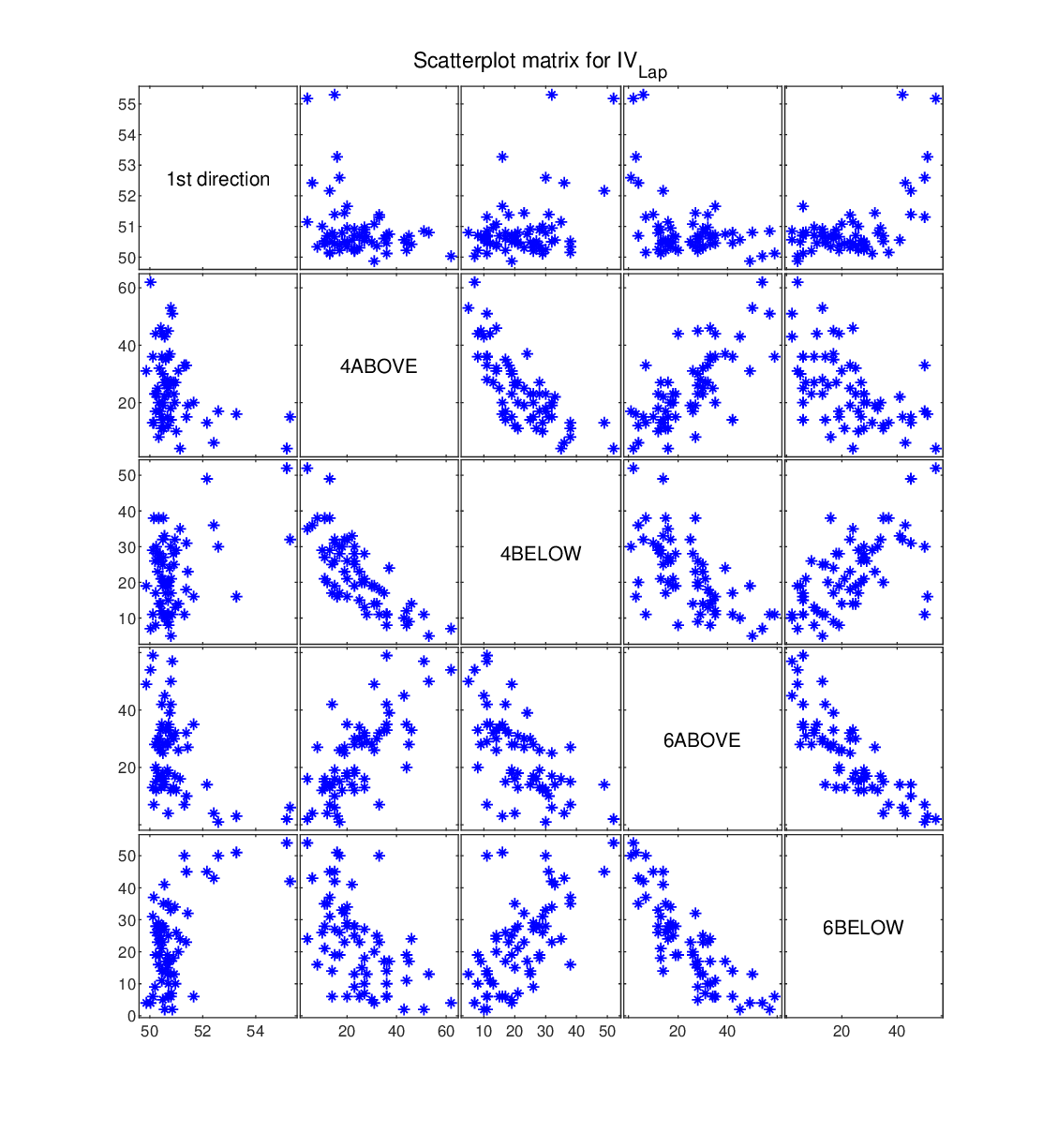}
	\caption{Scatterplot matrices  of the four responses  versus the first estimated direction   obtained from
		the  IM$_{\mathrm{Lap}}$ and  IV$_{\mathrm{Lap}}$ methods
		 for the Minneapolis school data.
		\label{fig:exampl_4}}
\end{figure*}

\subsection{Plasma retinol and beta-carotene data}

In this section, we apply the proposed methods to analyze the plasma beta-carotene level dataset from \citet{NIERENBERG1989DETERMINANTS}. The study aims to investigate the associations of plasma retinol and beta-carotene concentrations with 12 personal characteristics and dietary factors among 315 patient.

\citet{Zhu2010ON} introduced a SDR method, referred to as the UIF method, to identify a low-dimensional projection of the dataset. The two response variables are plasma beta-carotene and plasma retinol. Following \citet{Zhu2010ON}, we use nine quantitative predictors:   age,  quetelet index (BMI), number of calories, grams of fat, grams of fiber, number of alcoholic drinks, cholesterol,   dietary beta-carotene  and  dietary retinol. A full description of the data is available in the StatLib database at \url{http://lib.stat.cmu.edu/datasets/Plasma_Retinol}.

As in \citet{Zhu2010ON},  the structural dimension is estimated as   $d=2$. Using the bootstrap approach described in the previous section, we obtain the full-data subspace estimate
   $\hat{\mathcal S}_{\bf Y|\bf X}$  and its bootstrap counterparts $\hat{\mathcal S}^b_{\bf Y|\bf X}$. To compare our methods with UIF, we adopt the performance measure from \citet{Zhu2010ON}: $1- r^b (d),$  where $r^b (d)$ denotes the trace correlation coefficient
    between $\hat{\mathcal S}_{\bf Y|\bf X}$ and $\hat{\mathcal S}^b_{\bf Y|\bf X}$.

Table \ref{real-data-2-table} summarizes   the results of   $1- r^b (d)$ over $B=2000$ simulations ($\gamma=1000$ for the kernel-based methods). The results indicate   that
IV$_{\mathrm{Guass}}$, IV$_{\mathrm{Lap}}$ and IV$_{\mathrm{RQ}}$ outperform the other methods. For comparison,   the median of $1- r^b (d)$ reported for the UIF method in Table 5.5 of \citet{Zhu2010ON} is 0.0669. In this dataset,     IV$_{\mathrm{Guass}}$,
IV$_{\mathrm{Lap}}$ and IV$_{\mathrm{RQ}}$  all achieve better performance than UIF.

 \begin{table*}[hptb]
	\caption{Summary  of $1- r^b (d)$  for the plasma   retinol and beta-carotene   data.\label{real-data-2-table} }
		\begin{tabular*}{\textwidth}{@{\extracolsep\fill}lccclcc}
			\toprule
			\multicolumn{1}{c}{Methods}
			&\multicolumn{1}{c}{Median  }
			&\multicolumn{1}{c}{SD}
			&& \multicolumn{1}{c}{Methods}
			&\multicolumn{1}{c}{Median }
			&\multicolumn{1}{c}{SD}
			\\
			\midrule
			PR-SIR                &0.125      &0.112      &&PR-SAVE                  &0.077     &0.159                   \\
			MDDM                  &0.100      &0.094      &&IV$_{\mathrm{PR}}$       &0.379     &0.133                                                                                                 \\
			IM$_{\mathrm{PR}}$    &0.097      &0.093      &&IV$_{\mathrm{Guass}}$    &0.058     &0.156                         \\
			IM$_{\mathrm{Guass}}$ &0.308      &0.167      &&IV$_{\mathrm{Lap}}$      &\textbf{0.054}     &\textbf{0.146}  \\
			IM$_{\mathrm{Lap}}$   &0.112      &0.099      &&IV$_{\mathrm{RQ}}$       &0.055     &0.152                            \\
			IM$_{\mathrm{QR}}$    &0.302      &0.155      &&       \\
			\botrule
		\end{tabular*}
\end{table*}

\section{Discussion}\label{Conclusion}

This paper   presents  a unified  framework  of  sufficient dimension reduction by  measuring independence through  inverse conditional moments between random vectors.  To quantify such independence, we employ both projection and kernel techniques, leading to two methodological variants: projection-based and kernel-based approaches. This yields four distinct matrix formulations: the projection-based ICMI and ICVI matrices, denoted as $\mathrm{ICMI.M}_\mathrm{PR}(\mathbf{X}|\mathbf{Y})$ and $\mathrm{ICVI.M}_\mathrm{PR}(\mathbf{X}|\mathbf{Y})$, and their kernel-based counterparts, $\mathrm{ICMI.M}_\mathrm{K}(\mathbf{X}|\mathbf{Y})$ and $\mathrm{ICVI.M}_\mathrm{K}(\mathbf{X}|\mathbf{Y})$.
Under standard linearity and constant variance conditions, we show that all four matrices can consistently recover the central subspace. We further  derive their empirical estimators and investigate their theoretical properties.

The proposed methods can be categorized into two classes according to \citet{sejdinovic2013}. Specifically, $\mathrm{ICMI.M}_\mathrm{PR}(\mathbf{X}|\mathbf{Y})$, and $\mathrm{ICVI.M}_\mathrm{PR}(\mathbf{X}|\mathbf{Y})$  fall into the class of distance-based measures, while $\mathrm{ICMI.M}_\mathrm{K}(\mathbf{X}|\mathbf{Y})$ and $\mathrm{ICVI.M}_\mathrm{K}(\mathbf{X}|\mathbf{Y})$ belong to the kernel-based family. The theoretical connection between these two classes can be systematically explored following \citet{sejdinovic2013}. Within this framework, the recent work  proposed by \citet{ying2022frechet} represents an important generalization of distance-based approaches, extending SDR to non-Euclidean data through general distance functions.

This work focuses   on settings without sparsity assumptions. However, the proposed   methods can be extended to ultrahigh-dimensional regression problems where $\mathbf{X}$ is sparse and $p \gg n$. For related methodologies in such contexts, we refer readers to \citet{Fang2018Han} and \citet{wang2019yu}.  It is also important to note that the robustness  established here is   limited to the response variable. Extending the proposed procedures to handle heavy-tailed predictors remains an important direction for our future research.

\section*{Acknowledgement}
We sincerely thank the editor, the associate editor and  the reviewers for their insightful suggestions which have improved the manuscript significantly.

\backmatter

\section*{Declarations}
Jicai Liu's research was  supported by the Humanities and Social Sciences Youth Foundation of Ministry of Education of China (23YJC910003). Jinhong Li's work was supported by the Summit Advancement Disciplines Program of Zhejiang Province (Zhejiang Gongshang University - Statistics) and the Collaborative Innovation Center of Statistical Data Engineering Technology \& Application.

No conflict of interest exits in the submission of this manuscript, and manuscript is approved by all authors for publication. I would like to declare on behalf of my co-authors that the work described was original research that has not been published previously, and not under consideration for publication
elsewhere, in whole or in part. All the authors listed have approved the manuscript that is enclosed.

\bigskip

\begin{appendices}

\section*{Appendix: Proofs}
\setcounter{equation}{0}
\setcounter{lemma}{0}
\renewcommand{\theequation}{A.\arabic{equation}}
\renewcommand{\thelemma}{A.\arabic{lemma}}

\begin{proof}[\textbf{Proof of   \eqref{equ-kernel}}]

We first prove the ``$\Rightarrow$" part. If $E\{\mathbf{X}|\mathbf{Y}\}=E \{\mathbf{X}\} $, we have   that
\begin{eqnarray*}
	   {E}\{{X}_{j}| \phi(\mathbf{Y})\}={E} \{{E}\{{X}_{j}|\mathbf{Y}\}| \phi(\mathbf{Y}) \}= {E}\{{X}_{j} \},
	\end{eqnarray*}
 for  $j=1,\cdots,p.$

 We now prove   the ``$\Leftarrow$" part.  Note that
	$$  {\boldsymbol{\Lambda}}_{\mathrm{K}, jj}
\!	=\! \langle{E}\{( {X}_{1 j}-E\{ {X}_{j}\})\phi(\mathbf{Y}_1)\},{E}\{ ({X}_{2 j}\!-\! E\{ {X}_j\})\phi(\mathbf{Y}_2)\}\rangle_\mathcal{H}.$$
Thus, we have
	\begin{eqnarray}\label{equ-2}
		{E}\{{X}_{j}| \phi(\mathbf{Y})\}= {E}\{{X}_{j} \}&\Longrightarrow&
		{\boldsymbol{\Lambda}}_{\mathrm{K}, jj} =0.
	\end{eqnarray}
By Theorem \ref{Bochner-induced},  we obtain
	\begin{eqnarray*}
		 { \boldsymbol{ \Lambda}}_{\mathrm{K}, jj}=0
 &\Longleftrightarrow& E\{{X}_{j}|\mathbf{Y}\}=E \{{X}_{j}\}.
	\end{eqnarray*}
This, together with \eqref{equ-2}, completes the proof of    the ``$\Leftarrow$" part.	
\end{proof}

\begin{proof}[\textbf{Proof of Proposition  \ref{PR_CUME}}]
	For any nonzero scalars  ${v}_1,{v}_2\in  \mathbb{R}$, we have
	\begin{eqnarray*}
		\mathrm{ang}( {v}_1, {v}_2)&=&\arccos \Big\{ \frac{ {v}_1   {v}_2}{\| {v}_1\|_2\|  {v}_2\|_2}\Big\}\\
		&=& \pi[ I( {v}_1<0)I( {v}_2>0) +I( {v}_1>0)I( {v}_2<0)].
	\end{eqnarray*}
	Suppose that ${Y}$ is continuous and univariate. Then,  we obtain
	\begin{eqnarray*}
		&&\mathrm{ang}( {Y}_1- {Y}_3,  {Y}_2- {Y}_3)\\&=& \pi [ I( {Y}_1< {Y}_3) I( {Y}_2> {Y}_3)+I( {Y}_1> {Y}_3) I( {Y}_2< {Y}_3)],
	\end{eqnarray*}
since  $\operatorname{pr}({Y}_i= {Y}_j)=0,$ $i\neq j$.  It follows that
	\begin{eqnarray*}
		&&\mathrm{ICMI.M}_\mathrm{PR}(\textbf{X}|{Y})\\
		&=&\!\!-\mathrm{E}\{ ( \mathbf{X}_1-E\{ \mathbf{X}_1\}) (\mathbf{X}_2-E\{ \mathbf{X}_2\})^T\\ &&\times\mathrm{ang}( {Y}_1-  {Y}_3, {Y}_2-  {Y}_3)\}\\
		\!\!&=&\!\!-\pi\mathrm{E}\{ ( \mathbf{X}_1-E\{ \mathbf{X}_1\}) (\mathbf{X}_2-E\{ \mathbf{X}_2\})^T\\
		&&\times[I( {Y}_1< {Y}_3)-I( {Y}_1< {Y}_3) I( {Y}_2\leqslant{Y}_3)] \}\\
		&&-\pi\mathrm{E}\{ ( \mathbf{X}_1-E\{ \mathbf{X}_1\}) (\mathbf{X}_2-E\{ \mathbf{X}_2\})^T\\
		&&\times
		[I( {Y}_2< {Y}_3)-I( {Y}_1\leqslant {Y}_3) I( {Y}_2<{Y}_3)] \}\\
		\!\!&=&\!\!2\pi\mathrm{E}\{ ( \mathbf{X}_1-E\{ \mathbf{X}_1\}) (\mathbf{X}_2-E\{ \mathbf{X}_2\})^T\\
		&&\times
		I( {Y}_1< {Y}_3) I( {Y}_2\leqslant{Y}_3)] \}\\
&=& \!\! 2\pi{ \boldsymbol{ \Lambda}}_\mathrm{CUME},
	\end{eqnarray*}
	provided  $\varpi(\cdot)=1$.
\end{proof}

\begin{proof}[\textbf{Proof of Theorem  \ref{them-projection}}]
	By  \eqref{intr0} and the  tower property  of conditional expectation, we have
	\begin{eqnarray*}
		\mathrm{E} \{\mathbf{X}_i-E\{ \mathbf{X}_i\}|\mathbf{Y}_i \}
		&=&\mathrm{E} \{\mathrm{E}[\mathbf{X}_i-E\{ \mathbf{X}_i\}|\mathbf{Y}_i, \mathbf{\Gamma}^T\mathbf{X}_i ] |\mathbf{Y}_i \}\\
		&=&\mathrm{E} \{\mathrm{E}[\mathbf{X}_i-E\{ \mathbf{X}_i\}|  \mathbf{\Gamma}^T\mathbf{X}_i] |\mathbf{Y}_i \}.
	\end{eqnarray*}
	This, together with  the linearity condition that  $E\{\mathbf{X} |\mathbf{\Gamma}^T \mathbf{X}\}=\mathbf{P}_{\mathbf{\Gamma}}^T(\mathbf{\Sigma}) \mathbf{X}$,  yields that
	\begin{eqnarray*}
		&&{ \boldsymbol{ \Lambda}}_\mathrm{PR}\\&=&\mathrm{ICMI.M}_\mathrm{PR}(\textbf{X}|\mathbf{Y})\nonumber\\
		& =& -\mathrm{E} \{ \mathrm{E}\{\mathbf{X}_1-E\{ \mathbf{X}_1\}|\mathbf{Y}_1\} \mathrm{E}\{\mathbf{X}_2-E\{ \mathbf{X}_2\}|\mathbf{Y}_2\}^T\nonumber\\
		&&\times \mathrm{ang}(\mathbf{Y}_1- \mathbf{Y}_3,\mathbf{Y}_2- \mathbf{Y}_3) \} \nonumber\\
		& =&-\mathrm{E} \{\mathbf{P}_{\Gamma}^T(\mathbf{\Sigma}) \mathrm{E}\{\mathbf{X}_1-E\{ \mathbf{X}_1\}|\mathbf{Y}_1\} \mathrm{E}\{\mathbf{X}_2-E\{ \mathbf{X}_2\}|\mathbf{Y}_2\}^T\\
		&&\times\mathbf{P}_{\Gamma}(\mathbf{\Sigma}) \mathrm{ang}(\mathbf{Y}_1- \mathbf{Y}_3,\mathbf{Y}_2- \mathbf{Y}_3) \} \nonumber\\
		& =&\mathbf{P}_{\Gamma}^T(\mathbf{\Sigma}){ \boldsymbol{ \Lambda}}_\mathrm{PR} \mathbf{P}_{\Gamma}(\mathbf{\Sigma}). \nonumber
	\end{eqnarray*}
	Thus,  $\mathcal{S}({ \boldsymbol{ \Lambda}}_\mathrm{PR}) \subseteq \mathbf{\Sigma} \mathcal{S}_{Y | \mathrm{X}}$.
	\end{proof}

\begin{proof}[\textbf{Proof of Theorem  \ref{them-kenerl-consistency}}]
	
We prove the theorem using an $\varepsilon$-net argument. Let $\mathbb{S}^{p-1}$ be the unit sphere in $\mathbb{R}^p$ equipped with Euclidean distance, and let $\mathcal{N}_\varepsilon$ be an $\varepsilon$-net of $\mathbb{S}^{p-1}$. By Corollary~4.2.13 of \citet{vershynin2018high}, the cardinality of $\mathcal{N}_\varepsilon$ satisfies
\begin{equation}\label{net-bound}
|\mathcal{N}_{\varepsilon}| \leq \left(1 + \frac{2}{\varepsilon}\right)^p.
\end{equation}
For any symmetric matrix $\mathbf{A} \in \mathbb{R}^{p \times p}$, Lemma~4.4.1 of \citet{vershynin2018high} gives the following bound on the operator norm
\begin{equation}\label{operator-bound}
		\|\mathbf{A}\|= \sup _{\boldsymbol{v} \in \mathbb{S}^{p-1}} |\boldsymbol{v}^{T} \mathbf{A} \boldsymbol{v} | \leq \frac{1}{1-2 \epsilon} \sup _{\boldsymbol{v} \in \mathcal{N}_{\epsilon}} |\boldsymbol{v}^{T} \mathbf{A} \boldsymbol{v} |.
	\end{equation}

Taking $\varepsilon = 1/4$, we obtain from \eqref{net-bound} and \eqref{operator-bound}
 \begin{equation*}
		|\mathcal{N}_{1 / 4}|< 9^{p}   ~~ \mathrm{and}  ~ ~\|\mathbf{A}\| \leq 2 \sup _{\boldsymbol{v} \in \mathcal{N}_{1 / 4}} |\boldsymbol{v}^{T} \mathbf{A} \boldsymbol{v} |.
	\end{equation*}
Then, for any   $t> 0$, we have  	
	\begin{eqnarray*}\label{approach-net}
		&&\operatorname{pr}\left(\|\widehat{ \boldsymbol{ \Lambda}}_\mathrm{PR} -\boldsymbol{ \Lambda}_\mathrm{PR}\|  \geq t  \right)\nonumber\\
		  & \leq&  \operatorname{pr}\left(\sup _{\boldsymbol{v} \in \mathcal{N}_{1 / 4}}|\boldsymbol{v}^{\mathrm{T}}\{\widehat{ \boldsymbol{ \Lambda}}_\mathrm{PR} -\boldsymbol{ \Lambda}_\mathrm{PR}\} \boldsymbol{v}| \geq t/2\right)\nonumber\\
		& \leq&   \sum_{\boldsymbol{v} \in \mathcal{N}_{1 / 4}} \operatorname{pr}\left( |\boldsymbol{v}^{\mathrm{T}}\{\widehat{ \boldsymbol{ \Lambda}}_\mathrm{PR} -\boldsymbol{ \Lambda}_\mathrm{PR}\} \boldsymbol{v}| \geq t/2\right).
	\end{eqnarray*}
By Lemma  \ref{lemma1-thprojection},  we have
	\begin{eqnarray}\label{unbound-theoem-proof1}
		&&\operatorname{pr}\left(\|\widehat{ \boldsymbol{ \Lambda}}_\mathrm{PR} -\boldsymbol{ \Lambda}_\mathrm{PR}\|  \geq t  \right)\nonumber\\
		& \leq&  9^{p} \sup _{\boldsymbol{v} \in\mathcal{S}^{p-1}  }  \operatorname{pr}\left( |\boldsymbol{v}^{\mathrm{T}}\{\widehat{ \boldsymbol{ \Lambda}}_\mathrm{PR} -\boldsymbol{ \Lambda}_\mathrm{PR}\} \boldsymbol{v}| \geq t/2 \right)\nonumber\\
		& \leq&   9^{p} c \exp \left[-C n\min \left(\frac{t^{2}}{ \|\mathbf{\Sigma}\|^{2}}, \frac{ t}{\|\mathbf{\Sigma}\|}\right)\right]\nonumber\\
		& \leq&   c \exp \left[-C n  \frac{t^{2}}{ \|\mathbf{\Sigma}\|^{2} + t\|\mathbf{\Sigma}\|}+p  \log 9 \right].
	\end{eqnarray}
where the last inequality follows from $\min(a,b) \geq ab/(a+b)$.

	For any $0<\alpha<1,$ let
	\begin{eqnarray}\label{exp-bound-inequality}
\exp \left[-C n  \frac{t^{2}}{ \|\mathbf{\Sigma}\|^{2} + t\|\mathbf{\Sigma}\|}+p  \log 9 \right]\leq\alpha.
\end{eqnarray}
Taking logarithms and rearranging, we obtain
	$$  \frac{p-\log \alpha}{n}  \leq\frac{Ct^{2}}{ \|\mathbf{\Sigma}\|^{2} + t\|\mathbf{\Sigma}\|} ~~\Big(\leq \frac{Ct^{2}}{ \|\mathbf{\Sigma}\|^{2}}\Big).$$
Thus, \eqref{exp-bound-inequality} holds for any  $t$, satisfying
	$$t  \geq C^{-1}  \|\mathbf{\Sigma}\| \sqrt{ \frac{p+\log(1/\alpha)}{n}}.$$

Taking $t= C^{-1} \|\mathbf{\Sigma}\| \sqrt{ \frac{p+\log(1/\alpha)}{n}}$ in \eqref{unbound-theoem-proof1}, we obtain
 	\begin{eqnarray*}
		\operatorname{pr}\left(\|\widehat{ \boldsymbol{ \Lambda}}_\mathrm{PR} -\boldsymbol{ \Lambda}_\mathrm{PR}\|  \geq C^{-1} \|\mathbf{\Sigma}\| \sqrt{ \frac{p+\log(1/\alpha)}{n}} \right)
		\leq    \alpha.
	\end{eqnarray*}
This complete the proof.
\end{proof}

\begin{lemma}\label{lemma1-thprojection} Assume that  Condition 1 holds.
	For any   $\boldsymbol{v} \in\mathbb{S}^{p-1},$ we have
	\begin{eqnarray*}
		&&\operatorname{pr}\left( |\boldsymbol{v}^{\mathrm{T}}\{\widehat{ \boldsymbol{ \Lambda}}_\mathrm{PR} -\boldsymbol{ \Lambda}_\mathrm{PR}\} \boldsymbol{v}| \geq t \right)\\
		 &&\leq  c \exp \left[-C n\min \left(\frac{t^{2}}{ \|\mathbf{\Sigma}\|^{2}}, \frac{ t}{\|\mathbf{\Sigma}\|}\right)\right], ~~ t>0\nonumber,
	\end{eqnarray*}
	where  $c$ and $C$ are two positive constants.
\end{lemma}

\begin{proof}[\textbf{Proof of Lemma  \ref{lemma1-thprojection}}]
	  By direct calculation, we obtain  the following decomposition
	\begin{eqnarray*}
		\boldsymbol{v}^{\mathrm{T}} \boldsymbol{\widehat{ \Lambda}}_\mathrm{PR}  \boldsymbol{v}
		 &=&  -\frac{1}{n^3}  \sum_{i,r,l=1}^n\!\!
		\boldsymbol{v}^{\mathrm{T}} [(\mathbf{X}_i\!-\!E\{\mathbf{X}\})\!-\!(\overline{\mathbf{X}}_{n}\!-\!E\{\mathbf{X}\} ) ]\\&&\times [(\mathbf{X}_l\!-\!E\{\mathbf{X}\})\!-\!(\overline{\mathbf{X}}_{n}\!-\!E\{\mathbf{X}\} ) ]^T\boldsymbol{v}\\&&\times\mathrm{ang}(\mathbf{Y}_i -\mathbf{Y}_r, \mathbf{Y}_l-\mathbf{Y}_r) \\
		&=&-S_{1n} -S_{2n}  + 2S_{3n},
	\end{eqnarray*}
where the three terms are defined as
	\begin{eqnarray*}
		 S_{1n}&=&\frac{1}{n^3} \sum_{i,r,l=1}^n \boldsymbol{v}^{\mathrm{T}}[\mathbf{X}_i-E\{\mathbf{X}\}] [\mathbf{X}_l-E\{\mathbf{X}\}]^T\boldsymbol{v}\\
		&&\times\mathrm{ang}(\mathbf{Y}_i  -\mathbf{Y}_r, \mathbf{Y}_l-\mathbf{Y}_r), \nonumber\\
		S_{2n}&=& \frac{1}{n^3}\sum_{i,r,l=1}^n \mathrm{ang}(\mathbf{Y}_i  -\mathbf{Y}_r, \mathbf{Y}_l-\mathbf{Y}_r)\\
		&&\times\frac{1}{n^2}\sum_{j, m=1}^n
		\boldsymbol{v}^{\mathrm{T}}[\mathbf{X}_j-E\{\mathbf{X}\}]  [\mathbf{X}_m-E\{\mathbf{X}\}]^T\boldsymbol{v}  ,   \nonumber\\
		S_{3n}&=&\frac{1}{n^3} \sum_{i, r,l=1}^n  \mathrm{ang}(\mathbf{Y}_i  -\mathbf{Y}_r, \mathbf{Y}_l-\mathbf{Y}_r)\\
		&&\times\{\boldsymbol{v}^{\mathrm{T}}[\mathbf{X}_l-E\{\mathbf{X}\}]\}\frac{1}{n}
		\sum_{j =1}^n  \boldsymbol{v}^{\mathrm{T}}[\mathbf{X}_j-E\{\mathbf{X}\}].
	\end{eqnarray*}
	Therefore,   we have
	\begin{eqnarray}\label{widehat.omega-omega}
		&&\operatorname{pr}\left( |\boldsymbol{v}^{\mathrm{T}}\{\widehat{ \boldsymbol{ \Lambda}}_\mathrm{PR} -\boldsymbol{ \Lambda}_\mathrm{PR}\} \boldsymbol{v}| \geq t\right)\nonumber\\
		&\leq&  \operatorname{pr} \left( |S_{1n} \! + \boldsymbol{v}^{\mathrm{T}} \boldsymbol{ \Lambda}_\mathrm{PR} \boldsymbol{v})|\geq    {t}/{3}\right)\!+ \!\operatorname{pr} \left( |S_{2n} |\geq  {t}/{3}\right) \nonumber\\
		&& + \operatorname{pr} \left( |S_{3n}|\geq   {t}/{6}\right).
	\end{eqnarray}
We proceed to bound each term by U-statistics theory. For notational convenience, we assume throughout that the predictors have been centered, i.e., $\mathbf{X}_i$ is replaced by $\mathbf{X}_i - E{\mathbf{X}}$, while maintaining the original notation.

	\emph{ \textbf{Step 1: The tail bound on  $S_{1n} \! + \boldsymbol{v}^{\mathrm{T}} \boldsymbol{ \Lambda}_\mathrm{PR} \boldsymbol{v} $.}}
We begin by expressing $S_{1n}$ as
	\begin{eqnarray*}
		S_{1n} \!\! &= &\!\! \frac{1}{n^3} \sum_{i,r,l=1}^n\{\boldsymbol{v}^T\mathbf{X}_i\} \{ \mathbf{X}_l^T\boldsymbol{v}\} \mathrm{ang}(\mathbf{Y}_i \!\!  -\!\!\mathbf{Y}_r, \mathbf{Y}_l-\mathbf{Y}_r) \nonumber\\
		\!\!&= &\!\!\frac{1}{n^3} \Big\{ \!\!\! \sum_{\substack{i\neq r,i\neq l,\\ r\neq l}}\!\!\!\!\{\boldsymbol{v}^T\mathbf{X}_i\} \{ \mathbf{X}_l^T\boldsymbol{v}\} \mathrm{ang}(\mathbf{Y}_i  -\mathbf{Y}_r, \mathbf{Y}_l \!\!-\!\!\mathbf{Y}_r)\\
		&& +\!\!\!
		\sum_{\substack{i\neq r,i\neq l,\\ r= l}}\!\!\!\! \{\boldsymbol{v}^T\mathbf{X}_i\} \{ \mathbf{X}_l^T\boldsymbol{v}\} \mathrm{ang}(\mathbf{Y}_i  -\mathbf{Y}_r, \mathbf{Y}_l-\mathbf{Y}_r)\nonumber\\
		&&+ \!\!\! \sum_{\substack{i\neq r,i= l,\\ r\neq l}} \!\! \!\!\{\boldsymbol{v}^T\mathbf{X}_i\} \{ \mathbf{X}_l^T\boldsymbol{v}\} \mathrm{ang}(\mathbf{Y}_i  -\mathbf{Y}_r, \mathbf{Y}_l-\mathbf{Y}_r)\\
		&&+\!\!\! \sum_{\substack{i = r,i\neq l,\\ r\neq l}} \!\!\!\!  \{\boldsymbol{v}^T\mathbf{X}_i\} \{ \mathbf{X}_l^T\boldsymbol{v}\} \mathrm{ang}(\mathbf{Y}_i \!\!  -\!\! \mathbf{Y}_r, \mathbf{Y}_l-\mathbf{Y}_r)\Big\}\nonumber\\
		\!\! &= &\!\! \frac{1}{n^3}  \!\!\!\!\sum_{\substack{i\neq r,i\neq l,\\ r\neq l}} \!\!\!\!\{\boldsymbol{v}^T\mathbf{X}_i\} \{ \mathbf{X}_l^T\boldsymbol{v}\} \mathrm{ang}(\mathbf{Y}_i  -\mathbf{Y}_r, \mathbf{Y}_l-\mathbf{Y}_r)\\
		&&+\frac{1}{n^3} \sum_{i\neq r,i= l} \!\!\!\!\{\boldsymbol{v}^T\mathbf{X}_i\} \{ \mathbf{X}_l^T\boldsymbol{v}\} \mathrm{ang}(\mathbf{Y}_i \!\!  -\!\!\mathbf{Y}_r, \mathbf{Y}_l \!\!-\!\!\mathbf{Y}_r)
		\nonumber\\
		&=& \frac{(n-1)(n-2)} {n^2}{T}_{1 n},\nonumber
	\end{eqnarray*}
	where    ${T}_{1 n}$  is  a  U-statistic, given by
	\begin{eqnarray*}
		{T}_{1 n}\!\!\!  &=&  \!\!\begin{pmatrix} n \\  3 \end{pmatrix}^{\!\!\!\!-1}  \!\!\!\! \sum_{i<r<l}\!\! \frac{ 1}{3} \Big\{ \{\boldsymbol{v}^T\mathbf{X}_i\} \{ \mathbf{X}_l^T\boldsymbol{v}\} \mathrm{ang}(\mathbf{Y}_i \!\!  -\!\! \mathbf{Y}_r, \mathbf{Y}_l \!\!-\!\! \mathbf{Y}_r)\\
		&&+ \{\boldsymbol{v}^T\mathbf{X}_r\} \{ \mathbf{X}_l^T\boldsymbol{v}\} \mathrm{ang}(\mathbf{Y}_r  -\mathbf{Y}_i, \mathbf{Y}_l-\mathbf{Y}_i)\\
		&& + \{\boldsymbol{v}^T\mathbf{X}_i\} \{ \mathbf{X}_r^T\boldsymbol{v}\} \mathrm{ang}(\mathbf{Y}_i  -\mathbf{Y}_l, \mathbf{Y}_r-\mathbf{Y}_l)\Big \}\\
		&=& \begin{pmatrix} n \\  3 \end{pmatrix}^{-1} \sum_{i<r<l} k (\mathbf{X}_{i},\mathbf{Y}_{i}; \mathbf{X}_{r}, \mathbf{Y}_{r};\mathbf{X}_{l}, \mathbf{Y}_{l} ),
	\end{eqnarray*}
	%and
	%\begin{eqnarray*}
	%{T}_{2 n}%&=&\begin{pmatrix} n \\  2 \end{pmatrix}^{-1} \sum_{i\neq r} \mathrm{ang}( {X}_{i,k} - {X}_{r,k},  {X}_{i,k}- {X}_{r,k})Y_i^2 \\
	%&=&\begin{pmatrix} n \\  2 \end{pmatrix}^{-1}\sum_{i<r}\frac{1}{2}  \mathrm{ang}(\mathbf{Y}_i  -\mathbf{Y}_r, \mathbf{Y}_i-\mathbf{Y}_r)\Big\{\{\mathbf{u}_{1}^T\mathbf{X}_i\} \{ \mathbf{X}_i^T\mathbf{u}_{2}\}+\{\mathbf{u}_{1}^T\mathbf{X}_r\} \{ \mathbf{X}_r^T\mathbf{u}_{2} \} \Big\}\\
	%   &=& 0,
	%  \end{eqnarray*}
and   $k (\mathbf{X}_{i},\mathbf{Y}_{i}; \mathbf{X}_{r}, \mathbf{Y}_{r};\mathbf{X}_{l}, \mathbf{Y}_{l} )$ is defined as
\begin{eqnarray*}
	&&k (\mathbf{X}_{i},\mathbf{Y}_{i}; \mathbf{X}_{r}, \mathbf{Y}_{r};\mathbf{X}_{l}, \mathbf{Y}_{l} )\\
	&=&    \frac{ 1}{3} \Big\{ \{\boldsymbol{v}^T\mathbf{X}_i\} \{ \mathbf{X}_l^T\boldsymbol{v}\} \mathrm{ang}(\mathbf{Y}_i  -\mathbf{Y}_r, \mathbf{Y}_l-\mathbf{Y}_r)\\
	&&+ \{\boldsymbol{v}^T\mathbf{X}_r\} \{ \mathbf{X}_l^T\boldsymbol{v}\} \mathrm{ang}(\mathbf{Y}_r  -\mathbf{Y}_i, \mathbf{Y}_l-\mathbf{Y}_i)\\
	&& + \{\boldsymbol{v}^T\mathbf{X}_i\} \{ \mathbf{X}_r^T\boldsymbol{v}\} \mathrm{ang}(\mathbf{Y}_i  -\mathbf{Y}_l, \mathbf{Y}_r-\mathbf{Y}_l)\Big \}.
\end{eqnarray*}

Let $S_{1}=-\boldsymbol{v}^{\mathrm{T}} \boldsymbol{ \Lambda}_\mathrm{PR} \boldsymbol{v}.$
Note that
\begin{eqnarray}\label{s1nk1}
	&&\operatorname{pr}\left( |S_{1 n} + \boldsymbol{v}^{\mathrm{T}} \boldsymbol{ \Lambda}_\mathrm{PR} \boldsymbol{v} | \geq  t \right)  \nonumber\\
	&=&\!\!\operatorname{pr}\left(  |\frac{(n\!-\!1)(n\!-\!2)}{n^2} [{T}_{1 n}-S_{1}]\!-\!\frac{ 3n\!-\!2}{n^2} S_{1} | \geq  t  \right) \nonumber\\
	& \leq & \operatorname{pr}\left(  | {T}_{1 n}-S_{1} | \geq  \frac{t}{2} \right)+\operatorname{pr}\left(  |\frac{ 3n-2}{n^2} S_{1}  | \geq \frac{t}{2} \right).\nonumber\\
\end{eqnarray}
Since  $E\{ {T}_{1 n}\}=S_{1}$ and    $0\leq \mathrm{ang}(\mathbf{Y}_1  -\mathbf{Y}_3, \mathbf{Y}_4-\mathbf{Y}_3)\leq \pi,$ we have
\begin{eqnarray}\label{S1boundness}
	|E\{ {T}_{1 n}\}|&\leq&   E|\{\boldsymbol{v}^T\mathbf{X}_1\} \{ \mathbf{X}_4^T\boldsymbol{v}\} \mathrm{ang}(\mathbf{Y}_1  -\mathbf{Y}_3, \mathbf{Y}_4-\mathbf{Y}_3) | \nonumber\\ &\leq& \pi E\{|\boldsymbol{v}^T\mathbf{X} |^2\}.
\end{eqnarray}
Under Condition 1, we have
\begin{eqnarray*}
	\operatorname{pr}\left(|\boldsymbol{v}^T\mathbf{X}| \geq t \right) \leq 2 \exp [-C t^{2} /\|\boldsymbol{v}^T\mathbf{X} \|_{\psi_{2}}^{2}]\\ \leq  2 \exp [-C t^{2} /\|\mathbf{X}\|_{\psi_{2}}^{2}],
\end{eqnarray*}
where    $\|\mathbf{X}\|_{\psi_{2}} =\sup _{\boldsymbol{v} \in \mathbb{S}^{p-1}}\|\boldsymbol{v}^T\mathbf{X}\|_{\psi_{2}}$.   This implies
\begin{eqnarray*} E\{|\boldsymbol{v}^T\mathbf{X}|^2\}&=&\int_{0}^{\infty}\operatorname{pr}\left(|\boldsymbol{v}^T\mathbf{X}| \geq \sqrt{u}\right) d u\nonumber\\
	&\leq&  \int_{0}^{\infty}2 \exp \left(-C u  /\|\mathbf{X}\|_{\psi_{2}}^{2}\right)  d u=\frac{2\|\mathbf{X}\|_{\psi_{2}}^{2}}{C}.
\end{eqnarray*}
Thus,  we   obtain
$|S_{1}|\leq  {2\pi\|\mathbf{X}\|_{\psi_{2}}^{2}}/{c}<\infty,$       implying
\begin{equation}\label{s1nk2}
	\operatorname{pr}\Big(|\frac{ 3n-2}{n^2} S_{1} | \geq t \Big)=0,
\end{equation}
for  sufficiently large $n$, i.e.,   $n \geq 6\pi \|\mathbf{X}\|_{\psi_{2}}^{2}/(Ct)$.
%This, together with \eqref{s1nk1}, yields that
% \begin{eqnarray}\label{s1nk2}
	% \operatorname{pr}\left( |S_{1 n}-S_{1} | \geq  \frac{t}{6} \right)
	%    \leq   \operatorname{pr}\left(  | {T}_{1 n}-S_{1} | \geq  \frac{t}{12} \right) .
	% \end{eqnarray}

To establish a concentration inequality for $T_{1n}$, we employ the decoupling technique from \citet[Section~5.1.6]{By1980Approximation}. Specifically, define
\begin{eqnarray*}
	&&\Omega\left(\mathbf{X}_{1}, \mathbf{Y}_{1} ; \cdots ; \mathbf{X}_{n}, \mathbf{Y}_{n}\right)\\&=&\frac{1}{m} \sum_{r=0}^{m-1}k (\mathbf{X}_{1+2 r},\mathbf{Y}_{1+2 r}; \mathbf{X}_{2+2 r}, \mathbf{Y}_{2+2 r}; \mathbf{X}_{3+2 r}, \mathbf{Y}_{3+2 r} ),
\end{eqnarray*}
where $m=\lfloor n/3\rfloor.$
As shown in \citet[Section~5.1.6]{By1980Approximation}, we have
\begin{equation*}{T}_{1 n}= \frac{1}{n !} \sum_{n !} \Omega\left(\mathbf{X}_{i_{1}}, \mathbf{Y}_{i_{1}} ; \cdots ; \mathbf{X}_{i_{n}}, \mathbf{Y}_{i_{n}}\right),\end{equation*}
where $\sum_{n !}$ denote summation
over all $n !$ permutations $\left(i_{1}, \cdots, i_{n}\right)$ of $(1, \cdots, n).$
By   Markov's inequality, for any $\lambda>0,$ we have
\begin{equation*}
	\operatorname{pr}\left(    {T}_{1 n}-S_{1}   \geq  t\right) \leq \exp \{-\lambda t \} E\left\{\exp\left\{ \lambda[ {T}_{1 n}-S_{1}]\right\}\right\}.
\end{equation*}
Applying Jensen's inequality   yields
\begin{eqnarray*}
	&&E\left\{\exp\left\{ \lambda{T}_{1 n}\right\}\right\}\\
	 &=&E\Big[\exp  \{\frac{\lambda}{n !} \sum_{n !} \Omega (\mathbf{X}_{i_{1}}, \mathbf{Y}_{i_{1}} ; \cdots ; \mathbf{X}_{i_{n}}, \mathbf{Y}_{i_{n}} ) \}\Big] \nonumber \\
	&\leq& \frac{1}{n !} \sum_{n !} E \Big[\exp\{\lambda \Omega (\mathbf{X}_{i_{1}}, \mathbf{Y}_{i_{1}} ; \cdots ; \mathbf{X}_{i_{n}}, \mathbf{Y}_{i_{n}} )\}\Big ] \nonumber\\
	&=& E\Big[\exp \Big\{ \frac{\lambda}{m} \sum_{r=0}^{m-1}k (\mathbf{X}_{1+2 r},\mathbf{Y}_{1+2 r}; \mathbf{X}_{2+2 r}, \mathbf{Y}_{2+2 r};\\
	&&\mathbf{X}_{3+2 r}, \mathbf{Y}_{3+2 r} )\Big\}\Big]
	\nonumber\\
	&=&E^{m} \{\exp\{\lambda k(\mathbf{X}_{1},\mathbf{Y}_{1}; \mathbf{X}_{2}, \mathbf{Y}_{2};\mathbf{X}_{3}, \mathbf{Y}_{3} )/m \} \}.
\end{eqnarray*}
This leads to
\begin{eqnarray}\label{expo-markov-generating memont}
	&&\operatorname{pr}\left(    {T}_{1 n}-S_{1}   \geq  t \right)\nonumber \\
	&\leq& \exp \{-\lambda t \}
	E^{m} \{\exp  \{\lambda [k(\mathbf{X}_{1},\mathbf{Y}_{1}; \mathbf{X}_{2}, \mathbf{Y}_{2};\mathbf{X}_{3}, \mathbf{Y}_{3} )\nonumber\\
	&&-S_1]/m \} \}.
\end{eqnarray}

We now show that $k(\mathbf{X}{1},\mathbf{Y}{1}; \mathbf{X}{2}, \mathbf{Y}{2};\mathbf{X}{3}, \mathbf{Y}{3})$ is   sub-exponential. Using the inequality $|a+b+c|^p \leq 3^{p-1}(|a|^p+|b|^p+|c|^p)$ for $p \geq 1$, we obtain
 \begin{eqnarray*}
	 E\{|k(\mathbf{X}_{1},\mathbf{Y}_{1}; \mathbf{X}_{2}, \mathbf{Y}_{2};\mathbf{X}_{3}, \mathbf{Y}_{3} )|^p \}  &\leq&  \pi^p    E^2\{| \boldsymbol{v}^T\mathbf{X}  |^p \},
\end{eqnarray*}
which implies
\begin{eqnarray*}
&&p^{-1} \{   E|k(\mathbf{X}_{1},\mathbf{Y}_{1}; \mathbf{X}_{2}, \mathbf{Y}_{2};\mathbf{X}_{3}, \mathbf{Y}_{3} ) |^p \}^{1/p} \\
&  \leq& \pi \Big[ { E^{1/p}\{| \boldsymbol{v}^T\mathbf{X}  |^p  \}}/{\sqrt{p}} \Big ]^2.
\end{eqnarray*}
By   the  definitions of   $\|\cdot\|_{\psi_{1}}$  and  $\|\cdot\|_{\psi_{2}}$, we have
$$
\left\| k(\mathbf{X}_{1},\mathbf{Y}_{1}; \mathbf{X}_{2}, \mathbf{Y}_{2};\mathbf{X}_{3}, \mathbf{Y}_{3} )\right\|_{\psi_{1}}   \leq  \pi   \|\boldsymbol{v}^T\mathbf{X} \|_{\psi_2}  ^2.
$$
Thus,         $k(\mathbf{X}_{1},\mathbf{Y}_{1}; \mathbf{X}_{2}, \mathbf{Y}_{2};\mathbf{X}_{3}, \mathbf{Y}_{3} )$ is  sub-exponential.
Moreover, by \eqref{S1boundness},  we have
\begin{eqnarray}\label{kernel-subexp-norm}
	&&\left\|k(\mathbf{X}_{1},\mathbf{Y}_{1}; \mathbf{X}_{2}, \mathbf{Y}_{2};\mathbf{X}_{3}, \mathbf{Y}_{3} )-S_{1}\right\|_{\psi_{1}}\nonumber\\
	&\leq& \left\|k(\mathbf{X}_{1},\mathbf{Y}_{1}; \mathbf{X}_{2}, \mathbf{Y}_{2};\mathbf{X}_{3}, \mathbf{Y}_{3} )\right\|_{\psi_{1}}+ |S_{1} |\nonumber\\
	&\leq&  \pi \sup _{\boldsymbol{v} \in \mathbb{S}^{p-1}}\|\boldsymbol{v}^T\mathbf{X}\|_{\psi_{2}}^2 + \pi \boldsymbol{v}^TE\{ \mathbf{X}  \mathbf{X}^T\}\boldsymbol{v} \nonumber\\
	&\leq&  2\pi \|\mathbf{\Sigma}\|.
\end{eqnarray}
%For the last inequality above, we used Lemma  A.3 of  \citet{bunea2015sample}.

Let $K=\|\mathbf{\Sigma}\|$ for brevity. By property (e) in Proposition 2.7.1 of \citet{vershynin2018high} and \eqref{kernel-subexp-norm}, we  can  bound  \eqref{expo-markov-generating memont}.  Specifically, if $\lambda$ satisfies
\begin{equation}\label{constraint:MGF}
	\Big|\frac{\lambda}{m}\Big| \leq \frac{c}{K},
\end{equation}
we have
\begin{eqnarray*}
	&& E  \{\exp   (m^{-1} \lambda [k(\mathbf{X}_{1},\mathbf{Y}_{1}; \mathbf{X}_{2}, \mathbf{Y}_{2};\mathbf{X}_{3}, \mathbf{Y}_{3} )-S_1]     )  \} \\
	& \leq&  E  \{\exp \Big ( m^{-2}\lambda^2  \|k(\mathbf{X}_{1},\mathbf{Y}_{1}; \mathbf{X}_{2}, \mathbf{Y}_{2};\mathbf{X}_{3}, \mathbf{Y}_{3} )-S_{1} \|_{\psi_{1}}    )  \}  \\
	& \leq& \exp [ C \lambda^2  K^2 /m^2 ].
\end{eqnarray*}
This, together with  \eqref{expo-markov-generating memont}, yields
\begin{equation}\label{tail-bound-optimization}
	\operatorname{pr}\left( {T}_{1 n}-S_{1}   \geq  t  \right) \leq \exp \Big[-\lambda t  +C \lambda^2 K^2 /m \Big].
\end{equation}

To obtain a sharp bound for \eqref{tail-bound-optimization}, we minimize the right-hand side over $\lambda > 0$ subject to constraint \eqref{constraint:MGF}. The optimal choice is
$$\lambda_{\mathrm{opt}}=\min \Big(\frac{m t}{ 2 C K^{2}}, \frac{m c}{K} \Big).$$
Then, we have
\begin{eqnarray*}
	\operatorname{pr}\Big({T}_{1 n}\!-\!S_{1}  \geq  t \Big)\leq  \exp \Big[- \frac{mt^{2}}{4 C  K^{2} }\Big],
\end{eqnarray*}
if $ \frac{m t}{ 2 C K^{2}}\leq  \frac{m c}{K};$
and
\begin{eqnarray*}
	\operatorname{pr}\Big({T}_{1 n}\!-\!S_{1}  \geq  t \Big)&\leq& \exp \Big[-t\lambda_{\mathrm{opt}} +C \lambda_{\mathrm{opt}}\frac{m t}{ 2 C K^{2}} \frac{K^2}{m} \Big] \\
	&\leq &  \exp \Big[-\frac{m c t}{2K } \Big],
\end{eqnarray*} if $\frac{m t}{ 2 C K^{2}}>  \frac{m c}{K}$.
Combining  both inequalities, we have
\begin{eqnarray*}
	\operatorname{pr}\Big({T}_{1 n}\!-\!S_{1}  \geq  t \Big)
	&\leq&  \exp \Big[-m\min \Big(\frac{t^{2}}{4 C  K^{2} }, \frac{c t}{2K }\Big)\Big]\\
	&\leq& \exp \Big[-C_1n\min \Big(\frac{t^{2}}{ K^{2}}, \frac{ t}{K}\Big)\Big].
\end{eqnarray*}

Applying the same argument to $-[T_{1n} - S_{1}]$ gives an identical bound for $\operatorname{pr}( -[T_{1n} - S_{1}] \geq t )$. Therefore,
\begin{eqnarray*}
	 \operatorname{pr}\Big(|{T}_{1 n}\!-\!S_{1}|  \geq t \Big)
	 \leq  2 \exp \Big[-C_1n\min \Big(\frac{t^{2}}{ K^{2}}, \frac{ t}{K}\Big)\Big].
\end{eqnarray*}
This, together with  \eqref{s1nk1} and \eqref{s1nk2}, we obtain
\begin{eqnarray}\label{term1}
	 \operatorname{pr}\left( |S_{1 n}-S_{1} | \geq t \right)
	\leq  2 \exp \Big[- Cn\min \Big(\frac{t^{2}}{ K^{2}}, \frac{ t}{K}\Big)\Big].\nonumber\\
\end{eqnarray}

\emph{ \textbf{Step 2: The tail bound on  $S_{2n}$.}}
We begin with the decomposition
$$S_{2n}=S_{2n,1}S_{2n,2},$$
where
\begin{eqnarray*}
&&S_{2n,1}= \frac{1}{n^3}\sum_{i,r,l=1}^n \mathrm{ang}(\mathbf{Y}_i  -\mathbf{Y}_r, \mathbf{Y}_l-\mathbf{Y}_r),\\
 && S_{2n,2}=\frac{1}{n^2}\sum_{j, m=1}^n \{\boldsymbol{v}^{\mathrm{T}}\mathbf{X}_j\} \{ \mathbf{X}_m^T\boldsymbol{v} \}.
\end{eqnarray*}
Let
$ S_{2,1}=E \{\mathrm{ang}(\mathbf{Y}_1  -\mathbf{Y}_3, \mathbf{Y}_4-\mathbf{Y}_3) \}.$  Then,   we have
\begin{eqnarray}\label{S2n-S2}
	S_{2n} =[S_{2n,1}-S_{2,1} ]S_{2n,2} +S_{2,1} S_{2n,2}.
\end{eqnarray}

We can express   $S_{2n,1}$   as
\begin{eqnarray*}
	S_{2n,1}&=& \frac{1}{n^3}\sum_{i,r,l=1}^n \mathrm{ang}(\mathbf{Y}_i  -\mathbf{Y}_r, \mathbf{Y}_l-\mathbf{Y}_r)\\
	& =&\frac{(n-1)(n-2)} {n^2} {T}_{2n,1},\nonumber
\end{eqnarray*}
where  ${T}_{2n,1}$  is  a U-statistic, given by
\begin{eqnarray*}
	{T}_{2n,1}&=&   \begin{pmatrix} n \\  3 \end{pmatrix}^{  -1}   \sum_{i<r<l}   \frac{ 1}{3} \Big\{  \mathrm{ang}(\mathbf{Y}_i  -\mathbf{Y}_r, \mathbf{Y}_l-\mathbf{Y}_r)\\
	&& + \mathrm{ang}(\mathbf{Y}_r  -\mathbf{Y}_i, \mathbf{Y}_l-\mathbf{Y}_i)
	\\&&+    \mathrm{ang}(\mathbf{Y}_i  -\mathbf{Y}_l, \mathbf{Y}_r-\mathbf{Y}_l)\Big \}\\
	&=& \begin{pmatrix} n \\  3 \end{pmatrix}^{-1} \sum_{i<r<l}\tilde{ k} (\mathbf{X}_{i},\mathbf{Y}_{i}; \mathbf{X}_{r}, \mathbf{Y}_{r};\mathbf{X}_{l}, \mathbf{Y}_{l} ).
\end{eqnarray*}
Note that   $0 \leq |\tilde{ k}(\mathbf{X}_{i},\mathbf{Y}_{i}; \mathbf{X}_{r}, \mathbf{Y}_{r};\mathbf{X}_{l}, \mathbf{Y}_{l} )|\leq \pi $.  By Theorem 5.6.1.A of \citet{By1980Approximation}, we   obtain
\begin{eqnarray*}
	\operatorname{pr}\left( |T_{2n,1}-E\{T_{2n,1}\} | \geq t  \right)
	\leq  2 \exp\Big[-2\lfloor n/3\rfloor t^2/ \pi^2 \Big].
\end{eqnarray*}
Since $E\{T_{2n,1}\}=S_{2,1}$ and    $|S_{2,1}|\leq \pi,$ we have
\begin{eqnarray}\label{S12nk-2k}
	\operatorname{pr}\left( |S_{2n,1}-S_{2,1} | \geq  t \right)
	& \leq & \operatorname{pr}\left(  | T_{2n,1}-S_{2,1} | \geq  \frac{t}{2} \right)\nonumber\\
	&&+\operatorname{pr}\left(  |\frac{ 3n-2}{n^2} S_{2,1} | \geq \frac{t}{2} \right)\nonumber\\
	& \leq &  2 \exp\Big[- \frac{n t^2}{6\pi^2}\Big],
\end{eqnarray}
for sufficiently large $n.$

Note that   $E\{\mathbf{X}_i\} = 0$  and $\boldsymbol{v}^{\mathrm{T}}\mathbf{X}$ is sub-gaussian by Condition 1. By    Hoeffding's inequality in  Theorem 2.6.3  in  \citet{vershynin2018high},    we obtain
\begin{eqnarray}\label{S22nk-2k-barX}
	\operatorname{pr} \Big(  | \boldsymbol{v}^T\overline{\mathbf{X}}_{n} |  \geq   t \Big)
	&\leq&  2 \exp \Big[-\frac{C n t^{2}}{ \|\boldsymbol{v}^T\mathbf{X} \|_{\psi_{2}}^{2} } \Big]\nonumber\\
	&\leq&  2 \exp \Big[- \frac{C n t^{2}}{ \|\mathbf{X}\|_{\psi_{2}}^{2}} \Big].
\end{eqnarray}
Thus, we have
\begin{eqnarray}\label{S22nk-2k}
	\operatorname{pr} \Big(  |S_{2n,2} | \geq   t \Big)&\leq&
	\operatorname{pr} \Big(  | \boldsymbol{v}^T\overline{\mathbf{X}}_{n} |  \geq   \sqrt{t} \Big)\nonumber\\
	&\leq&  2 \exp \Big[- \frac{C n t }{ \|\mathbf{X}\|_{\psi_{2}}^{2}} \Big]\nonumber\\
	&\leq&    2 \exp \Big[- \frac{C n t }{ K } \Big].
\end{eqnarray}
The last inequality follows from  $\|\mathbf{X}\|_{\psi_{2}}^{2} \leq C(\mathrm{tr}(\mathbf{\Sigma}) + \|\mathbf{\Sigma}\|)\leq CK.$

From\eqref{S12nk-2k}-\eqref{S22nk-2k}, we obtain
\begin{eqnarray*}
	&&\operatorname{pr}  ( |S_{2n,1}-S_{2,1}||S_{2n,2}  | \! \geq  t  )\\
	& \leq &    \operatorname{pr}  ( | S_{2n,1}-S_{2,1} | \! \geq   \sqrt{t}   ) \!+\!
	\operatorname{pr}  ( | S_{2n,2}  | \! \geq  \sqrt{t}~  )  \\
	& \leq&2 \exp\Big[- \frac{n t }{6\pi^2}\Big] +2 \exp \Big[-\frac{C n \sqrt{t} }{ K } \Big].
\end{eqnarray*}
Moreover, using $|S_{2,1}|\leq \pi$,      we have
\begin{eqnarray*}
	\operatorname{pr} \Big( | S_{2,1}||S_{2n,2}  |  \geq  t  \Big)   & \leq&    \operatorname{pr} \Big(  |S_{2n,2}  | \!\geq \!\frac{t}{\pi } \Big)  \leq     2 \exp \Big[- \frac{C n t }{ K } \Big].
\end{eqnarray*}
  Note that   $t^2\leq t \leq t^{1/2}$, $t\in[0,1]$. By    \eqref{S2n-S2}, we have
\begin{eqnarray}\label{term2}
	\operatorname{pr} \Big( | S_{2n } | \geq t \Big)  & \leq &    \operatorname{pr} \Big( |S_{2n,1}-S_{2,1}||S_{2n,2} |   \geq \frac{t}{2}\Big)\nonumber\\
	&&+ \operatorname{pr} \Big( | S_{2,1}||S_{2n,2}  | \geq  \frac{t}{2}   \Big)\nonumber\\
	  & \leq&   6 \exp \Big[- \frac{C n t }{ K } \Big].
\end{eqnarray}

\emph{ \textbf{Step 3: The tail bound on  $S_{3n}$.}} We employ the decomposition
$$S_{3n}=S_{3n,1}S_{3n,2},$$
where
\begin{align*}
S_{3n,1} &= \frac{1}{n^3} \sum_{i, r,l=1}^n \mathrm{ang}(\mathbf{Y}_i - \mathbf{Y}_r, \mathbf{Y}_l - \mathbf{Y}_r) {\boldsymbol{v}^T\mathbf{X}_l}, \\
S_{3n,2} &= \frac{1}{n} \sum_{j=1}^n \boldsymbol{v}^{\mathrm{T}}\mathbf{X}_j.
\end{align*}
Define $ S_{3,1}=E \{  \{ \mathbf{X}_4^T\boldsymbol{v}\} \mathrm{ang}(\mathbf{Y}_1  -\mathbf{Y}_3, \mathbf{Y}_4-\mathbf{Y}_3) \}.$

Note that $ S_{3n,1 }$ can be 	expressed as
\begin{eqnarray*}
	S_{3n,1 }  =\frac{(n-1)(n-2)} {n^2}  {T}_{3n,1 },\nonumber
\end{eqnarray*}
where    ${T}_{3n,1}$  is  a U-statistic, given by
\begin{eqnarray*}
	{T}_{3n,1 }&=&\begin{pmatrix} n \\  3 \end{pmatrix}^{-1} \sum_{i<r<l}\widetilde{\tilde{ k}} (\mathbf{X}_{i},\mathbf{Y}_{i}; \mathbf{X}_{r}, \mathbf{Y}_{r};\mathbf{X}_{l}, \mathbf{Y}_{l} ),
\end{eqnarray*}
with $ \widetilde{\tilde{ k}} (\mathbf{X}_{i},\mathbf{Y}_{i}; \mathbf{X}_{r}, \mathbf{Y}_{r};\mathbf{X}_{l}, \mathbf{Y}_{l} )=  \frac{ 1}{6} [ \{\boldsymbol{v}^T\mathbf{X}_i+\boldsymbol{v}^T\mathbf{X}_l\}  \mathrm{ang}(\mathbf{Y}_i  -\mathbf{Y}_r, \mathbf{Y}_l-\mathbf{Y}_r) + \{\boldsymbol{v}^T\mathbf{X}_r+\boldsymbol{v}^T\mathbf{X}_l\} \mathrm{ang}(\mathbf{Y}_r  -\mathbf{Y}_i, \mathbf{Y}_l-\mathbf{Y}_i)   + \{\boldsymbol{v}^T\mathbf{X}_i+\boldsymbol{v}^T\mathbf{X}_r\}   \mathrm{ang}(\mathbf{Y}_i  -\mathbf{Y}_l, \mathbf{Y}_r-\mathbf{Y}_l) ].  $

Under Condition 1 and following arguments analogous to Step 1, we establish
\begin{eqnarray*}
	\operatorname{pr}\left( |S_{3n,1}-S_{3,1} | \geq  t \right)
	\leq 2 \exp \Big[-Cn\min \Big(\frac{t^{2}}{ K^{2}}, \frac{ t}{K}\Big)\Big].
\end{eqnarray*}
By \eqref{S22nk-2k-barX}, we have
\begin{eqnarray*}
	\operatorname{pr}\left( |S_{3n,2} | \geq  t  \right)  \leq
	2 \exp \Big[- \frac{Cn t^{2}}{ K } \Big].
\end{eqnarray*}
Furthermore,
\begin{eqnarray*}
	|S_{3,1}|&=&|E \{  \{ \mathbf{X}_4^T\boldsymbol{v}\} \mathrm{ang}(\mathbf{Y}_1  -\mathbf{Y}_3, \mathbf{Y}_4-\mathbf{Y}_3) \}|\\&\leq&\pi E \{ |\mathbf{X}_4^T\boldsymbol{v}|\}\leq\pi \sqrt{\frac{2}{C}}\|\mathbf{X}\|_{\psi_{2}} .
\end{eqnarray*}
Thus, we have
\begin{eqnarray}\label{term3}
	&& \operatorname{pr} \Big( |S_{3n } | \geq t \Big)\nonumber\\  & \leq &    \operatorname{pr} \Big( |S_{3n,1}-S_{3,1}||S_{3n,2} |   \geq \frac{t}{2}\Big)\nonumber\\
	&&+ \operatorname{pr} \Big( | S_{3,1}||S_{3n,2}  | \geq  \frac{t}{2}   \Big)\nonumber\\
	& \leq &    \operatorname{pr} \Big ( | S_{3n,1}-S_{3,1} | \! \geq   \sqrt{\frac{t}{2}} \Big) \!+\!
	\operatorname{pr}  \Big( | S_{3n,2}  | \! \geq  \sqrt{\frac{t}{2} }  \Big) \nonumber\\
	&&+
	\operatorname{pr} \Big ( | S_{3n,2}  | \! \geq  \frac{t}{C\|\mathbf{X}\|_{\psi_{2}} }  \Big)\nonumber\\
	\!\!\!  & \leq&  6 \exp \Big[-Cn\min \Big(\frac{t^{2}}{ K^{2}}, \frac{ t}{K}\Big)\Big].
\end{eqnarray}
Combining \eqref{term1}, \eqref{term2}, \eqref{term3} with \eqref{widehat.omega-omega}, we have
\begin{eqnarray*}
	 \operatorname{pr} ( |\boldsymbol{v}^T\{\widehat{ \boldsymbol{ \Lambda}}_\mathrm{PR} -\boldsymbol{ \Lambda}_\mathrm{PR}\} \boldsymbol{v}| \geq t )
 	\leq  c \exp \{-Cn\min  (\frac{t^{2}}{ K^{2}}, \frac{ t}{K} )\},
\end{eqnarray*}
with $K=\|\mathbf{\Sigma}\|$.
\end{proof}

\begin{proof}[\textbf{Proof of Theorem  \ref{them-projection-consistency}}]
By    Theorem 4.7.1    in  \citet{vershynin2018high}, we   obtain
\begin{equation}\label{Theorem 4.7.1}
	{E}\{\|\widehat{\mathbf{\Sigma}}_{n}-\mathbf{\Sigma} \|\} \leq C K^{2}\left(\sqrt{\frac{p}{n}}+\frac{p}{n}\right)\|\mathbf{\Sigma}\|.
\end{equation}
Note that
\begin{eqnarray*}
	&&\widehat{\mathbf{\Sigma}}_{n}^{-1} \widehat{\boldsymbol{ \Lambda}}_\mathrm{PR}-\mathbf{\Sigma}^{-1} \boldsymbol{ \Lambda}_\mathrm{PR}\\
	&=&\mathbf{\Sigma}^{-1} (\mathbf{\Sigma}-\widehat{\mathbf{\Sigma}}_{n} ) \widehat{\mathbf{\Sigma}}_{n}^{-1} (\widehat{ \boldsymbol{ \Lambda}}_\mathrm{PR}- { \boldsymbol{ \Lambda}}_\mathrm{PR} )\\&&+\mathbf{\Sigma}^{-1} (\mathbf{\Sigma}-\widehat{\mathbf{\Sigma}}_{n} ) \widehat{\mathbf{\Sigma}}_{n}^{-1}  { \boldsymbol{ \Lambda}}_\mathrm{PR}+\mathbf{\Sigma}^{-1} (\widehat{ \boldsymbol{ \Lambda}}_\mathrm{PR}- { \boldsymbol{ \Lambda}}_\mathrm{PR}).
\end{eqnarray*}
This, together with   Theorem  \ref{them-kenerl-consistency}  and  \eqref{Theorem 4.7.1}, yields that
\begin{eqnarray*}
	\| \widehat{\mathbf{\Sigma}}_n^{-1}\widehat{ \boldsymbol{ \Lambda}}_\mathrm{PR}-  {\mathbf{\Sigma}}^{-1} { \boldsymbol{ \Lambda}}_\mathrm{PR}\|=O_{p}\left( \sqrt{ \frac{p }{n}}\right),
\end{eqnarray*}
under Condition 2.
\end{proof}

\begin{proof}[\textbf{Proof of Theorem  \ref{Bochner-induced}}]
(i) By    Lemma  \ref{Bochner}, we have
\begin{eqnarray*}
	&&\mathrm{ICMI}_\mathrm{\mu}(X|\mathbf{Y})\nonumber\\&=&  \int_{\mathbb{R}^q} |\mathrm{E}\{{X}\exp(i \textbf{u}^T\textbf{Y} )\}- \mathrm{E}\{ X \}\mathrm{E}\{\exp(i \textbf{u}^T \textbf{Y})\}  |^2   d \mu(\textbf{u})\nonumber\\
	&=&\int_{\mathbb{R}^q}\mathrm{E}\{ ( {X}_1-E\{ {X}_1\}) ({X}_2-E\{ {X}_2\}) \nonumber\\&&\times \exp (i \textbf{u}^T(\textbf{Y}_1-\textbf{Y}_2))\}d \mu(\textbf{u})\nonumber\\
	&=&\mathrm{E}\Big\{ ( {X}_1-E\{ {X}_1\}) ({X}_2-E\{ {X}_2\})\nonumber \\
	&&\times \int_{\mathbb{R}^q}\exp (i \textbf{u}^T(\textbf{Y}_1-\textbf{Y}_2))d \mu(\textbf{u})\Big\}
	\\
	&=&\mathrm{E}\{ ( {X}_1-E\{ {X}_1\}) ({X}_2-E\{ {X}_2\})K(\mathbf{Y}_1-\mathbf{Y}_2)  \}.\nonumber
\end{eqnarray*}

(ii)  From the first equality above,   $\mathrm{ICMI}_\mathrm{\mu}(X|\mathbf{Y})=0$ holds if and only if   $$\mathrm{E}\{{X}\exp(i \textbf{u}^T\textbf{Y} )\}=\mathrm{E}\{ X \}\mathrm{E}\{\exp(i \textbf{u}^T\textbf{Y})\},$$
for all $\mathbf{u} \in \mathbb{R}^q$, which  is equivalent   to
$$\mathrm{E} \{[ {E}\{X | \mathbf{Y}\}-\mathbb{E}\{X\}]\exp(i \textbf{u}^T\textbf{Y} )  \}=0.$$
By the uniqueness of Fourier transforms, it follows that
$\mathbb{E}\{X | \mathbf{Y}\}=\mathbb{E}\{X\},$ a.s.
\end{proof}

\begin{proof}[\textbf{Proof of Theorem  \ref{them-kenel}}]
 Under  the linearity condition that $E\{\mathbf{X} |\mathbf{\Gamma}^T \mathbf{X}\}=\mathbf{P}_{\mathbf{\Gamma}}^T(\mathbf{\Sigma}) \mathbf{X}$,
  and following arguments analogous to that in Theorem~\ref{them-projection}, we obtain
\begin{eqnarray*}
	{ \boldsymbol{ \Lambda}}_\mathrm{K} & =& \mathrm{E}\{ \mathrm{E}\{\mathbf{X}_1-E\{ \mathbf{X}_1\}|\mathbf{Y}_1\} \mathrm{E}\{\mathbf{X}_2-E\{ \mathbf{X}_2\}|\mathbf{Y}_2\}^T\nonumber \\&&\times K(\mathbf{Y}_1-\mathbf{Y}_2)  \}\nonumber\\
	& =&\mathrm{E}\{\mathbf{P}_{\mathbf{\Gamma}}^T(\mathbf{\Sigma}) \mathrm{E}\{\mathbf{X}_1-E\{ \mathbf{X}_1\}|\mathbf{Y}_1\}\\
	&&\times \mathrm{E}\{\mathbf{X}_2-E\{ \mathbf{X}_2\}|\mathbf{Y}_2\}^T\mathbf{P}_{\mathbf{\Gamma}}(\mathbf{\Sigma})  K(\mathbf{Y}_1-\mathbf{Y}_2)  \} \nonumber\\
	& =&\mathbf{P}_{\mathbf{\Gamma} }^T(\mathbf{\Sigma}) { \boldsymbol{ \Lambda}}_\mathrm{K} \mathbf{P}_{\mathbf{\Gamma} }(\mathbf{\Sigma}). \nonumber
\end{eqnarray*}
Thus,  $\mathcal{S}( { \boldsymbol{ \Lambda}}_\mathrm{K}) \subseteq \mathbf{\Sigma} \mathcal{S}_{\mathbf{Y} | \mathbf{X}}$.
\end{proof}

\begin{proof}[\textbf{Proof of Theorem  \ref{them- Bochner-consistency}}]

By the Cauchy-Schwarz inequality, for any $\mathbf{y}_1, \mathbf{y}_2 \in \mathbb{R}^q$,
  it follows that
$$
K^{2}(\mathbf{y}_1, \mathbf{y}_2)=\langle\boldsymbol{\phi}(\mathbf{y}_1), \boldsymbol{\phi}(\mathbf{y}_2)\rangle_\mathcal{H}^2 \leq K(\mathbf{y}_1, \mathbf{y}_1) K(\mathbf{y}_2, \mathbf{y}_2).
$$
Since $K(\cdot, \cdot)$ is translation-invariant, we have
\begin{eqnarray*}\label{boundness-kerenl}
	0< K(\mathbf{y}_1, \mathbf{y}_2)  &=& K(\mathbf{y}_1-\mathbf{y}_2)
	   \leq  K(\mathbf{y}_1, \mathbf{y}_1)= K(0),
\end{eqnarray*}
which implies that $K(\cdot, \cdot)$ is bounded.

By the boundedness of $K(\cdot, \cdot)$  and arguments similar to those in Theorems~\ref{them-kenerl-consistency}-\ref{them-projection-consistency}, Theorem~\ref{them- Bochner-consistency} follows.
 \end{proof}

\begin{proof}[\textbf{Proof of Theorem  \ref{them-kenel-HD}}]
Note that
\begin{eqnarray*}
	{\boldsymbol{ \Lambda}}_\mathrm{V.PR}&=& -\mathrm{E}\{ [\mathbf{Z}_1\mathbf{Z}_1^T-\boldsymbol{\Sigma}][\mathbf{Z}_2\mathbf{Z}_2^T-
	\boldsymbol{\Sigma}]\\
	&&\times  \mathrm{ang}(\mathbf{Y}_1- \mathbf{Y}_3,\mathbf{Y}_2- \mathbf{Y}_3)  \}\\
	& =& \mathrm{E} \{\mathrm{E}\{[\mathbf{Z}_1\mathbf{Z}_1^T-\boldsymbol{\Sigma}]|\mathbf{Y}_1\} \mathrm{E}\{[\mathbf{Z}_2\mathbf{Z}_2^T-\boldsymbol{\Sigma}]|\mathbf{Y}_2\}^T\\ &&\times\mathrm{ang}(\mathbf{Y}_1- \mathbf{Y}_3,\mathbf{Y}_2- \mathbf{Y}_3) \}.
\end{eqnarray*}

By   the linearity condition
$E\{\mathbf{X} | \mathbf{\Gamma}^T \mathbf{X}\}=\mathbf{P}_{\mathbf{\Gamma}}^T(\mathbf{\Sigma}) \mathbf{X}$,
and the constant variance condition
$\operatorname{Var}(\mathbf{X}| \mathbf{\Gamma}^T \mathbf{X})
=\mathbf{\Sigma}-\mathbf{\Sigma}\mathbf{P}_{\mathbf{\Gamma}}(\mathbf{\Sigma})$,
we derive
\begin{eqnarray*}
	&&\mathrm{E}\{[\mathbf{Z}_i\mathbf{Z}_i^T-\boldsymbol{\Sigma}]|\mathbf{Y}_i\}\\
	&=&\mathrm{E} \{\mathrm{E}[(\mathbf{Z}_i\mathbf{Z}_i^T-\boldsymbol{\Sigma})|\mathbf{Y}_i, \mathbf{\Gamma}^T\mathbf{X}_i ] |\mathbf{Y}_i \}\\
	&=&\mathrm{E} \{\mathrm{E}[(\mathbf{Z}_i\mathbf{Z}_i^T-\boldsymbol{\Sigma})|  \mathbf{\Gamma}^T\mathbf{X}_i] |\mathbf{Y}_i \}\\
	&=&\mathrm{E} \Big\{
	\mathrm{E}[ \mathbf{X}_i\mathbf{X}_i^T |  \mathbf{\Gamma}^T\mathbf{X}_i]
	-\mathrm{E}[\mathbf{X}_i]\mathrm{E}[\mathbf{X}_i^T |\mathbf{\Gamma}^T\mathbf{X}_i]\\
	&&-\mathrm{E}[\mathbf{X}_i |\mathbf{\Gamma}^T\mathbf{X}_i]\mathrm{E}[\mathbf{X}_i^T]
	|\mathbf{Y}_i \Big\} +\mathrm{E}[\mathbf{X}_i]\mathrm{E}[\mathbf{X}_i^T]-\boldsymbol{\Sigma}\\
	&=&\mathrm{E} \Big\{\operatorname{Var}(\mathbf{X}_i| \mathbf{\Gamma}^T \mathbf{X}_i)+
	\mathrm{E}[\mathbf{X}_i |\mathbf{\Gamma}^T\mathbf{X}_i]
	\mathrm{E}[\mathbf{X}_i^T |\mathbf{\Gamma}^T\mathbf{X}_i]
	\\
	&&-\mathrm{E}[\mathbf{X}_i]\mathrm{E}[\mathbf{X}_i^T |\mathbf{\Gamma}^T\mathbf{X}_i]-\mathrm{E}[\mathbf{X}_i |\mathbf{\Gamma}^T\mathbf{X}_i]\mathrm{E}[\mathbf{X}_i^T]
	|\mathbf{Y}_i \Big\} \\
	&&+\mathrm{E}\{\mathbf{X}_i\}\mathrm{E}\{\mathbf{X}_i^T\}-\boldsymbol{\Sigma}
	\\
	&=&\mathrm{E} \Big\{\mathbf{\Sigma}-\mathbf{\Sigma}\mathbf{P}_{\Gamma}+
	\mathbf{P}_{\Gamma}^{\top} \mathbf{X}_i
	\mathbf{X}_i^{\top}\mathbf{P}_{\Gamma}-\mathrm{E}[\mathbf{X}_i]\mathbf{X}_i^T\mathbf{P}_{\Gamma}\\
	&&	-\mathbf{P}_{\Gamma}^T \mathbf{X}_i\mathrm{E}[\mathbf{X}_i^T]
	) \Big|\mathbf{Y}_i \Big\} +\mathrm{E}[\mathbf{X}_i]\mathrm{E}[\mathbf{X}_i^T]-\boldsymbol{\Sigma}
	\\
	&=&\mathbf{P}_{\Gamma}^T\mathrm{E}\{
	\mathbf{X}_i\mathbf{X}_i^T|\mathbf{Y}_i \}\mathbf{P}_{\Gamma}
	-\mathrm{E}\{\mathbf{X}_i\}\mathrm{E}\{\mathbf{X}_i^T|\mathbf{Y}_i\}\mathbf{P}_{\Gamma}\\
	&&-\mathbf{P}_{\Gamma}^T \mathrm{E}\{\mathbf{X}_i|\mathbf{Y}_i\}\mathrm{E}\{\mathbf{X}_i^T\}
	+\mathrm{E}\{\mathbf{X}_i\}\mathrm{E}\{\mathbf{X}_i^T\} -\mathbf{\Sigma}\mathbf{P}_{\Gamma}
	\\
	&=& \mathbf{P}_{\Gamma}^T\mathrm{E}\{
	\mathbf{X}_i\mathbf{X}_i^T|\mathbf{Y}_i  \}\mathbf{P}_{\Gamma}
	-\mathbf{P}_{\Gamma}^T\mathrm{E}[\mathbf{X}_i]
	\mathrm{E}\{\mathbf{X}_i^T|\mathbf{Y}_i \}\mathbf{P}_{\Gamma}\\
	&&-\mathbf{P}_{\Gamma}^T\mathrm{E}\{\mathbf{X}_i|\mathbf{Y}_i \}\mathrm{E}[\mathbf{X}_i^T]\mathbf{P}_{\Gamma}
	+\mathrm{E}\{\mathbf{X}_i\}\mathrm{E}\{\mathbf{X}_i^T\}-\mathbf{\Sigma}\mathbf{P}_{\Gamma}
	\\
	&=&\mathbf{P}_{\Gamma}^T
	\mathrm{E}\{[\mathbf{Z}_i\mathbf{Z}_i^T-\boldsymbol{\Sigma}]|\mathbf{Y}_i\}\mathbf{P}_{\Gamma}.
\end{eqnarray*}
Here, $\mathbf{P}_{\Gamma}=\mathbf{P}_{\Gamma}(\mathbf{\Sigma})$ for   brevity.
Then, we have
\begin{eqnarray*}
	{\boldsymbol{ \Lambda}}_\mathrm{V.PR}&=& -\mathrm{E}\{ [\mathbf{Z}_1\mathbf{Z}_1^T-\boldsymbol{\Sigma}][\mathbf{Z}_2\mathbf{Z}_2^T-
	\boldsymbol{\Sigma}]\nonumber\\&&\times  \mathrm{ang}(\mathbf{Y}_1- \mathbf{Y}_3,\mathbf{Y}_2- \mathbf{Y}_3)  \}\\
	& =& \mathrm{E} \{\mathrm{E}\{[\mathbf{Z}_1\mathbf{Z}_1^T-\boldsymbol{\Sigma}]|\mathbf{Y}_1\} \mathrm{E}\{[\mathbf{Z}_2\mathbf{Z}_2^T-\boldsymbol{\Sigma}]|\mathbf{Y}_2\}^T\nonumber\\ &&\times\mathrm{ang}(\mathbf{Y}_1- \mathbf{Y}_3,\mathbf{Y}_2- \mathbf{Y}_3) \}
	\nonumber\\
	& =&-\mathbf{P}_{\Gamma}^T(\mathbf{\Sigma})\mathrm{E}\{ [\mathbf{Z}_1\mathbf{Z}_1^T\boldsymbol{\Sigma}]\mathbf{P}_{\Gamma}(\mathbf{\Sigma})[\mathbf{Z}_2\mathbf{Z}_2^T-
	\boldsymbol{\Sigma}]\nonumber\\&&\times  \mathrm{ang}(\mathbf{Y}_1- \mathbf{Y}_3,\mathbf{Y}_2- \mathbf{Y}_3)  \} \mathbf{P}_{\Gamma}(\mathbf{\Sigma}). \nonumber
\end{eqnarray*}
Thus,  $\mathcal{S}({\boldsymbol{ \Lambda}}_\mathrm{V.PR}) \subseteq \mathbf{\Sigma} \mathcal{S}_{Y | \mathrm{X}}$.  Similarly, we obtain $\mathcal{S}({\boldsymbol{ \Lambda}}_\mathrm{V.K}) \subseteq \mathbf{\Sigma} \mathcal{S}_{Y | \mathrm{X}}$.
\end{proof}




\end{appendices}


\end{document}